\documentclass[12pt]{amsart}

\usepackage[hmargin=0.8in,height=8.6in]{geometry}
\usepackage{amssymb,amsthm, times, enumerate}
\usepackage{delarray,verbatim}

\usepackage{ifpdf}
\ifpdf
\usepackage[pdftex]{graphicx}
 \else
\usepackage[dvips]{graphicx}
 \fi

\usepackage{ifpdf}
\usepackage{color}
\definecolor{webgreen}{rgb}{0,.5,0}
\definecolor{webbrown}{rgb}{.8,0,0}
\definecolor{emphcolor}{rgb}{0.95,0.95,0.95}

\usepackage{hyperref}
\hypersetup{%
	colorlinks=true,
	linkcolor=webbrown,
	filecolor=webbrown,
	citecolor=webgreen,
	breaklinks=true}
\ifpdf \hypersetup{pdftex,
	pdfstartview=FitH, 
	bookmarksopen=true,
	bookmarksnumbered=true
} \else \hypersetup{dvips} \fi

\newcommand {\e}{\mathbb{E}}

\numberwithin{equation}{section}

\newtheorem{theorem}{Theorem}[section]
\newtheorem{proposition}{Proposition}[section]

\newtheorem{remark}{Remark}[section]
\newtheorem{lemma}{Lemma}[section]

\newtheorem{assumption}{Assumption}[section]

\numberwithin{remark}{section} \numberwithin{proposition}{section}
\numberwithin{corollary}{section}
\newcommand {\R}{\mathbb{R}}

\newcommand {\F}{\mathcal{F}}

\newcommand {\p}{\mathbb{P}}

\newcommand {\E}{\mathbb{E}}

\newcommand{\diff}{{\rm d}}

\newcommand{\lev}{L\'{e}vy }

\title{
The Leland-Toft optimal capital structure model under Poisson observations
}

\author[Z. Palmowski]{Zbigniew Palmowski$^*$\,}\thanks{$*$\,Faculty of Pure and Applied Mathematics,
Wroc\l aw University of Science and Technology,
Wyb. Wyspia\'nskiego 27, 50-370 Wroc\l aw, Poland. Email: \mbox{{\em
        zbigniew.palmowski@pwr.edu.pl}}}

\author[J.L. P\'erez]{Jos\'e Luis P\'erez$^{**}$\,}\thanks{$**$\,Department of Probability and Statistics, Centro de Investigaci\'on en Matem\'aticas, A.C. Calle Jalisco S/N
C.P. 36240, Guanajuato, Mexico. Email: \mbox{{\em
        jluis.garmendia@cimat.mx}}}

\author[B. A. Surya]{Budhi Arta Surya$^\dag$\,}\thanks{$\dag$\,School of Mathematics and Statistics, Victoria University of Wellington, Gate 6, Kelburn PDE, Wellington 6140, New Zealand. Email: \mbox{{\em
        budhi.surya@vuw.ac.nz}}}
\author[K. Yamazaki]{\,Kazutoshi Yamazaki$^\ddag$}\thanks{$\ddag$\, Department of Mathematics,
Faculty of Engineering Science, Kansai University, 3-3-35 Yamate-cho, Suita-shi, Osaka 564-8680, Japan. Email: \mbox{{\em
kyamazak@kansai-u.ac.jp}}}

\begin{document}
\maketitle
\begin{abstract} 
This paper revisits the optimal capital structure model with
endogenous bankruptcy, first studied by Leland  \cite{Leland94} and Leland and Toft \cite{Leland96}. Unlike in the standard case, where shareholders continuously observe the asset value and bankruptcy is executed instantaneously and without delay,
 the information of the asset value is assumed to be updated only at intervals,
 modeled by the jump times of an independent Poisson
process. Under the spectrally negative L\'evy model, we obtain the optimal bankruptcy strategy and the corresponding capital structure.   A series of numerical studies enable analysis of the sensitivity of observation frequency in the optimal solutions,  optimal leverage and credit spreads.

\end{abstract}
 \noindent \small{\textbf{Keywords:}\,  Credit risk, optimal capital structure, spectrally negative \lev processes, scale functions }\\
 \noindent \small{\textbf{JEL Classification:}\, D92, G32, G33}\\
\noindent \small{\textbf{Mathematics Subject Classification (2010):}\,60G40, 60G51, 91G40}


\section{Introduction}


The study of capital structures dates back to the seminal work by Modigliani and Miller \cite{Modigliani}, which shows that, \emph{in a frictionless economy,} the value of a firm is invariant to the choice of capital structures.  While the  Modigliani-Miller (MM) theory 
 is regarded as an effective starting point for research on capital structures and has provided valuable insights in the field, it is not directly applicable to businesses. In reality,  selection of capital structures is not perfectly random.
Instead, it depends significantly on factors such as industry type, county and  corporate law. In the field of corporate finance, various approaches have been taken to explain how
much debt a firm should issue. A reasonable conclusion can be obtained only after challenging some of the assumptions of the classical MM theory.

The \emph{trade-off theory} is one well-known approach for the study of capital structures. While various frictions may affect a firm's decisions,  (1) \emph{bankruptcy costs} and (2) \emph{tax benefits} are believed to be the most important factors.
By issuing debt, bankruptcy costs increase, while at the same time the firm can enjoy tax shields for coupon payments to the bondholders.  The trade-off theory states that firms issue the appropriate debt  to solve the trade-off between  minimizing bankruptcy costs and maximizing tax benefits.
To formulate this optimization problem,
one needs an efficient and realistic way of modeling not only bankruptcy but also tax benefits, which depend heavily on the dynamics of the firm's asset value. For more details on the trade-off theory and its review, see, e.g., \cite{Kraus, Frank, Ju}. 





Classically, there are two models of bankruptcy in credit risk:  the \emph{structural approach} and the \emph{reduced-form approach} (see \cite{Bielecki}). The former, first proposed by Black and Cox \cite{Black}, models bankruptcy time as the first time the asset value goes below a fixed barrier. The latter models it as the first jump epoch of a doubly stochastic process (known hereafter as the Cox process) where the jump rate is driven by another stochastic process.  Both approaches were developed extensively in the 2000s and are now commonly used throughout the asset pricing and credit risk literature.
An extension of the structural approach, which we call the \emph{excursion (Parisian) approach,} models it as the first instance in which the amount of time \emph{the asset price stays continuously below a threshold exceeds a given grace period}. Motivated by the Parisian option, this is sometimes called the \emph{Parisian ruin} (see \cite{Chesney}). In the corporate finance literature, the approach has been used to model the reorganization process (Chapter 11), as in \cite{Francois_Morellec, Broadie}.
 Here, reorganization is undertaken whenever the asset value is below a threshold; although there is a chance of recovering to reach above the threshold, if reorganization time exceeds the grace period, the firm is liquidated.  For more information, see the literature review in Section \ref{section_literature_review}.


\subsection{A new model of bankruptcy} \label{section_model_bankruptcy}


This paper considers the scenario where asset value information is updated only at epochs
 $(T_n^\lambda)_{n \geq 1}$, given by the jump times of a Poisson process $(N^\lambda_t)_{t \geq 0}$ with fixed rate $\lambda$.  Given a bankruptcy barrier $V_B$, chosen by the equity holders, bankruptcy is triggered at the first update time where the asset process $(V_t)_{t \geq 0}$ is below $V_B$:
\begin{align}
\inf \{ T_i^\lambda: V_{T_i^\lambda} < V_B \}. \label{our_default}
\end{align}
This is also written as the \emph{classical bankruptcy time}
\begin{align}
\inf \{ t > 0: V^\lambda_{t} < V_B \}, \label{our_default_continuous}
\end{align}
of the asset value if it is only updated at $(T_n^\lambda; n \geq 1)$:
\begin{align*}
V^\lambda_t := V_{T^\lambda_{N^\lambda_t}}, \quad t \geq 0.
\end{align*}
Here $T^\lambda_{N^\lambda_t}$ is the \emph{most recent update time} before $t$.
In Figure \ref{plot_simulated}, we plot sample paths of $(V_t)_{t \geq 0}$,  $(V_t^\lambda)_{t \geq 0}$, $(T_n^\lambda)_{n \geq 1}$ and the corresponding bankruptcy time.

 \begin{figure}[htbp]
\begin{center}
\begin{minipage}{1.0\textwidth}
\centering
\begin{tabular}{c}
 \includegraphics[scale=0.5]{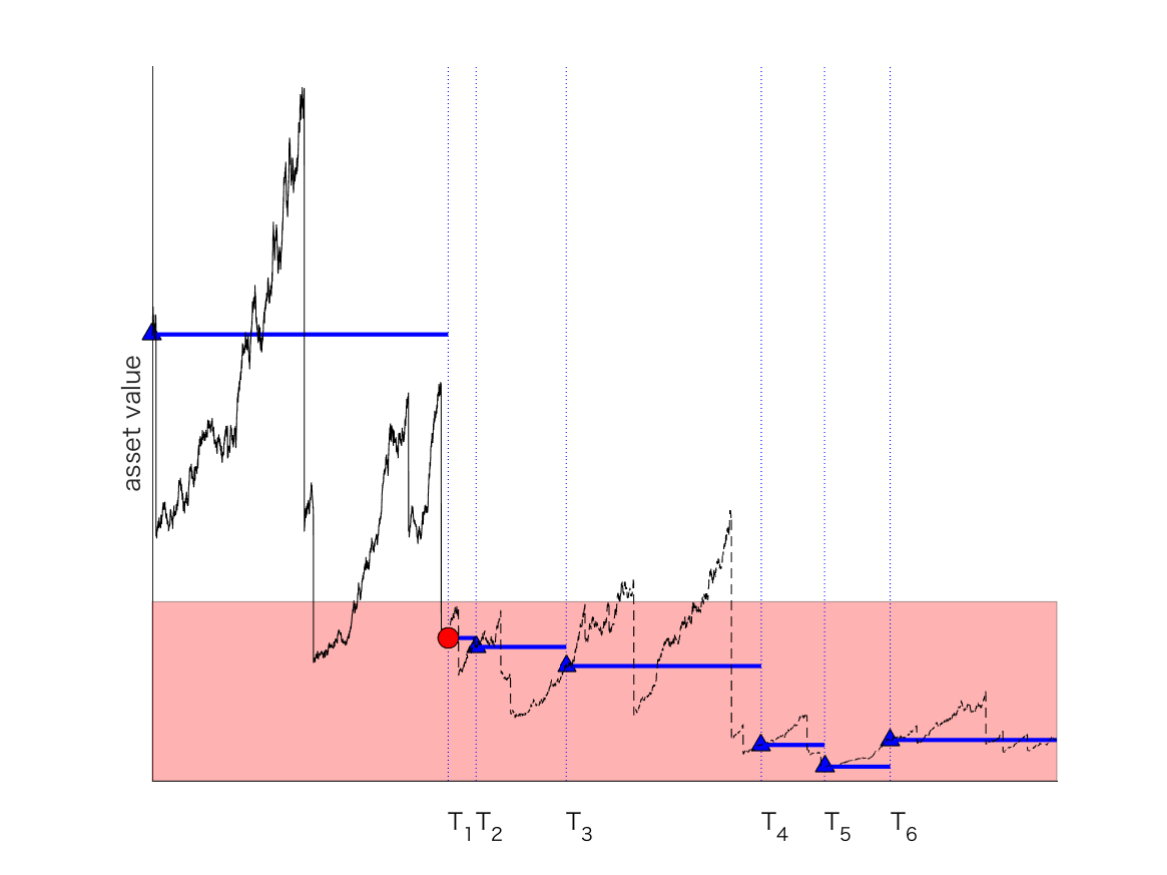}
\end{tabular}
\caption{Sample paths of the asset value $(V_t)_{t \geq 0}$ (black lines) and $(V_t^\lambda)_{t \geq 0}$ (horizontal blue lines) along with the Poisson arrival times $(T_n^\lambda)_{n \geq 1}$ (indicated by dotted vertical lines). The red zone $(0, V_B)$ is given by the rectangle colored in red.
The asset values at  bankruptcy and other observation times are indicated by the red circle and blue triangles, respectively.  Here, the bankruptcy time corresponds to $T_1^\lambda$, but the asset value has crossed $V_B$ before and then recovered back before $T_1^\lambda$.
Note that $(V_t^\lambda)_{t \geq 0}$ has a positive jump at $T_6^\lambda$.
}  \label{plot_simulated}
\end{minipage}
\end{center}
\end{figure}

The bankruptcy model
 \eqref{our_default} is closely related to the reduced-form and excursion approaches reviewed above.
\begin{enumerate}
\item The bankruptcy time \eqref{our_default} is equivalent to the Parisian ruin with the (constant) grace period replaced with an exponential time clock, the first epoch being the time spent continuously below $V_B$ for more than an independent exponential time. For more details see Appendix \ref{bankruptcy_parisianruin}.
\item It is also equivalent to the bankruptcy time in the reduced-form credit risk model, where 
the bankruptcy time is the first jump time of the Cox process  with hazard rate given by $(h_t := \lambda \mathbf{1}_{\{V_t < V_B \}})_{t \geq 0}$. As in Figure \ref{plot_simulated}, the region $(0, V_B)$ can be seen as the ``red zone"; here, bankruptcy is triggered at rate $\lambda$ whereas, in the ``healthy zone" $(V_B, \infty)$,  this probability is negligible.
\end{enumerate}


There are several motivations for considering the bankruptcy strategy \eqref{our_default} for the study of capital structures.

First, in reality, it is not possible to continuously observe the accurate status of a firm and make bankruptcy decisions instantaneously.  In addition, unlike in the case of American options pricing, for which computer programs can be set up to exercise automatically, in our case, information is acquired by humans. As observed in the literature of \emph{rational inattention} \cite{Sims}, the amount of information a decision maker can capture and handle is limited, and instead they rationally decide to stay with imperfect information.  Taking a bankruptcy decision requires complex information and it is more realistic to assume that the information for the decision makers is updated only at random discrete times.  While they are expected to respond promptly,  delays are inevitable and possibly have a significant impact on bankruptcy costs.


Second,
 the majority of the existing literature assumes continuous observation using a continuous asset value process -- in this case, the asset value at bankruptcy is, in any event, precisely $V_B$.  Unfortunately, it is unreasonable to assume that one can precisely predict the asset value at bankruptcy, which is in reality random.  The randomness can be realized by adding negative jumps to the process.  We  underline that in our model this
randomness can also be achieved by any choice (continuous or c\'{a}dl\'{a}g) of the underlying process.
See Figure \ref{fig_bankruptcy_value_r} in Section  \ref{section_numerics}.






Third, this model generalizes the classical model and allows more flexibility by having one more parameter $\lambda$. The classical structural model (with instantaneous liquidation upon downcrossing the barrier) corresponds to the case $\lambda = \infty$ and the no-bankruptcy model corresponds to the case $\lambda = 0$.  With careful calibration of $\lambda$, the model can potentially estimate the bankruptcy costs and tax benefits more precisely. Typically, for calibration, credit spread data is used.  As shown in the numerical results (see Figure \ref{plot_credit_spread}), a variety of term structures can be achieved by choosing the value of $\lambda$.

Finally, thanks to the equivalence of our bankruptcy time with the classical bankruptcy time \eqref{our_default_continuous} of the process $(V^\lambda_t)_{t \geq 0}$, this research can be considered a contribution  to the classical structural approach.  Existing results featuring asset value processes with two sided jumps are rather limited.
However, we provide a new analytically tractable case for $(V^\lambda_t)_{t \geq 0}$, containing two-sided jumps even when  $(V_t)_{t \geq 0}$ does not have positive jumps
(see Figure \ref{plot_simulated}).  By appropriately selecting the driving process $(V_t)_{t \geq 0}$ as well as $\lambda$, it is possible to construct a wide range of stochastic processes with two-sided jumps.






\subsection{Contributions of the paper}
This model is built based on
the seminal paper by Leland and Toft \cite{Leland96}, with a feature of endogenous default. While Leland \cite{Leland94}'s framework is more frequently used and is certainly more mathematically tractable, its extension  \cite{Leland96} more accurately  captures  the flow of debt financing by successfully avoiding the use of perpetual bonds assumed in \cite{Leland94}.

In addition, while the majority of papers in financial economics assume a geometric Brownian motion for the asset price $(V_t)_{t \geq 0}$, we  follow the works of Hilberink and Rogers \cite{Hilberink}, Kyprianou and Surya \cite{Kyprianou} and Surya and Yamazaki \cite{Surya_Yamazaki} and consider an exponential \lev process with arbitrary negative jumps (spectrally negative \lev processes). Although it is  more desirable to also  allow positive jumps as in Chen and Kou \cite{Chen_Kou_2009},  as discussed in \cite{Hilberink}, negative jumps occur more frequently and  effectively model  the downward risks.
With the spectrally negative assumption, semi-explicit expressions of the equity value as well as the optimal bankruptcy threshold are elicited, without focusing on a particular set of jump measures. Again, see the discussion above on how our model is capable of modeling the two-sided jump case in the classical structural approach, even when a spectrally negative \lev process is used for $(V_t)_{t \geq 0}$. For a more general study of financial models using \lev processes, the reader should   refer to Cont and Tankov \cite{ContTankov}.


To solve the problem, recent developments of the \emph{fluctuation theory} of \lev processes are utilized.
First, the firm/debt/equity values are expressed in terms of the so-called \emph{scale functions}, which exist for a general spectrally negative \lev process.  These permit direct computation of the optimal bankruptcy barrier and the corresponding firm/debt/equity values. 


With these analytical results, a sequence of numerical experiments can be conducted.
Here, to easily comprehend the impacts of the parameters describing the problem, we use a (spectrally negative) hyperexponential jump diffusion (a mixture of Brownian motion and i.i.d.\ hyperexponentially distributed jumps), for which the scale function can be written as a sum of exponential functions.  The equity/debt/firm values can be written explicitly and the optimal bankruptcy barrier can be computed instantaneously by a classical bisection method. The optimal capital structure is obtained by solving the \emph{two-stage optimization problem} as proposed in \cite{Leland96}.  In addition, with numerical Laplace inversion, we also obtain the term structures of credit spreads and the density/distribution of the bankruptcy time and the corresponding asset value.
Because various numerical experiments have already been conducted in other papers, here we focus on analyzing the impacts of the frequency of observation $\lambda$.  We verify the convergence to the classical case of \cite{Hilberink, Kyprianou}, and also observe monotonicity, with respect to $\lambda$, of  the bankruptcy barrier, firm value under the optimal capital structure, the optimal leverage, and the credit spread.



\subsection{Related literature} \label{section_literature_review}
Before concluding this section, we review several relevant papers motivating our problem.

The most relevant paper, to our best knowledge, is Francois and Mollerec \cite{Francois_Morellec}, in which the authors modeled the reorganization process (Chapter 11)
using the excursion approach with a deterministic grace period as described above.
Broadie et al. \cite{Broadie} considered a similar model with an additional barrier for immediate liquidation upon crossing, whereas Moraux \cite{Moraux} considered a variant of \cite{Francois_Morellec}  using the \emph{occupation time approach}, in which distress level accumulates without being reset each time the asset process recovers to a healthy state.
These papers are based on
Leland \cite{Leland94}, with perpetual bonds  and asset values driven by geometric Brownian motions for mathematical tractability. However, it is significantly more challenging than the classical structural approach and hence most of them rely on numerical approaches.  In this paper, on the other hand,  semi-analytical solutions for a more general asset value process with jumps are obtained as a result of the
use of Poisson arrival times for the update times.

This paper is also motivated by Duffie and Lando \cite{Duffie}, in which they modeled the asymmetry of information between firms and bond investors.
The authors assumed that bond investors cannot observe the firm's assets directly and that instead, they  receive only periodic and imperfect accounting reports on the firm's status. Under these assumptions, the authors successfully explained the non-zero credit spread limit.

Regarding the study of \lev processes observed at Poisson arrival times, there has been substantial progress in the last few years. Recently, Albrecher and Ivanovs \cite{Albrecher_Ivanovs} investigated close links between \lev processes observed continuously and periodically. In results similar to those for
the classical hitting time at a barrier, they found that the exit identities under periodic observation can be obtained,
if the Wiener-Hopf factorization is known.  In particular, when focusing on the spectrally one-sided case, these can be written in terms of the scale function.
For the results of our paper, we use the joint Laplace transform of the bankruptcy time \eqref{our_default} and the asset value in that instance, which is obtained in \cite{Albrecher, Albrecher_Ivanovs}. In addition, we obtain the resolvent measure killed at the first Poissonian downward passage time 
\eqref{our_default} for the computation of tax benefits.

Regarding the optimal stopping problems under Poisson observations, perpetual American options have been studied by Dupuis and Wang \cite{Dupuis_Wang} for the geometric Brownian motion case. This has recently been generalized to the \lev case by P\'erez and Yamazaki \cite{Perez_Yamazaki_options}.  Several key studies have been performed on the application of scale functions in optimal stopping in the continuous observation setting (e.g., \cite{Alili_Kyprianou, Avram_Kyprianou_Pistorius, Long, Rodosthenous, Surya}). The periodic observation model is more frequently used in the insurance community, in particular in the optimal dividend problem (see \cite{Avanzi_Cheung_Wong_Woo, Avanzi_Tu_Wong, Noba_Perez_Yamazaki_Yano}).


To the best of our knowledge, this is the first attempt to introduce Poisson observations in the problem of capital structures.  We believe the techniques used in this paper can be used similarly in related problems described above when the Poisson observation is introduced.






\subsection{Organization of the paper}

The organization of this paper is as follows.
In Section \ref{section_problem} we present formally the main problem that we work on in this article.
In Section \ref{section_equity_value}, we compute the equity value using the scale function, and, in Section \ref{section_optimal_barrier}, we identify the optimal barrier.
Section  \ref{section_two_stage} considers the two-stage problem to obtain the optimal capital structure.
Section \ref{section_numerics} deals with numerical examples confirming theoretical results. Section \ref{section_conlusion} concludes the paper.
Long proofs are deferred to the Appendix.

\section{Problem Formulation} \label{section_problem}

Let $(\Omega, \F, \p)$ be a complete probability space hosting a  L\'{e}vy process $X=(X_t)_{t\ge 0}$.  
The value of the \emph{firm's asset} is assumed to evolve according to an \emph{exponential \lev process} given by,  for the initial value $V > 0$,
\[V_t := V e^{X_t},\qquad t\geq0.
\]
   Let $r > 0$ be the positive risk-free interest rate and $0 \leq \delta < r$ the total payout rate to the firm's investors.  We assume that the market is complete and this requires $(e^{-(r-\delta) t} V_t)_{t \geq 0}$ to be a $\p$-martingale.

The firm is partly financed by debt with a constant debt profile: it issues,  for some given constants $p, m > 0$, new debt at a constant rate $p$ with maturity profile $\varphi(s) := m e^{-ms}$.  In other words,  the face value of the debt issued in the small time interval $(t, t+\diff t)$ that matures in the small time interval $(t+s, t+s+\diff s)$ is approximately given by $p \varphi(s) \diff t \diff s$.
Assuming the infinite past,
the face value of debt held at time $0$ that matures in $(s, s+ \diff s)$ becomes
\begin{align}
\left[  \int_{-\infty}^0 p \varphi(s-u) \diff u \right] \diff s = p e^{-m s} \diff s, \label{mat_profile}
\end{align}
and {\it the face value of all debt} is a constant value,
\begin{align*}
P := \int_0^\infty p e^{-m s} \diff s = \frac p m.
\end{align*}
For more details,  see \cite{Hilberink, Kyprianou}.

Let $(N^\lambda_t)_{t \geq 0}$ be {\it an independent Poisson process with rate $\lambda > 0$} and $\mathcal{T} := (T_n^\lambda)_{n \geq 1}$ be its jump  times. Suppose the bankruptcy is triggered at the first time of $\mathcal{T}$ the asset value process $(V_t)_{t \geq 0}$ goes below a given level $V_B > 0$:
\begin{align}
T_{V_B}^-  := \inf \left\{ S \in \mathcal{T}: V_{S} < V_B \right\}
\label{default_time}
\end{align}
with the convention $\inf \emptyset =\infty$. In our model, it is more natural to assume that the bankruptcy decision can be made at time zero.  Hence, we modify the above and consider the random time
\begin{align}
\overline{T}_{V_B}^-  := \inf \left\{ S \in \mathcal{T} \cup \{0\}: V_{S} < V_B \right\} = T_{V_B}^- \mathbf{1}_{\{V \geq V_B \}}.
\label{default_time_modified}
\end{align}



(i)  Suppose $V \geq V_B$ so that $\overline{T}_{V_B}^-  = T_{V_B}^-$.

The debt pays {\it a constant coupon flow at a fixed rate $\rho > 0$ and a constant fraction $0 <\alpha < 1$ of the asset value is lost at the bankruptcy time $T_{V_B}^-$}.
In this setting, the value of the debt with a unit face value  and maturity $t>0$ becomes
\begin{align}
	d (V; V_B, t) := \E \left[ \int_0^{t \wedge T_{V_B}^-  } e^{-rs} \rho \diff s \right] + \E\left[ e^{-rt} \mathbf{1}_{\{t < T_{V_B}^-  \}}\right] + \frac 1 P \E \left[ e^{-r T_{V_B}^-}V_{T_{V_B}^-} \left(1-\alpha  \right) \mathbf{1}_{\{T_{V_B}^-  < t \}}\right]. \label{debt_constant_unit}
\end{align}	
Here, the first term is the total value of the coupon payments accumulated until maturity or bankruptcy whichever comes first;  the second term is the value of the principle payment; the last term corresponds to the $1/P$ fraction of the remaining asset value that is distributed, in the event of bankruptcy, to the bondholder of a unit face value.
Integrating this, the \emph{total value of debt} becomes, by \eqref{mat_profile} and Fubini's theorem,
	\begin{align*}
		\mathcal{D} (V; V_B) &:= \int_0^\infty p e^{-m t} d (V; V_B, t) \diff t \\
		&= \E\left[ \int_0^{T_{V_B}^-}  e^{-(r+m)t} \left( P \rho+ p \right) \diff t\right] +  \E \left[ e^{-(r+m) T_{V_B}^-}V_{T_{V_B}^-} \left(1-\alpha  \right) \mathbf{1}_{\{T_{V_B}^-  < \infty \}}\right].
	\end{align*}
	
Regarding the \emph{value of the firm}, it is assumed that there is a corporate tax rate $\kappa > 0$ and its (full) rebate on coupon payments is gained if and only if $V_t \geq V_T$ for some given cut-off level $V_T \geq 0$ (for the case $V_T = 0$, it enjoys the benefit at all times).  Based on the trade-off theory (see e.g.\ \cite{Brealey_Myers_2001}), {\it the firm value} becomes the sum of the asset value and total value of tax benefits less the value of loss at bankruptcy, given by
\begin{align}\label{fun_V}
	\mathcal{V}(V;V_B) &:= V + \E \left[ \int_0^{T_{V_B}^- } e^{-rt} \mathbf{1}_{\{V_t \geq  V_T \}} P \kappa \rho \diff t\right] -  \alpha  \E \left[ e^{-r T_{V_B}^-}V_{T_{V_B}^-} \mathbf{1}_{\{T_{V_B}^-  < \infty \}}   \right].
\end{align}	

(ii)  Suppose $V < V_B$ so that $\overline{T}_{V_B}^-  = 0$ a.s. Then, 
\begin{align}
\mathcal{D} (V; V_B) = \mathcal{V} (V; V_B) = (1-\alpha) V. 
\label{D_V_when_V_less_V_B}
\end{align}

The problem is to pursue an \emph{optimal bankruptcy level} $V_B \geq 0$ that maximizes the \emph{equity value},
\begin{align}
\mathcal{E}(V;V_B) :=\mathcal{V}(V;V_B) - \mathcal{D} (V; V_B),
\label{equity}
\end{align}
 subject to the \emph{limited liability constraint},
\begin{align}
\mathcal{E}(V;V_B) \geq 0, \quad V \geq  V_B, \label{constraint}
\end{align}
if such a level exists. Here,  $V_B = 0$ means that it is never optimal to go bankrupt with the limited liability constraint satisfied for all $V >0$. 
Note that when $V < V_B$ then \eqref{D_V_when_V_less_V_B} gives $\mathcal{E}(V;V_B) = 0$.

\section{Computation of the equity value} \label{section_equity_value}

Suppose from now on that $(X_t)_{t \geq 0}$ is a spectrally negative \lev process, that is a L\'evy process without positive jumps.
We denote by
\begin{align} \label{laplace_exponent}
\psi(\theta) := \log \e\big[{\rm e}^{\theta X_1}\big], \qquad \theta\ge 0
\end{align}
its Laplace exponent with the right-inverse
		\begin{align}
				\Phi(q) := \sup \{ s \geq 0: \psi(s) = q\}, \quad q \geq 0.
			\label{def_varphi}
		\end{align}
\subsection{Scale functions}
The starting point of whole analysis is introducing  the so-called $q$-scale function $W^{(q)}(x)$, with $q\geq 0$ and $x\in\mathbb{R}$.
It features invariably in almost all known fluctuation identities of spectrally negative L\'evy processes; see
Zolotarev \cite{Z} and Tak\'acs \cite{Ta} for the origin of this function. See also \cite{Kyprianou, KKR}  for a detailed review.

Fix $q \geq 0$. The $q$-scale function $W^{({q})}$ is the mapping from $\R$ to $[0, \infty)$ that takes value zero on the negative half-line, while on the positive half-line it is a continuous and strictly increasing function with the Laplace transform:
		\begin{align} \label{scale_function_laplace}
			\begin{split}
				\int_0^\infty  \mathrm{e}^{-\theta x} W^{({q})}(x) \diff x &= \frac 1 {\psi(\theta)-q}, \quad \theta > \Phi({q}).
			\end{split}
		\end{align}
Define also the second scale function: 
\begin{align*} 
Z^{({q})}(x; \theta ) &:=e^{\theta x} \left( 1 + (q- \psi(\theta )) \int_0^{x} e^{-\theta  z} W^{(q)}(z) \diff z	\right), \quad x \in \R, \, \theta  \geq 0.
\end{align*}
In particular, for $x \in \R$, we let $Z^{(q)}(x) :=Z^{(q)}(x; 0)$ and, for $\lambda > 0$,
\begin{align*} 
	\begin{split}
		Z^{(q)}(x; \Phi(q+\lambda)) &=e^{\Phi(q+\lambda) x} \left( 1 -\lambda \int_0^{x} e^{-\Phi(q+\lambda) z} W^{(q)}(z) \diff z \right).
	\end{split}
\end{align*}
In the next section, we see that the equity value \eqref{equity} can be written in terms of the scale functions $W^{(q)}$ and $Z^{(q)}$.

\subsection{Related fluctuation identities} \label{section_fluctuation_identities}

For $y \in \mathbb{R}$, let $\p_y$ be the conditional probability under which the initial value of the spectrally negative \lev process is $X_0 = y$. 

Following equation (4.5) in \cite{Kyprianou} (see also Emery \cite{Emery} and \cite[eq. (3.19)]{APP2007}), the joint Laplace transform of the first passage time
\begin{equation}\label{tauzero}\tau_0^-:=
\inf\{t\geq 0: X_t<0\} \end{equation}
and $X_{\tau_0^-}$
is given by the following identity
\begin{align}\label{der_gamma_0_0}
\begin{split}
H^{(q)}(y; \theta) &:= \E_{y} \left[ e^{-q \tau_0^- + \theta X_{\tau_0^-}} \mathbf{1}_{\{\tau_0^-< \infty \}}\right]
=Z^{({q})}(y;\theta)-\frac{\psi(\theta)-q}{\theta-\Phi(q)}W^{({q})}(y),
\end{split}
\end{align}
where $y\in\R$, $\theta\geq 0$, and $q \geq 0$.
Similar results have been obtained for the Poisson observation case.
Recall that $\mathcal{T} :=(T_n^\lambda; n \geq 1)$ is
the set of jump times of an independent Poisson process. We define
\begin{align}
\tilde{T}_{z}^-  := \inf \left\{ S \in \mathcal{T} : X_{S} < z \right\}, \quad z \in \R.
\label{default_time}
\end{align}
By equation (14) of Theorem 3.1 in \cite{Albrecher}, for $\theta \geq 0$ and $y \in \mathbb{R}$,
\begin{align}\label{fun_gamma}
\begin{split}
J^{(q, \lambda)}(y; \theta) &:= \E_y\left[ e^{-q \tilde{T}_0^- + \theta X_{\tilde{T}_0^-}} \mathbf{1}_{\{\tilde{T}_0^-< \infty \}} \right] \\
&=\frac{\lambda}{\lambda +{q}-\psi(\theta)}\left(Z^{({q})}(y;\theta)-Z^{({q})}(y;\Phi({q}+\lambda))\frac{\psi(\theta)-{q}}{\lambda}\frac{\Phi({q}+\lambda)-\Phi({q})}{\theta-\Phi({q})}\right) \\
&=\Big[1-\frac{(\psi(\theta)-q)}{(\theta-\Phi(q))} \frac{(\Phi(\lambda+q)-\theta )}{(\lambda+q-\psi(\theta))}\Big]Z^{(q)}(y;\theta) \\
&\hspace{0.85cm}-\frac{(\psi(\theta)-q)}{(\theta-\Phi(q))} \frac{(\Phi(\lambda+q)-\Phi(q))}{(\lambda+q-\psi(\theta))}\big(Z^{(q)}(y;\Phi(\lambda+q)) -Z^{(q)}(y;\theta)\big).
\end{split}
\end{align}

\begin{remark}\label{remark_gamma_0} (1) We have
\begin{align*}
\begin{split}
J^{(q, \lambda)}(0; 1)
&=\frac{\lambda}{\lambda +{q}-\psi(1)} - \frac{\psi(1)-{q}}{\lambda +{q}-\psi(1)}  \frac{\Phi({q}+\lambda)-\Phi({q})}{1-\Phi({q})}=1-\frac{\psi(1)-{q}}{\lambda +{q}-\psi(1)}  \frac{\Phi({q}+\lambda)-1}{1-\Phi({q})} > 0, \\
J^{(q, \lambda)}(0; 0)
&=\frac{\lambda}{\lambda +{q}} - \frac{{q}}{\lambda +{q}}  \frac{\Phi({q}+\lambda)-\Phi({q})}{\Phi({q})} =1- \frac{{q}}{\lambda +{q}}  \frac{\Phi({q}+\lambda)}{\Phi({q})} > 0,
\end{split}
\end{align*}
where the positivity holds 
	by the probabilistic expression of $J^{(q,\lambda)}$ as in \eqref{fun_gamma}.

(2) 
We have
\begin{align} \label{J_less_than_1}
J^{(q, \lambda)}(y; \theta) < 1, \quad q > 0, \, \theta \geq 0, \, y \in \mathbb{R}.
\end{align}
To see this,
by the memoryless property of the exponential random variable, we can write, for some independent exponential random variable $\mathbf{e}_\lambda$, the first observation time at which $X$ is below zero is $\tau_0^- + \mathbf{e}_\lambda$ and hence $\tilde{T}_0^-$ is bounded from below by an exponential random variable. In addition, we must have $X_{\tilde{T}_0^-} \leq 0$ $\mathbb{P}_y$-a.s. and hence we have \eqref{J_less_than_1}.
\end{remark}

In order to write the equity value, we obtain an expression for
\begin{align} \label{Lambda_def}
\Lambda^{(r,\lambda)}(y,z) &:= \E_y\left[ \int_0^{\tilde{T}_{z}^-} e^{-rt}\mathbf{1}_{\{X_t \geq \log V_T\}}  \diff t\right], \quad y, z \in \mathbb{R}.
\end{align} 
In Appendix \ref{proof_prop_lambdaidentified}, we obtain the resolvent measure killed at $\tilde{T}_{z}^-$ and the following result as a corollary.
\begin{proposition}\label{lambdaidentified}
Fix $y, z \in \mathbb{R}$. For $V_T > 0$, we have
\begin{align*}
\Lambda^{(r,\lambda)}(y,z)
& = Z^{(r)}(y-z;\Phi(r+\lambda))\frac  {\Phi(r+\lambda)- \Phi(r)} {\lambda}
\\ &\quad \times \Big( \frac 1 {\Phi(r)} Z^{(r+\lambda)}(z-\log V_T;\Phi(r)) - \frac \lambda {\Phi(r)} \overline{W}^{(r+\lambda)}(z-\log V_T)  \Big)
\\
&\quad- \overline{W}^{(r+\lambda)}(y- \log V_T) \mathbf{1}_{\{z > \log V_T\}} -  \overline{W}^{(r)}(y-\log V_T)  \mathbf{1}_{\{ z \leq \log V_T \}} \\ &\quad+\lambda \mathbf{1}_{\{ z > \log V_T\}} \int_0^{y-z}W^{(r)}(y-z-u)   \overline{W}^{(r+\lambda)}(u+z-\log V_T) \diff u,
\end{align*}
where $\overline{W}^{(q)}(y) := \int_0^y W^{(q)}(u) \diff u$ for all $q > 0$ and $y \in \R$. 

For $V_T = 0$, we have
$\Lambda^{(r,\lambda)}(y,z)=(1-J^{(r,\lambda)}(y-z;0))/r.$


\end{proposition}
%

\subsection{Expression for the equity value in terms of the scale function} \label{subsection_computation_equity_scale}

Using the identities in Section \ref{section_fluctuation_identities}, the equity value \eqref{equity} can be written as follows. Here, we focus on the case $V_B > 0$.  The case $V_B = 0$ (for which, as we will see, only the case $V_T = 0$ needs to be considered) is given later in \eqref{equity_value_0}.


First by \eqref{fun_gamma},  we have, for $q = r$ and $q = r+m$,
\begin{align*}
\E \left[ e^{-q T_{V_B}^-}V_{T_{V_B}^-} \mathbf{1}_{\{T_{V_B}^-  < \infty \}}   \right] &= V_B J^{(q, \lambda)} \Big(\log \frac V {V_B}; 1 \Big) \quad \textrm{and} \quad
\E \left[ e^{-q T_{V_B}^-} \mathbf{1}_{\{T_{V_B}^-  < \infty \}}   \right] &= J^{(q, \lambda)} \Big(\log \frac V {V_B}; 0 \Big).
\end{align*}
In addition, 
by \eqref{Lambda_def},
\begin{align*}
\E \left[ \int_0^{T_{V_B}^- } e^{-rt} \mathbf{1}_{\{V_t \geq  V_T \}}  \diff t\right] = \Lambda^{(r,\lambda)} (\log V, \log V_B).
\end{align*}
Hence, we can write
\begin{align} \label{value_debt_asset}
\begin{split}
\mathcal{D} (V; V_B)
&= \frac {P \rho+ p} {r+m}  \Big(1- J^{(r+m, \lambda)}\Big(\log \frac {V} {V_B};0 \Big) \Big) + \left(1-\alpha  \right) V_B J^{(r+m, \lambda)} \Big(\log \frac {V} {V_B}; 1 \Big), \\
\mathcal{V}(V;V_B)
&= V + P \kappa \rho \Lambda^{(r,\lambda)}(\log V, \log V_B) -  \alpha  V_B J^{(r, \lambda)}\Big(\log \frac {V} {V_B}; 1 \Big),
\end{split}
\end{align}
and therefore, by taking their difference, the equity value is 
\begin{align} \label{equity_value}
\begin{split}
\mathcal{E}(V;V_B)
&= V + P \kappa \rho \Lambda^{(r,\lambda)}(\log V, \log V_B) -  \alpha  V_B J^{(r, \lambda)}\Big(\log \frac {V} {V_B}; 1\Big) \\
&\qquad- \frac {P \rho+ p} {r+m}  \Big(1- J^{(r+m, \lambda)}\Big(\log \frac {V} {V_B};0 \Big) \Big) - \left(1-\alpha  \right) V_B J^{(r+m, \lambda)}\Big(\log \frac {V} {V_B}; 1 \Big).
\end{split}
\end{align}
	


\section{Optimal barrier} \label{section_optimal_barrier}

Having the equity value $\mathcal{E}(V;V_B)$ given in \eqref{equity_value} identified using equation \eqref{fun_gamma} and Proposition \ref{lambdaidentified}, we are ready to find
the optimal barrier $V_B^*$  maximizing it.
Our objective in this paper is to show that the optimal barrier is $V_B^*$ such that
\begin{align}\label{optimalB*}
\mathcal{E}(V_B^*;V_B^*)=0,
\end{align}
if it exists, where, by \eqref{equity_value} and Remark \ref{remark_gamma_0}(1), for $V_B > 0$,
\begin{align} \label{E_B_B}
\begin{split}
\mathcal{E}(V_B;V_B)  &= V_B + P \kappa \rho \Lambda^{(r,\lambda)}(\log V_B, \log V_B) \\ &\qquad-  \alpha  V_B J^{(r, \lambda)}(0; 1) - \frac {P \rho+ p} {r+m}  (1- J^{(r+m, \lambda)}(0;0) ) - \left(1-\alpha  \right) V_B J^{(r+m, \lambda)}(0; 1) \\
&= V_B [1 -  \alpha J^{(r, \lambda)}(0; 1) - (1-\alpha) J^{(r+m, \lambda)}(0; 1)] +  P \kappa \rho \Lambda^{(r,\lambda)}(\log V_B, \log V_B) \\&\qquad- \frac {P \rho+ p}{\lambda +r+m}  \frac{\Phi(r+m+\lambda)}{\Phi(r+m)}.
\end{split}
\end{align}
\subsection{Existence} \label{section_existence}


We first show the condition for the existence of $V_B^*$ satisfying \eqref{optimalB*}.  To this end, we show the following result; the proof is given in
Appendix \ref{proof_lemma_Lambda_limit}.
\begin{lemma}
	\label{lemma_Lambda_limit} 
The mapping $z \mapsto \Lambda^{(r,\lambda)}(z,z)$ is non-decreasing on $\R$ with the limit
\begin{align*}
\lim_{z \downarrow -\infty }\Lambda^{(r,\lambda)}(z,z)
= \begin{cases} 0 & \textrm{if } \;  V_T > 0, \\ \frac{1}{\lambda +r}  \frac{\Phi(r+\lambda)}{\Phi(r)} &\textrm{if } \;  V_T = 0. \end{cases}
\end{align*}
\end{lemma}


\begin{figure}[htbp]
\begin{center}
\begin{minipage}{1.0\textwidth}
\centering
\begin{tabular}{cc}
 \includegraphics[scale=0.5]{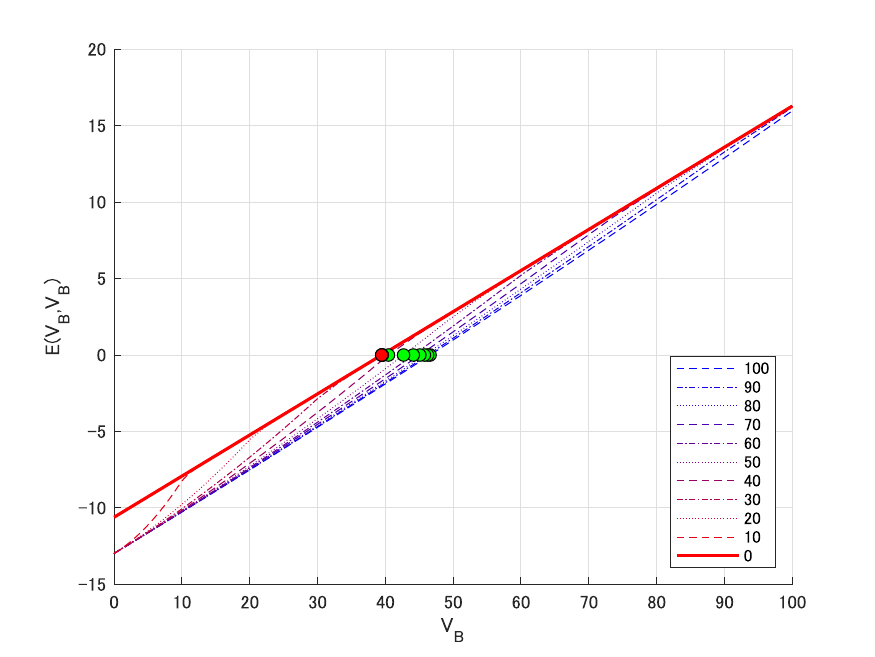} & \includegraphics[scale=0.5]{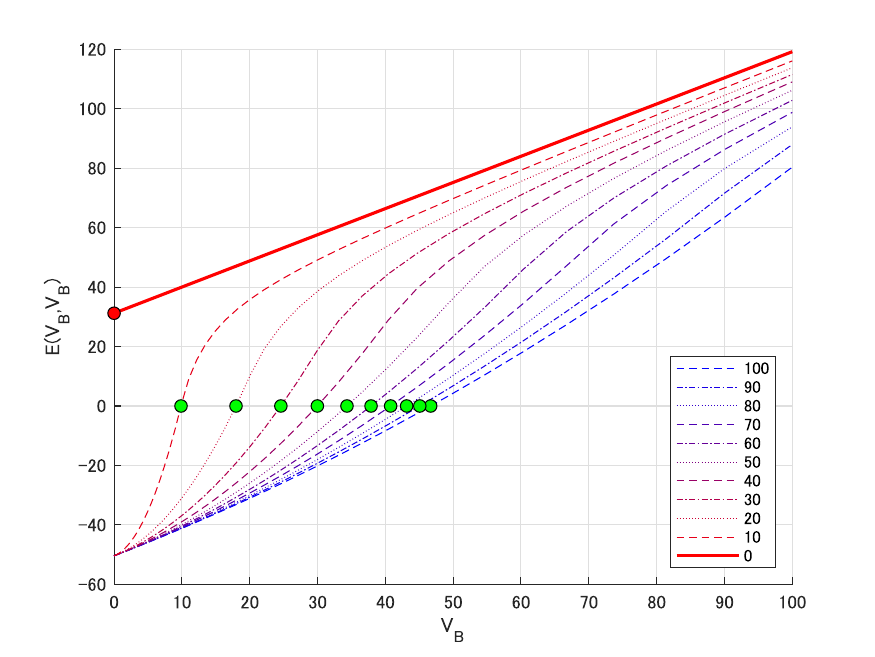}
\end{tabular}
\end{minipage}
\end{center}
\caption{Plots of $V_B \mapsto \mathcal{E}(V_B;V_B)$ for $V_T = 0,10,20,\ldots, 100$. Solid lines show for the case $V_T=0$
and dotted lines for the other cases. The points at $V_B^*$ are indicated by circles.   The left plot is based on the parameter set in \textbf{Case B} in Section \ref{section_numerics} (except $V_T$), and achieves $V_B^* > 0$ for all cases. The right plot is based on the same parameters except that we set $\kappa = 0.9999$, $m = 10$, $\rho = 0.2$ and $\lambda = 0.1$ to achieve $V_B^* = 0$ when $V_T = 0$.} \label{fig_monotonicity}
\end{figure}

This lemma leads to the following proposition. For numerical illustration, see Figure \ref{fig_monotonicity}.
\begin{proposition} \label{E_monotonicity} 

The mapping $V_B \mapsto \mathcal{E}(V_B;V_B)$ is strictly increasing on $(0, \infty)$ with the limit:
\begin{align*}
\lim_{V_B \downarrow 0 }\mathcal{E}(V_B;V_B)
&= \begin{cases} - \frac {P \rho+ p}{\lambda +r+m}  \frac{\Phi(r+m+\lambda)}{\Phi(r+m)} & \textrm{if } \;  V_T > 0, \\  \frac{P \kappa \rho}{\lambda +r}  \frac{\Phi(r+\lambda)}{\Phi(r)}-\frac {P \rho+ p}{\lambda +r+m}  \frac{\Phi(r+m+\lambda)}{\Phi(r+m)} & \textrm{if } \;  V_T = 0, \end{cases} \\
\lim_{V_B \uparrow \infty }\mathcal{E}(V_B;V_B) &= \infty.
\end{align*}
\end{proposition}
\begin{proof}
From Remark \ref{remark_gamma_0}(2),
we have
$1 -  \alpha J^{(r, \lambda)}(0; 1) - (1-\alpha) J^{(r+m, \lambda)}(0; 1) > 0$.
By this, Lemma \ref{lemma_Lambda_limit} and because $z \mapsto \Lambda^{(r,\lambda)}(z,z)$ is non-decreasing and bounded, the claim is immediate in view of the second equality of \eqref{E_B_B}.
\qed \end{proof}

Now by Proposition \ref{E_monotonicity}, we define the candidate optimal threshold $V_B^*$ formally, as follows.
\begin{enumerate}
\item For the case $V_T > 0$ and the case $V_T= 0$ with $P \kappa \rho \frac{1}{\lambda +r}  \frac{\Phi(r+\lambda)}{\Phi(r)}-\frac {P \rho+ p}{\lambda +r+m}  \frac{\Phi(r+m+\lambda)}{\Phi(r+m)} < 0$, we set $V_B^*>0$ such that
 $\mathcal{E}(V_B^*; V_B^*) = 0$, whose existence and uniqueness hold by Proposition \ref{E_monotonicity}. 
 \item For the case $V_T= 0$ with 
\begin{equation}\label{newref}
  \frac{P \kappa \rho}{\lambda +r}  \frac{\Phi(r+\lambda)}{\Phi(r)}-\frac {P \rho+ p}{\lambda +r+m}  \frac{\Phi(r+m+\lambda)}{\Phi(r+m)} \geq 0,
 \end{equation} we set $V_B^* =0$.
\end{enumerate}
The debt/firm/equity values for the case $V_B^* >0$ can be computed by \eqref{value_debt_asset} and \eqref{equity_value}.
For the case $V_B^* =0$, where necessarily $V_T=0$, we have, for all $V > 0$,
\begin{align*}
	\mathcal{D} (V; 0)&=\E\left[ \int_0^{\infty}  e^{-(r+m)t} \left( P \rho+ p \right) \diff t\right]=\frac{P \rho+ p}{r+m}, \\
	\mathcal{V}(V;0) &= V + \E \left[ \int_0^{\infty } e^{-rt}  P \kappa \rho \diff t\right]=V+ \frac{P \kappa \rho}{r},
\end{align*}
and therefore
\begin{align}\label{equity_value_0}
	\mathcal{E}(V;0)=V+ \frac{P \kappa \rho}{r}-\frac{P \rho+ p}{r+m}.
\end{align}

\subsection{Optimality}
For the rest of this section, we show the following one of our main results.
\begin{theorem} \label{theorem_main}
The barrier $V_B^*$ 
is optimal for the problem of maximizing \eqref{equity} subject to \eqref{constraint}. 
\end{theorem}

To prove the optimality, it is sufficient to show the following:
\begin{enumerate}
\item If $V_{B}^* > 0$, every threshold $V_B < V_{B}^*$ violates the limited liability constraint \eqref{constraint}.
\item $V_B^*$ attains a higher equity value than any $V_B > V_B^*$ does.
\item $V_B^*$ is feasible.
\end{enumerate}

\begin{proposition} \label{feasibility_B_less_star} 
Suppose $V_B^* > 0$. For $V_B < V_B^*$,  the limited liability constraint \eqref{constraint} is not satisfied.
\end{proposition}
\begin{proof}
	By the (strict) monotonicity as in Proposition \ref{E_monotonicity} and because  $\mathcal{E}(V_B^*; V_B^*) =0$ (given that $V_B^* > 0$), we have $\mathcal{E}(V_B ;V_B) < 0$ for $V_B < V_B^*$.  
\qed \end{proof}

The proof of the following is given in Appendix \ref{proof_prop_E_derivative_B}.
\begin{proposition} \label{prop_E_derivative_B}
For $V > V_B > 0$, we have
\begin{align} \label{eqn_E_derivative_B}
	&\frac{\partial}{\partial V_B} \mathcal{E}(V;V_B) =-(\Phi(r+m+\lambda)-\Phi(r+m))H^{(r+m)}\Big(\log \frac V {V_B}; \Phi(r+m+\lambda) \Big) \frac {L(\log V,\log V_B)}  {V_B}
\end{align}
where $H^{(r+m)}$ is as in \eqref{der_gamma_0_0} and, for $x,z \in \mathbb{R}$,
\begin{align*}
L(x,z) &:= \frac{H^{(r)}(x-z; \Phi(r+\lambda))}{H^{(r+m)}(x-z; \Phi(r+m+\lambda))} \frac{\Phi(r+\lambda)-\Phi(r)}{\Phi(r+m+\lambda)-\Phi(r+m)}\\ &\times \Bigg(\alpha(1-J^{(r, \lambda)}(0;1))e^{z}+P \kappa \rho \frac  {\Phi(r+\lambda)- \Phi(r)} {\lambda}\int_{-\infty}^{z-\log V_T} H^{(r+\lambda)}(y; \Phi(r))\diff y
 \Bigg)\\&+(1-\alpha)(1-J^{(r+m, \lambda)}(0;1)) e^{z}-
	\frac {P \rho+ p} {r+m}(1-J^{(r+m, \lambda)}(0;0)).
\end{align*}



\end{proposition}


The proof of the following results are given in Appendices \ref{proof_prop_E_der_B_negative} and \ref{proof_prop_E_derivative_x}.
\begin{proposition} \label{prop_E_der_B_negative}
	Suppose $V_B > V_B^* \geq  0$. 
	We have
$\frac \partial {\partial V_B} \mathcal{E}(V;V_B) < 0$ for $V > V_B$.
Hence, $\mathcal{E}(V; V_B) < \mathcal{E}(V; V_B^*)$ for all $V > V_B$.
\end{proposition}

\begin{proposition} \label{prop_E_derivative_x}
For $V > V_B > 0$,
 we have 
\begin{align*}
\frac \partial {\partial V} \mathcal{E}(V;V_B)&=1- \frac {V_B}{V} \Big[ \frac \partial {\partial V_B} \mathcal{E}(V;V_B) +\alpha   J^{(r, \lambda)}\Big(\log \frac V {V_B};1 \Big) + (1-\alpha) J^{(r+m, \lambda)}\Big(\log \frac V {V_B};1\Big) \Big]  \\ &+ \frac{P \kappa \rho}{V} R^{(r,\lambda)}\Big(\log \frac V {V_B},\log \frac {V_T} {V_B} \Big),
\end{align*}
where $R^{(r,\lambda)}$ is the resolvent density given in \eqref{resol_dens}.
\end{proposition}

\begin{proposition} \label{prop_B_star_feasible} 
We have $\mathcal{E}(V;V_B^*)\geq0$ for all $V\geq V_B^*$ when $V_B^* > 0$ and for all $V > 0$ when $V_B^* = 0$. In other words, $V_B^*$ is feasible.
\end{proposition}
\begin{proof}
(i) Suppose $V_B^*>0$. 
Because $R^{(r,\lambda)}$ is the resolvent density, it is nonnegative.
By this
together with Propositions  \ref{prop_E_der_B_negative} and \ref{prop_E_derivative_x}, 
 for
 $V> V_B^*>0$, 
\begin{align*}
	\frac \partial {\partial V} \mathcal{E}(V;V_B^*)&=1- \frac {V_B^*}{V}  \frac \partial {\partial V_B} \mathcal{E}(V;V_B^*)-\alpha \frac {V_{B}^*} {V} J^{(r, \lambda)}\Big(\log  \frac V {V_B^*};1 \Big) \\ &-(1-\alpha)\frac {V_{B}^*} {V}J^{(r+m, \lambda)} \Big(\log \frac V {V_B^*};1 \Big)+ \frac{P \kappa \rho}{V} R^{(r,\lambda)} \Big(\log \frac V {V_B^*},\log  \frac {V_T} {V_B^*} \Big)\notag\\
&\geq 1- \frac {V_{B}^*} {V} \Big[ \alpha J^{(r, \lambda)}\Big(\log  \frac V {V_B^*};1 \Big) +(1-\alpha) J^{(r+m, \lambda)} \Big(\log \frac V {V_B^*};1 \Big) \Big] \notag\\
	&\geq 1 - \frac {V_{B}^*} {V} \geq 0,
\end{align*}
where the second inequality holds by Remark  \ref{remark_gamma_0}(2).
Applying this and the fact that $\mathcal{E}(V_B^*;V_B^*)= 0$ when $V_B^* > 0$, the claim is immediate.

(ii) For the case $V_B^*=0$ recall that necessarily $V_T=0$, and hence by \eqref{equity_value_0} we obtain that \begin{equation}\label{der_v_0}\frac \partial {\partial V} \mathcal{E}(V;0)=1>0.\end{equation}
Moreover, by Remark \ref{remark_gamma_0}(1) and because $q \mapsto J^{(q, \lambda)}(0;0)$ is non-increasing in view of its probabilistic expression, for $m>0$, 
\[
\frac{{r}}{\lambda +{r}}  \frac{\Phi({r}+\lambda)}{\Phi({r})}=1-J^{(r, \lambda)}(0;0) \leq 1- J^{(r+m, \lambda)}(0;0)=\frac{{r+m}}{\lambda +{r}+m}  \frac{\Phi({r}+\lambda+m)}{\Phi({r+m})}.
\]	
By this and recalling inequality (\ref{newref}),
we have that
\begin{align*}
\left(\frac{P \kappa \rho}{r}-\frac {P \rho+ p}{r+m}\right) \frac{{r+m}}{\lambda +{r}+m}  \frac{\Phi({r}+\lambda+m)}{\Phi({r+m})}\geq  \frac{P \kappa \rho}{\lambda +r}  \frac{\Phi(r+\lambda)}{\Phi(r)}-\frac {P \rho+ p}{\lambda +r+m}  \frac{\Phi(r+m+\lambda)}{\Phi(r+m)} \geq 0.
\end{align*}
Hence, $\frac{P \kappa \rho}{r}-\frac {P \rho+ p}{r+m} \geq 0$ and therefore
\begin{align}\label{val_e_v_0}
\lim_{V \downarrow 0}\mathcal{E}(V;0)=\frac{P \kappa \rho}{r}-\frac{ P \rho+ p}{r+m}\geq0.
\end{align}
Using \eqref{der_v_0} together with \eqref{val_e_v_0} completes the proof.
\qed \end{proof}
{\it Proof of Theorem \ref{theorem_main}}\\
Now, by Propositions \ref{feasibility_B_less_star}, \ref{prop_E_der_B_negative}
and  \ref{prop_B_star_feasible}, the proof of Theorem \ref{theorem_main} is complete.
\qed

\begin{remark} \label{remark_lambda_infty}
Intuitively, as $\lambda \rightarrow \infty$, the optimal barrier is expected to converge to that in the classical case as in \cite{Kyprianou}. 
In order to confirm this assertion, we provide the following result; its proof is deferred to Appendix \ref{appendix_limit_lambda_ot}.
\begin{lemma}\label{limit_lambda_ot}
Suppose $V_T > 0$ and let $V_B\geq0$ be fixed.  We have
\begin{align}
		\lim_{\lambda\to\infty}&\frac{\lambda+r+m }{\Phi(\lambda+r+m)}	\mathcal{E}(V_B;V_B) \notag\\
		&= V_B \left[\alpha\frac{\psi(1)-r}{1-\Phi(r)}+(1-\alpha)\frac{\psi(1)-(r+m)}{1-\Phi(r+m)}\right] +   \frac{P \kappa \rho}{\Phi(r)}\left[\left(\frac {V_B}  {V_T} \right)^{\Phi(r)}\wedge 1\right] - \frac {P \rho+ p}{\Phi(r+m)}.\label{eq_limit_lambda_ot}
\end{align}
\end{lemma}
This is consistent with identity (3.26) in \cite{Kyprianou}, where the optimal bankruptcy level is such that the right-hand side of \eqref{eq_limit_lambda_ot} vanishes.
\end{remark}		

\section{Two-stage problem} \label{section_two_stage}

We now obtain the optimal leverage by solving the \emph{two-stage problem} as studied by \cite{Chen_Kou_2009,Leland94, Leland96} where the final goal is to choose $P$ that maximizes the firm's value $\mathcal{V}$. For fixed $V > 0$, the problem is formulated as
\begin{align}
\max_{P} \mathcal{V}(V; V_B^*(P), P) \label{two-stage-problem}
\end{align}
where we emphasize the dependency of $\mathcal{V}$ and $V_B^*$ on $P$.   

In this two-stage problem, it is worth investigating the shape of  $\mathcal{V}(V; V_B^*(P), P)$ with respect to $P$ to confirm whether it has a unique maximizer.  Chen and Kou \cite{Chen_Kou_2009} verified the concavity in the continuous observation case with a double jump diffusion as the underlying model and the assumption that $V_T=0$.



In this section, we show, in the periodic observation setting, the concavity for the case when $V_T=0$ and the following assumption is satisfied.
\begin{assumption} \label{assump_completely_monotone}
The \lev measure $\overline{\Pi}$ of the dual process $-X$ has a completely monotone density, i.e.\ 
$\overline{\Pi}$ has a density $\pi$ whose $n^{th}$ derivative $\pi^{(n)}$ exists for all $n \geq 1$ and satisfies
\[
(-1)^{n} \pi^{(n)} (x) \geq 0, \quad x > 0.
\]
\end{assumption}
Important examples satisfying Assumption \ref{assump_completely_monotone} include (the spectrally negative versions of) hyperexponential jump diffusion (as a generalization of \cite{Chen_Kou_2009}), variance gamma  process \cite{Madan},  CGMY process \cite{CGMY}, as well as meromorphic \lev process \cite{Kuznetsov2}.

To show this claim, first we show the following property.
\begin{lemma}\label{fun_H}  
Under Assumption \ref{assump_completely_monotone}, the mapping $x\mapsto H^{(r)}(x;\Phi(r+\lambda))$ is decreasing. 
\end{lemma}
\begin{proof}
	For the completely monotone case, it is known as in Theorem 2 of \cite{Loeffen} that the scale function admits the form
	\[
	W^{(r)}(x) = \Phi'(r) e^{\Phi(r) x} - \int_0^\infty e^{-xt} \mu^{(r)} (\diff t), \quad x \geq 0,
	\]
for some finite measure $\mu^{(r)}$.  
Substituting this and using Fubini's theorem, 
	\begin{align*} 
		Z^{({r})}(x; \Phi(r+\lambda) ) 
		&= e^{\Phi(r+\lambda) x} \left( 1 -\lambda\int_0^{x} e^{-\Phi(r+\lambda)  z} \Big[\Phi'(r) e^{\Phi(r) z} - \int_0^\infty e^{- zt} \mu^{(r)} (\diff t)\Big] \diff z	\right) \\
		&= e^{\Phi(r+\lambda) x} \left( 1 -\lambda  \Big[\Phi'(r) \frac {1-e^{-(\Phi(r+\lambda)-\Phi(r)) x}  } {\Phi(r+\lambda)-\Phi(r)} - \int_0^\infty  \int_0^{x}  e^{- z(t+\Phi(r+\lambda))} \diff z \mu^{(r)}   (\diff t) \Big]	\right) \\
		&=  e^{\Phi(r+\lambda) x}  -\lambda  \Big[\Phi'(r) \frac { e^{\Phi(r+\lambda) x} -e^{\Phi(r) x}  } {\Phi(r+\lambda)-\Phi(r)} - \int_0^\infty  \frac { e^{\Phi(r+\lambda) x}  - e^{- tx}} {t + \Phi(r+\lambda)}  \mu^{(r)}   (\diff t) \Big]	. 
	\end{align*}
	Now, substituting the above expressions in \eqref{der_gamma_0_0}, we have
	\begin{align*}
		H^{(r)}(x;  \Phi(r+\lambda) ) &=Z^{({r})}(x; \Phi(r+\lambda) )-\frac{\lambda}{ \Phi(r+\lambda) -\Phi(r)}W^{({r})}(x) \\
		&= e^{ \Phi(r+\lambda)  x}  -\lambda \Big[\Phi'(r) \frac {e^{ \Phi(r+\lambda)  x}-e^{\Phi(r) x}   } { \Phi(r+\lambda) - \Phi(r) } - \int_0^\infty  \frac { e^{ \Phi(r+\lambda)  x}  - e^{- xt}} {t +  \Phi(r+\lambda) }  \mu^{(r)}  (\diff t) \Big]	  \\
		&- \frac{\lambda}{ \Phi(r+\lambda) -\Phi(r)} \Big[ \Phi'(r) e^{\Phi(r) x} - \int_0^\infty e^{-tx} \mu^{(r)} (\diff t) \Big] \\
		&=  e^{ \Phi(r+\lambda)  x}A + B(x)
	\end{align*}
	where 
	\begin{align*}
	A &:= 1 -  \frac {  \lambda \Phi'(r) } {\Phi(r+\lambda) -\Phi(r)} + \int_0^\infty  \frac {\lambda} {t +  \Phi(r+\lambda) }  \mu^{(r)}   (\diff t), \\
B(x) &:= \lambda\int_0^\infty    e^{- xt}\left[\frac{1}{ \Phi(r+\lambda) -\Phi(r)} - \frac {1} {t +  \Phi(r+\lambda) }  \right] \mu^{(r)} (\diff t).
	\end{align*}
	Because $\lim_{x \rightarrow \infty} B(x) = 0$ by monotone convergence and $\lim_{x\to\infty}H^{(r)}(x;  \Phi(r+\lambda) ) =0$ in view of the probabilistic expression \eqref{der_gamma_0_0}, we must have that $A = 0$. Hence, $H^{(r)}(x;  \Phi(r+\lambda) )  = B(x)$ and its derivative becomes
	\begin{align*}
		\frac \partial {\partial x}H^{(r)}(x;  \Phi(r+\lambda) )   = B'(x)
		&=-\lambda\int_0^\infty   t e^{- xt}\left[\frac{1}{ \Phi(r+\lambda) -\Phi(r)} - \frac {1} {t +  \Phi(r+\lambda) }  \right] \mu^{(r)} (\diff t)<0,
	\end{align*}
where the negativity holds because $\Phi(r+\lambda) > \Phi(r) > 0$ and hence the integrand is always positive. This shows the claim.
\end{proof}	
Now suppose $V_T=0$ so that 
\begin{equation*}
\mathcal{V}(V;V_B, P)=V+\E\left[\int_0^{T_{V_B}^-}e^{-rt}P\kappa\rho dt\right]-\alpha\E\left[e^{-rT_{V_B}^-}
V_{T_{V_B}^-}\mathbf{1}_{\{T_{V_B}^-<\infty\}}\right].
\end{equation*}
By Proposition \ref{lambdaidentified} and identity \eqref{E_B_B}, the optimal barrier $V_B^*(P)$ is given by the root of $\mathcal{E}(V_B;V_B, P) = 0$ for the case $V_B^*(P) > 0$ where
\begin{align} \label{E_B_B_zero_case}
	\begin{split}
\mathcal{E}(V_B;V_B, P) &= V_B + \frac{P \kappa \rho}{r} (1- J^{(r, \lambda)}(0;0) ) \\ &\qquad-  \alpha  V_B J^{(r, \lambda)}(0; 1) - \frac {P \rho+ p} {r+m}  (1- J^{(r+m, \lambda)}(0;0) ) - \left(1-\alpha  \right) V_B J^{(r+m, \lambda)}(0; 1).
	\end{split}
\end{align}
Recall that by \eqref{newref} and $p = mP$,
\begin{align*}
V_B^* (P)= 0 \Longleftrightarrow \lim_{V_B \downarrow 0}\mathcal{E}(V_B;V_B, P) \geq 0 \Longleftrightarrow   \frac{\kappa \rho}{\lambda +r}  \frac{\Phi(r+\lambda)}{\Phi(r)}-\frac {\rho+ m}{\lambda +r+m}  \frac{\Phi(r+m+\lambda)}{\Phi(r+m)} \geq 0,
 \end{align*} 
which does not depend on the value of $P$. Hence, the criterion for $V_B^*(P) = 0$ is irrelevant to the selection of $P$.

(1) First consider the case $\kappa \rho \frac{1}{\lambda +r}  \frac{\Phi(r+\lambda)}{\Phi(r)}-\frac {\rho+ m}{\lambda +r+m}  \frac{\Phi(r+m+\lambda)}{\Phi(r+m)} \geq 0$ so that $V_B^*(P) = 0$ for any choice of $P > 0$. In this case, 
\[
\mathcal{V}(V;V_B^*(P), P) = \mathcal{V}(V;0,P)  = V + \frac {P \kappa \rho} r,
\]
which is linear (and hence concave) in $P$.

(2) Suppose $\kappa \rho \frac{1}{\lambda +r}  \frac{\Phi(r+\lambda)}{\Phi(r)}-\frac {\rho+ m}{\lambda +r+m}  \frac{\Phi(r+m+\lambda)}{\Phi(r+m)} < 0$ so that $V_B^*(P) > 0$ is irrelevant to the selection of $P$. Because $p=Pm$, by solving $\mathcal{E}(V_B;V_B, P) = 0$ with \eqref{E_B_B_zero_case},
\begin{align*}
V_B^*(P) &=-\frac{\displaystyle\frac{P \kappa \rho}{r} (1- J^{(r, \lambda)}(0;0) )-\frac {P \rho+ p} {r+m}  (1- J^{(r+m, \lambda)}(0;0) )}{\displaystyle1-\alpha  J^{(r, \lambda)}(0; 1)-\left(1-\alpha  \right)  J^{(r+m, \lambda)}(0; 1)} 
=\varepsilon P,
\end{align*}
where
\begin{align*}
\varepsilon := -\frac{\displaystyle\frac{ \kappa \rho}{r} (1- J^{(r, \lambda)}(0;0) )-\frac { \rho+ m} {r+m}  (1- J^{(r+m, \lambda)}(0;0) )}{\displaystyle1-\alpha  J^{(r, \lambda)}(0; 1)-\left(1-\alpha  \right)  J^{(r+m, \lambda)}(0; 1)} > 0. 
\end{align*}

Now, as in  \eqref{value_debt_asset} and Proposition \ref{lambdaidentified}, the firm's value is given by
\begin{align*}
\mathcal{V}(V;V_B^*(P),P)&=V+\frac{P \kappa \rho}{r} \left(1- J^{(r, \lambda)}\left(\log \frac{V}{V_B^*(P)};0\right) \right)-\alpha  V_B^*(P) J^{(r, \lambda)}\left(\log \frac{V}{V_B^*(P)};1\right)\\
&=V+\frac{P \kappa \rho}{r} \left(1- J^{(r, \lambda)}\left(\log \frac{V}{\varepsilon P};0\right)\right)-\alpha  \varepsilon P J^{(r, \lambda)}\left(\log \frac{V}{\varepsilon P};1\right).
\end{align*}
Differentiating the above expression and using Lemmas \ref{cztery} and \ref{pochodne} (in the appendix), we have 
\begin{align}\label{der}
\begin{split}
\frac{\partial}{\partial P}\mathcal{V}(V;V_B^*(P),P)&=\frac{\kappa \rho}{r} \left(1- J^{(r, \lambda)}\left(\log \frac{V}{\varepsilon P};0\right)\right)
\\
&-\frac{\kappa \rho}{\lambda+r}\frac{\Phi(r+\lambda)-\Phi(r)}{\Phi(r)}\Phi(r+\lambda)H^{(r)}\left(\log \frac{V}{\varepsilon P};\Phi(r+\lambda)\right)\\
&-\alpha  \varepsilon\frac{\psi(1)-r}{\lambda+r-\psi(1)}\frac{\Phi(r+\lambda)-\Phi(r)}{1-\Phi(r)}(\Phi(r+\lambda)-1)H^{(r)}\left(\log \frac{V}{\varepsilon P};\Phi(r+\lambda)\right).
\end{split}
\end{align}
Here by the convexity of $\psi$ on $[0,\infty)$, the coefficient $\frac{\psi(1)-r}{\lambda+r-\psi(1)}\frac{\Phi(r+\lambda)-\Phi(r)}{1-\Phi(r)}(\Phi(r+\lambda)-1)$ is positive. 

First, the mapping $x \mapsto J^{(r, \lambda)}(x;0)=\E_x [ e^{-r\tilde{T}_0^-}\mathbf{1}_{\{\tilde{T}_0^-<\infty\}} ] = \E [ e^{-r\tilde{T}_{-x}^-}\mathbf{1}_{\{\tilde{T}_{-x}^-<\infty\}} ]$ is decreasing, because $\tilde{T}_{-x}^-$ is increasing in $x$.
On the other hand, Lemma \ref{fun_H} 
shows that the mapping 
$x\mapsto H^{(r)}(x;\Phi(r+\lambda))$ is decreasing as well. 

Using these facts together with \eqref{der} we can conclude that $\frac{\partial }{\partial P}\mathcal{V}(V;V_B^*(P),P)$ is decreasing in $P$, and therefore that the firm's value $\mathcal{V}(V;V_B^*(P),P)$ is a concave function of $P$. In summary, we have the following.
\begin{theorem} \label{theorem_concavity_P}
Suppose $V_T = 0$ and Assumption \ref{assump_completely_monotone} is satisfied.
\\
(1) If $\kappa \rho \frac{1}{\lambda +r}  \frac{\Phi(r+\lambda)}{\Phi(r)}-\frac {\rho+ m}{\lambda +r+m}  \frac{\Phi(r+m+\lambda)}{\Phi(r+m)} \geq 0$, then $V_B^*(P) = 0$ for all $P > 0$ and we have  $\mathcal{V}(V; V_B^*(P),P) = V + \frac {P \kappa \rho} r$.
\\
(2) Otherwise,  $V_B^*(P)=\varepsilon P > 0$ for all $P > 0$ and $\mathcal{V}(V; V_B^*(P),P)$ is concave in $P$ for any $V > 0$.
\end{theorem}

\section{Numerical Examples} \label{section_numerics}


In this section, we confirm the analytical results obtained in the previous sections through a sequence of numerical examples.  In addition, we study numerically the impact of the rate of observation $\lambda$ on the optimal solutions, obtain the optimal leverage by considering the two-stage problem considered in \eqref{two-stage-problem}, 
and analyze the behaviors of credit spreads.

Throughout this section, we use $r=7.5\%$, $\delta=7\%$, $\kappa = 35\%$, $\alpha=50\%$ for the parameters of the problem as used in \cite{Hilberink, Kyprianou, Leland94, Leland96}.  Additionally, unless stated otherwise, we set $\rho = 8.162\%$ and $m=0.2$, which were used in \cite{Chen_Kou_2009}, $P=50$, and $\lambda = 4$ (on average four times per year). For the tax threshold,  we set
\begin{align}
V_T = P\rho/\delta \label{V_T_P}
\end{align}
 as used in \cite{Kyprianou} and also suggested by \cite{Hilberink, Leland96}. By the choice \eqref{V_T_P}, necessarily $V_T > 0$ and hence $V_B^* > 0$ as discussed in Section \ref{section_existence}.

For the process $(X_t)_{t \geq 0}$, we use a mixture of Brownian motion and a compound Poisson process with i.i.d.\ hyperexponential jumps:
$X_t=\mu t +\sigma B_t -\sum_{i=1}^{N_t}U_i$, $t \geq 0$, where $(B_t)_{t \geq 0}$ is a standard Brownian motion, $(N_t)_{t \geq 0}$ is a Poisson process with intensity $\gamma$ and $(U_i)_{i\geq 1}$ takes an exponential random variable with rate $\beta_i > 0$ with probability $p_i$ for $1 \leq i \leq m$, such that $\sum_{i=1}^m p_i = 1$.
Note that this satisfies the completely monotone condition given in Assumption \ref{assump_completely_monotone}.
The corresponding Laplace exponent \eqref{laplace_exponent} then becomes
\begin{align*}
\psi(s) = \mu s + \frac 1 2 \sigma^2 s^2  + \gamma \sum_{i=1}^m p_i
\left( \frac {\beta_i} {\beta_i + s}  -1\right), \quad s \geq 0.
\end{align*}
This is a special case of the phase-type \lev process \cite{Asmussen} and its scale function  has an explicit expression written as a sum of exponential functions; see e.g.\ \cite{Egami_Yamazaki_2010_2, KKR}.
In particular, we consider the following two parameter sets:
\begin{description}
\item[\textbf{Case A} (without jumps):]  $\sigma = 0.2$, $\mu = -0.015$, $\gamma = 0$;
\item[\textbf{Case B} (with jumps):] $\sigma = 0.2$, $\mu = 0.055$, $\gamma = 0.5$, $(p_1, p_2) = (0.9,0.1)$, and $(\beta_1, \beta_2) = (9,1)$.
\end{description}
Here, $\mu$ is chosen so that the martingale property $\psi(1)=r-\delta = 0.005$ is satisfied.  In \textbf{Case B}, the jump size $U$ models both small and large jumps (with parameters $9$ and $1$) that occur with probabilities $0.9$  and $0.1$, respectively.


 \subsection{Optimality}\label{subsec:optimal} Under the parameter settings described above, we first confirm the optimality of the suggested barrier $V_B^*$ that satisfies $\mathcal{E}(V_B^*;V_B^*) = 0$.  Because the mapping $V_B \mapsto \mathcal{E}(V_B;V_B)$ (given in \eqref{E_B_B}) is monotonically increasing (see Proposition \ref{E_monotonicity}), the value of $V_B^*$ is computed by classical bisection methods.
 The corresponding capital structure is  then computed by \eqref{value_debt_asset} and \eqref{equity_value}.

At the top of Figure \ref{figure_value}, for \textbf{Cases A} and \textbf{B}, we plot  $V \mapsto \mathcal{E}(V; V_B^*)$ along with $V \mapsto \mathcal{E}(V; V_B)$ for $V_B \neq V_B^*$. Here, we confirm Theorem \ref{theorem_main}: the level $V_B^*$ satisfies the limited liability constraint \eqref{constraint}, and any level $V_B$ lower than $V_B^*$ violates \eqref{constraint}, while for $V_B$ larger than $V_B^*$, $\mathcal{E}(V; V_B)$ is dominated by $\mathcal{E}(V; V_B^*)$. The corresponding debt and firm values are also plotted in Figure \ref{figure_value}.

\begin{figure}[htbp]
\begin{center}
\begin{minipage}{1.0\textwidth}
\centering
\begin{tabular}{cc}
\includegraphics[scale=0.5]{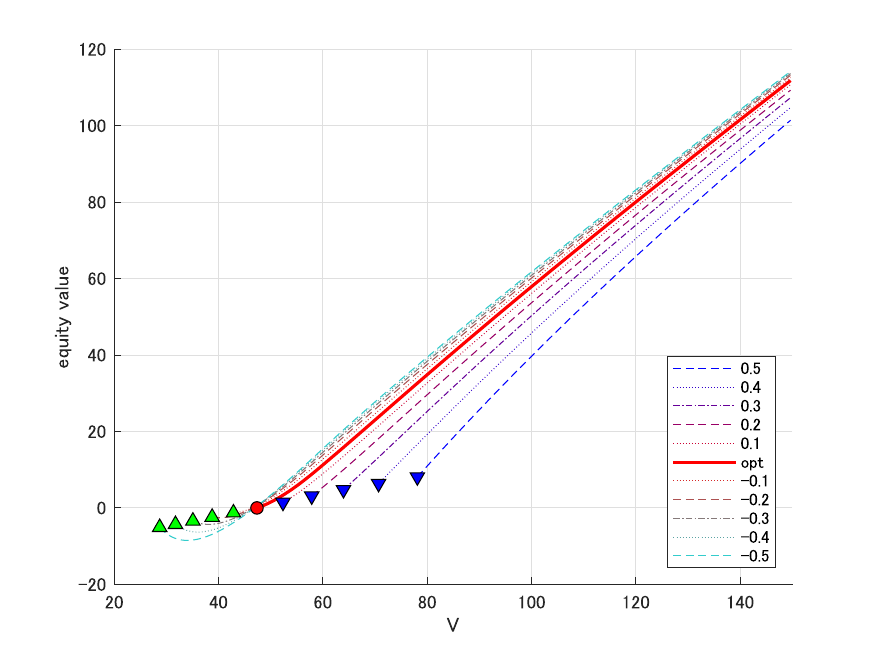} & \includegraphics[scale=0.5]{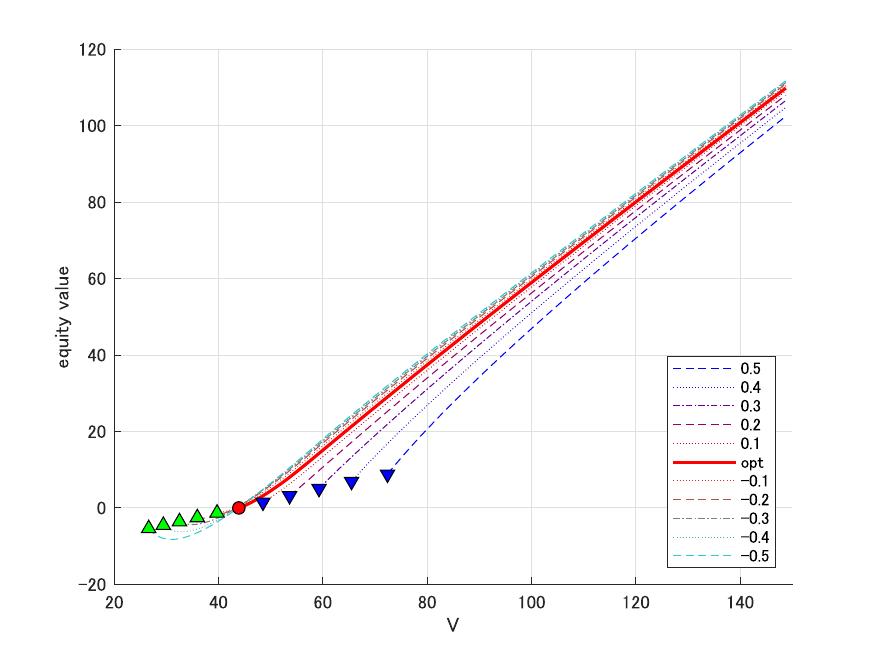} \\
\textbf{Case A}: equity value $V \mapsto \mathcal{E}(V; V_B)$ & \textbf{Case B}: equity value  $V \mapsto \mathcal{E}(V; V_B)$ \\
\includegraphics[scale=0.5]{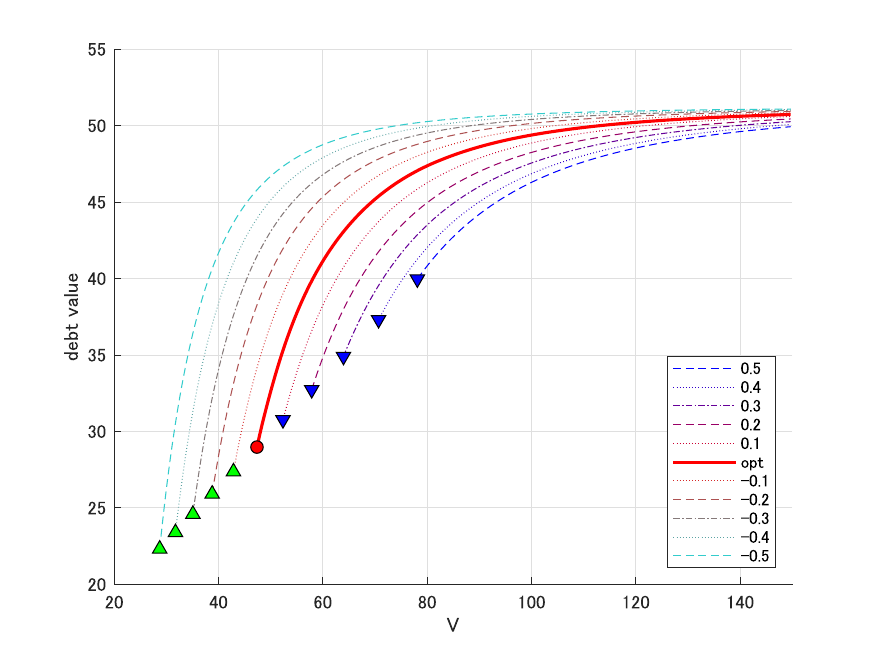} & \includegraphics[scale=0.5]{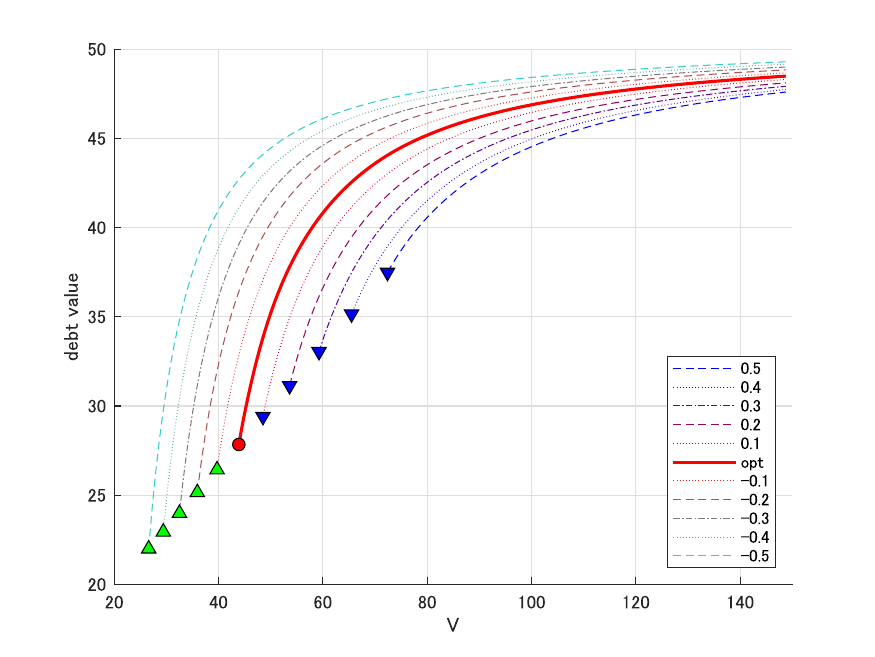} \\
\textbf{Case A}: debt value  $V \mapsto \mathcal{D}(V; V_B)$ & \textbf{Case B}: debt value  $V \mapsto \mathcal{D}(V; V_B)$ \\
\includegraphics[scale=0.5]{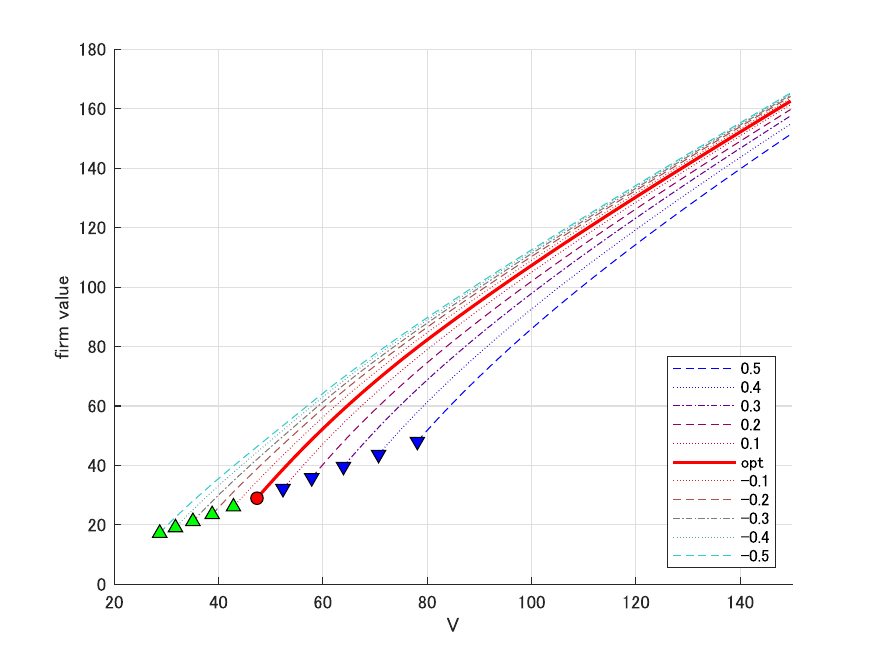} & \includegraphics[scale=0.5]{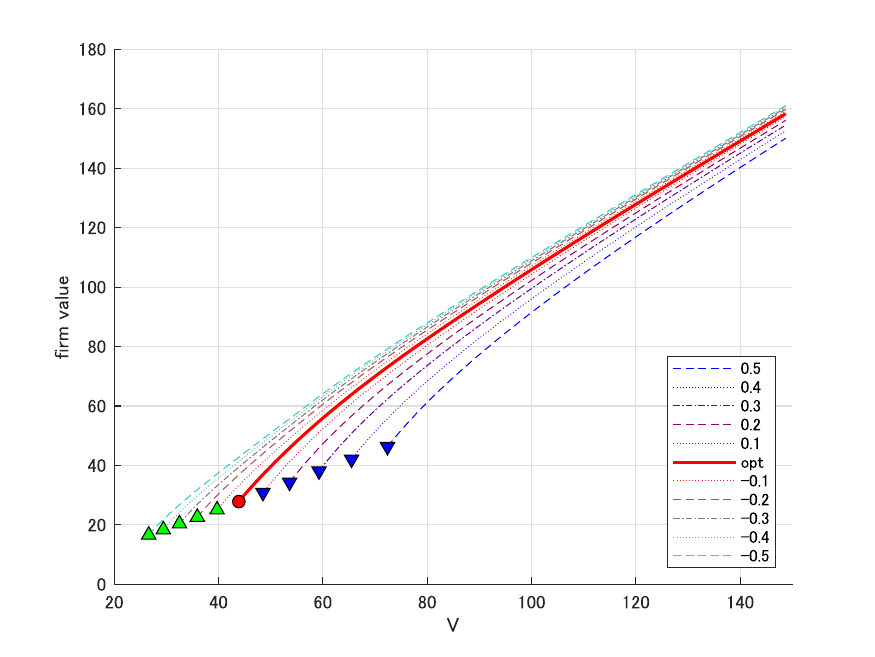} \\
\textbf{Case A}: firm  value  $V \mapsto \mathcal{V}(V; V_B)$ & \textbf{Case B}: firm value  $V \mapsto \mathcal{V}(V; V_B)$
\end{tabular}
\end{minipage}
\caption{\small{The equity/debt/firm values as functions of $V$  on $(V_B, \infty)$ for $V_B=V_B^*$ (solid) along with $V_B = V_B^*  \exp(\epsilon)$ (dotted) for $\epsilon = -0.5,-0.4, \ldots, -0.1, 0.1, 0.2, \ldots, 0.5$.} The values at $V = V_B$ are indicated by circles for $V_B = V_B^*$ whereas those for  $V_B < V_B^*$ (resp.\ $V_B > V_B^*$) are indicated by up (resp.\ down)-pointing triangles.} \label{figure_value}
\end{center}
\end{figure}


\subsection{Sensitivity with respect to $\lambda$ on the equity value} We now proceed to study the sensitivity of the optimal bankruptcy barrier and the equity value with respect to the rate of observation $\lambda$. On the left plot of Figure \ref{fig_r},  we show the equity value $\mathcal{E}(\cdot; V_B^*)$ for various values of $\lambda$ along with the classical (continuous-observation) case as obtained in \cite{Hilberink, Kyprianou}.  We see that the optimal barrier $V_B^*$ is decreasing in $\lambda$ and converges to the optimal barrier, say $\tilde{V}_B$, of the classical case. This confirms Remark \ref{remark_lambda_infty}.

We also confirm the convergence of $\mathcal{E}(V;V_B^*)$,  to the classical case, say $\tilde{\mathcal{E}}(V;\tilde{V}_B)$, for each starting value $V$. On the other hand, the monotonicity of $\mathcal{E}(V;V_B^*)$ with respect to $\lambda$ fails.  When $V$ is small, the equity value tends to be higher for small values of $\lambda$, but it is not necessarily so for higher values of $V$.  In order to investigate this, we show in the bottom plots of Figure \ref{fig_r}, the difference  $\mathcal{E}(V;V_B^*)- \tilde{\mathcal{E}}(V;\tilde{V}_B)$. We observe also the differences between \textbf{Cases A} and \textbf{B} -- in \textbf{Case A}, a lower value of $\lambda$ clearly achieves higher equity value when $V$ is large whereas this is not clear in \textbf{Case B}.

\subsection{Analysis of the bankruptcy time and the asset value at bankruptcy}


While it was confirmed that the barrier level $V_B^*$ is monotone in $\lambda$,
it is not clear how the distributions of $(T^-_{V_B^*}, V_{T^-_{V_B^*}})$ change in $\lambda$.  Here, by taking advantage of the joint Laplace transform $(q, \theta) \mapsto  J^{(q, \lambda)} (\cdot; \theta)$ as in \eqref{fun_gamma}, we compute numerically the density and distribution of the random variables $T^-_{V_B^*}$ and $V_{T^-_{V_B^*}}$ for each $\lambda$. We also obtain those in the classical case by inverting $(q, \theta) \mapsto H^{(q)} (\cdot; \theta)$ as in \eqref{der_gamma_0_0}.

For Laplace inversion, we adopt the Gaver-Stehfest algorithm, which was suggested to use in Kou and Wang \cite{Kou_Wang}  (see also Kuznetsov \cite{Kuznetsov} for its convergence results).
The algorithm is easy to implement and only requires real values. While a major challenge is to handle the cases involving large numbers, our case can be handled without difficulty  in the standard Matlab environment with double precision.

In our case, the scale function $W^{(q)}$ is written in terms of a linear sum of  $e^{\Phi(q)x}$ and $e^{-\xi_{i,q}x}$, $1 \leq i \leq n$ ($n=1$ in \textbf{Case A} and $n=3$ in \textbf{Case B}), where $\Phi(q)$ is as in \eqref{def_varphi} and $-\xi_{i,q}$ are the negative roots of $\psi(\cdot) = q$.
As in the proof of Lemma \ref{fun_H},
the terms for $e^{\Phi(q) x}$ all cancel out in the Laplace transforms $J^{(q, \lambda)}(\cdot; \theta)$ and $H^{(q)}(\cdot; \theta)$. Hence, the algorithm runs without the need of handling large numbers even for high values of $q$.  The same can be said about the parameter $\theta$.


For the initial value $V = 100$, we plot in Figure \ref{fig_bankruptcy_time_r} the density and distribution functions of $T^-_{V_B^*}$ and in Figure \ref{fig_bankruptcy_value_r} those for $V_{T^-_{V_B^*}}$ for the same parameter sets as used for
 Figure \ref{fig_r} (note that the value of $V_B^*$ depends on $\lambda$). For comparison, those in the classical case (computed by inverting $q, \theta \mapsto H^{(q)}(\log V; \theta)$) are also plotted. It is noted that in Figure \ref{fig_bankruptcy_value_r}, the distribution is not purely diffusive and instead the probability of the event $V_{T^-_{V_B^*}}=V_B^*$ is strictly positive. In particular, for \textbf{Case A}, $V_{T^-_{V_B^*}}=V_B^*$ a.s.
At least in our examples, the distribution functions for $T^-_{V_B^*}$  appear to be monotone in $\lambda$ while they are not for $V_{T^-_{V_B^*}}$.

\begin{figure}[htbp]
\begin{center}
\begin{minipage}{1.0\textwidth}
\centering
\begin{tabular}{cc}
\includegraphics[scale=0.5]{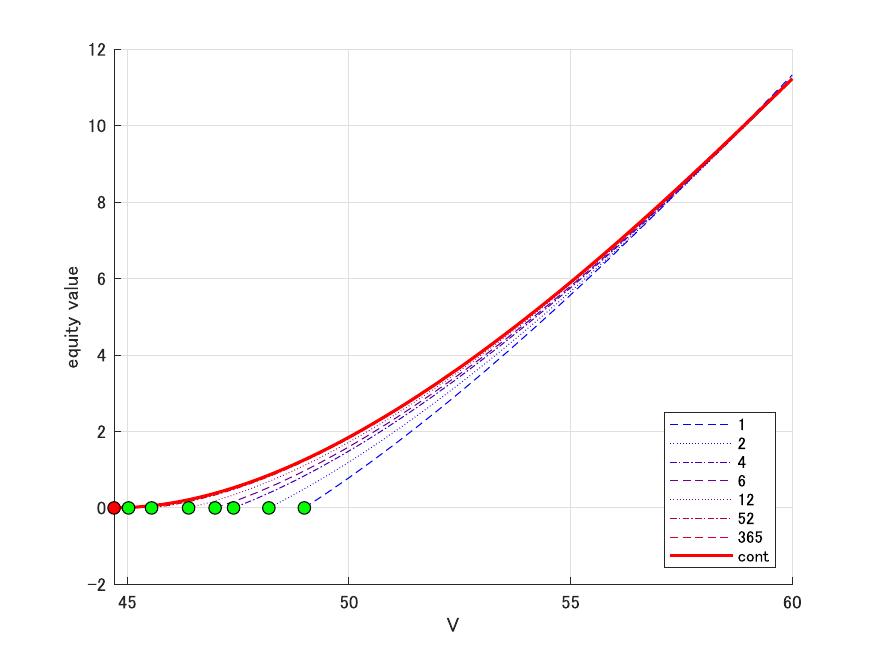} & \includegraphics[scale=0.5]{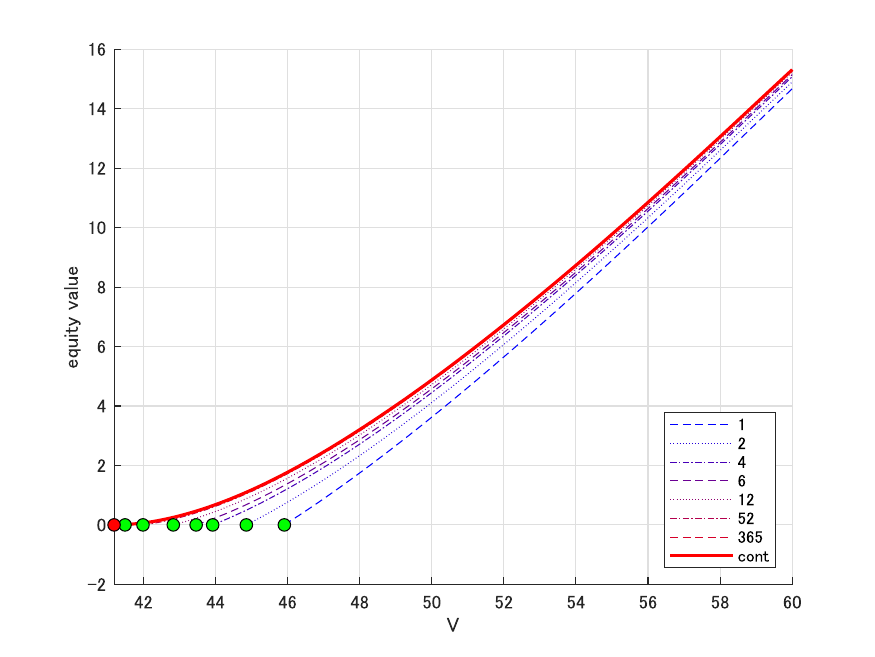} \\
\textbf{Case A: $V \mapsto \mathcal{E}(V; V_B^*)$} & \textbf{Case B: $V \mapsto \mathcal{E}(V; V_B^*)$} \\
\includegraphics[scale=0.5]{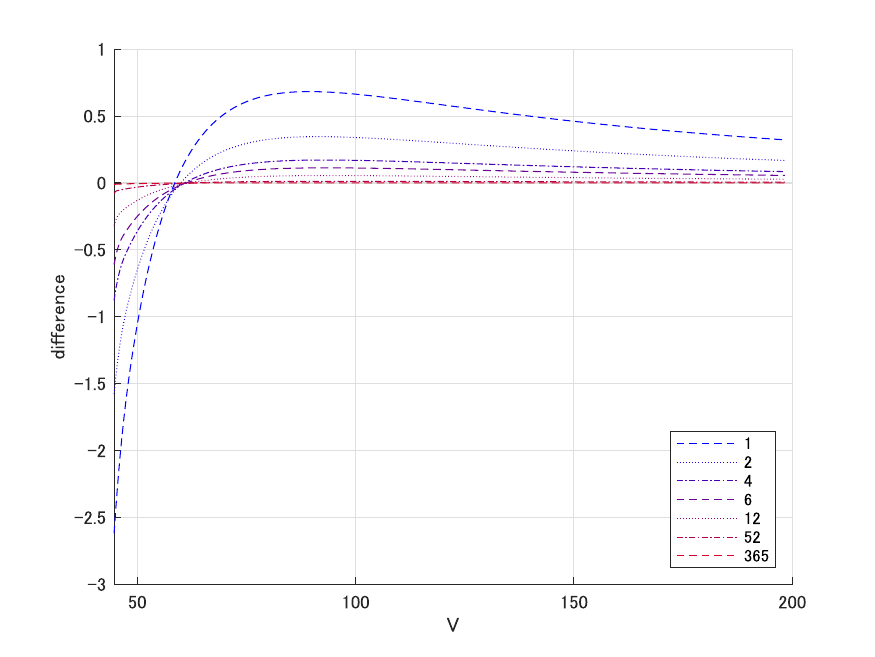}
& \includegraphics[scale=0.5]{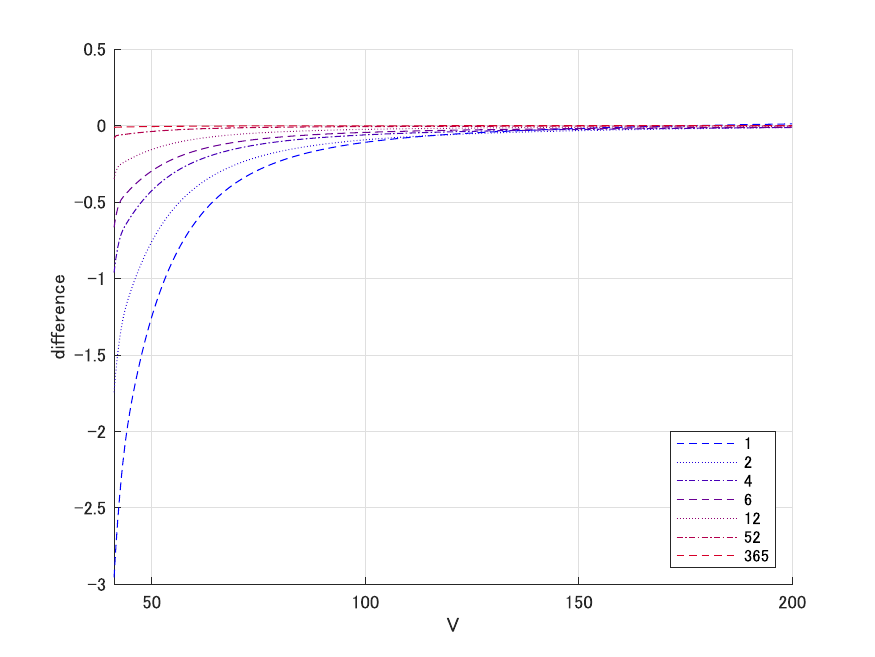} \\
\textbf{Case A: $V \mapsto \mathcal{E}(V; V_B^*) - \tilde{\mathcal{E}}(V; \tilde{V}_B)$} & \textbf{Case B: $V \mapsto \mathcal{E}(V; V_B^*) - \tilde{\mathcal{E}}(V; \tilde{V}_B)$}
\end{tabular}
\end{minipage}
\caption{\small{(Top) The equity values $\mathcal{E}(V; V_B^*)$ (dotted) for $\lambda = 1,2,4,6,12,52,365$ along with the classical case  $\tilde{\mathcal{E}}(V; \tilde{V}_B)$ (solid).  The corresponding values at $V = V_B^*$ are indicated by circles. (Bottom) The difference $\mathcal{E}(V; V_B^*) - \tilde{\mathcal{E}}(V; \tilde{V}_B)$ for the same set of $\lambda$.}} \label{fig_r}
\end{center}
\end{figure}




\subsection{Two-stage problem}

Now we consider the two-stage problem  \eqref{two-stage-problem}. Recall, as confirmed in Theorem \ref{theorem_concavity_P}, that the firm value $\mathcal{V}(V;V_B^*(P), P)$ is concave in $P$ for the case $V_T = 0$.  Here, in order to see if the concavity holds when $V_T > 0$, 
 we continue to use the tax cutoff level $V_T$ by \eqref{V_T_P} as a function of $P$.

For our numerical results, we set $V = 100$ and obtain $V_B^*$ for $P$ running from $0$ to $100$ (leverage $P/V$ running from $0$ to $1$).  The corresponding firm and debt values are computed for each $P$ and $V_B^* = V_B^*(P)$, and is shown in Figure \ref{figure_two_stage}.   For comparison, analogous results on the classical case are also plotted. 
Here, the concavity with respect to $P$ is confirmed in all considered cases.

Regarding the analysis with respect to $\lambda$, at least in these examples, we observe that the firm and debt values for each $P$ are monotone in $\lambda$ and converge to those in the classical case.
In addition, we see that the optimal face value $P^*$ decreases in $\lambda$ and converges to that in the classical case. 

\begin{figure}[htbp]
\begin{center}
\begin{minipage}{1.0\textwidth}
\centering
\begin{tabular}{cc}
\includegraphics[scale=0.5]{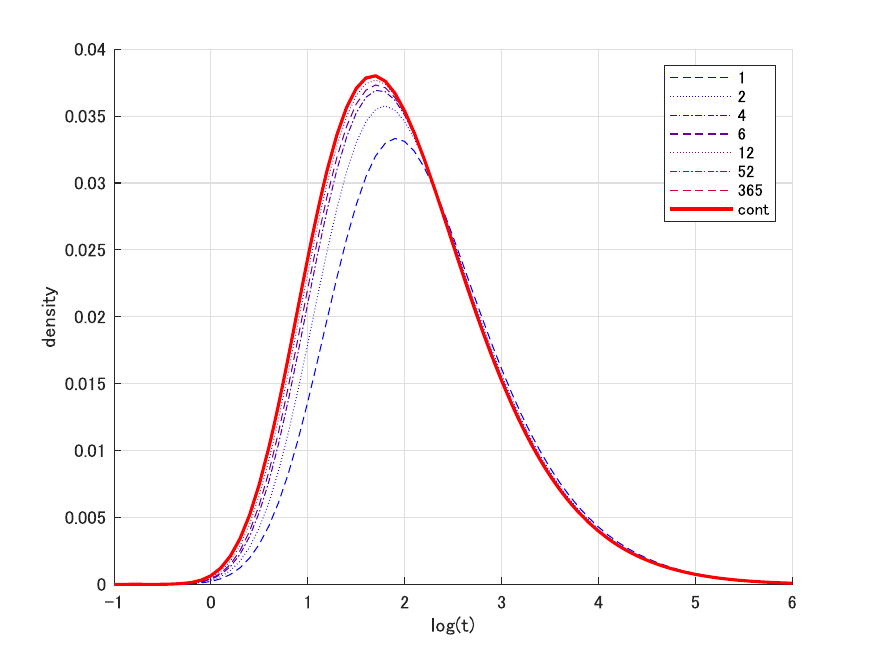} & \includegraphics[scale=0.5]{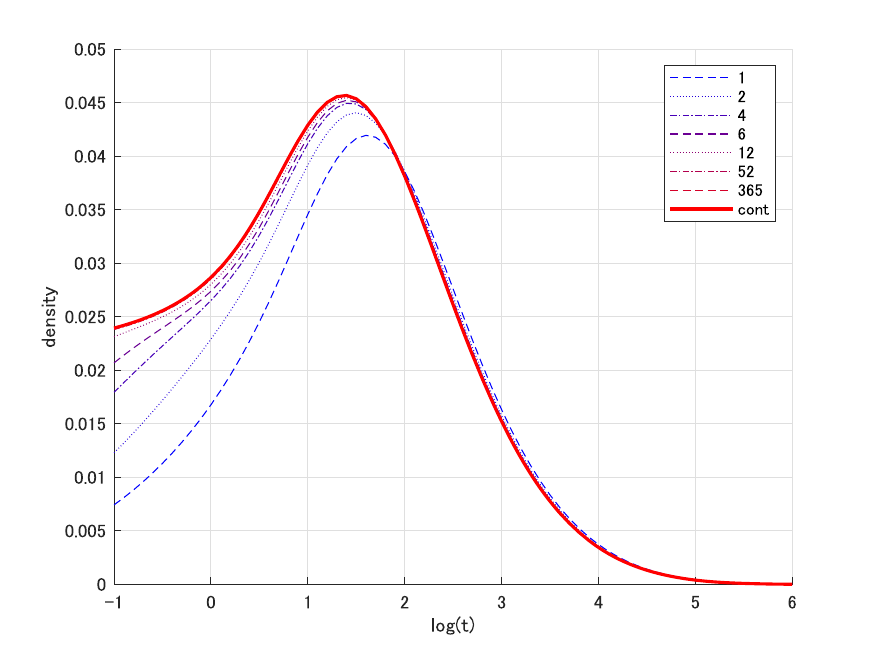} \\
\textbf{Case A: $\p (T_{V_B^*}^-  \in \diff t)/ \diff t$} & \textbf{Case B: $\p (T_{V_B^*}^-  \in \diff t)/ \diff t$}\\
\includegraphics[scale=0.5]{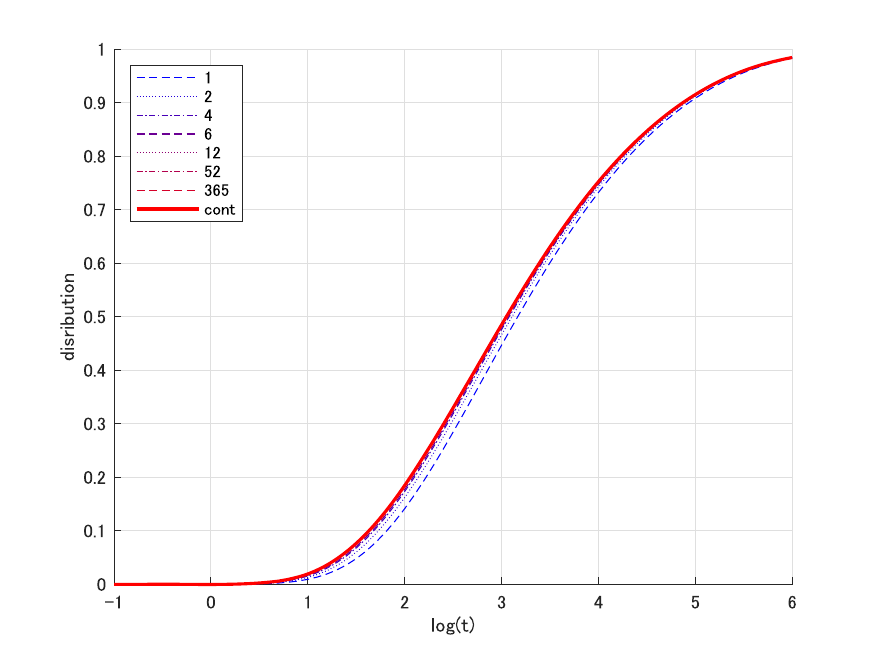}
& \includegraphics[scale=0.5]{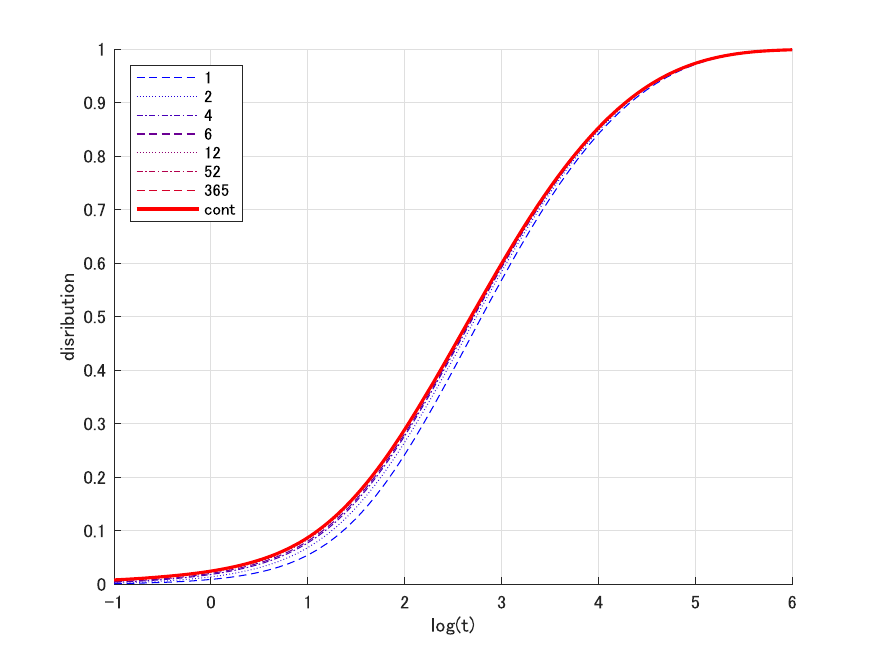} \\
\textbf{Case A: $\p (T_{V_B^*}^- \leq t)$} & \textbf{Case B: $\p (T_{V_B^*}^- \leq t)$}
\end{tabular}
\end{minipage}
\caption{
Density $\p (T_{V_B^*}^-  \in \diff t)/ \diff t$ and distribution $\p (T_{V_B^*}^- \leq t)$ (indicated by dotted lines) for $\lambda = 1,2,4,6,12,52,365$, the initial value $V = 100$, and  $V_B^*$ determined as in Figure \ref{fig_r}. The classical cases are also shown by solid lines. These values are plotted against the logarithm of time.} \label{fig_bankruptcy_time_r}
\end{center}
\end{figure}

\begin{figure}[htbp]
\begin{center}
\begin{minipage}{1.0\textwidth}
\centering
\begin{tabular}{cc}
\includegraphics[scale=0.5]{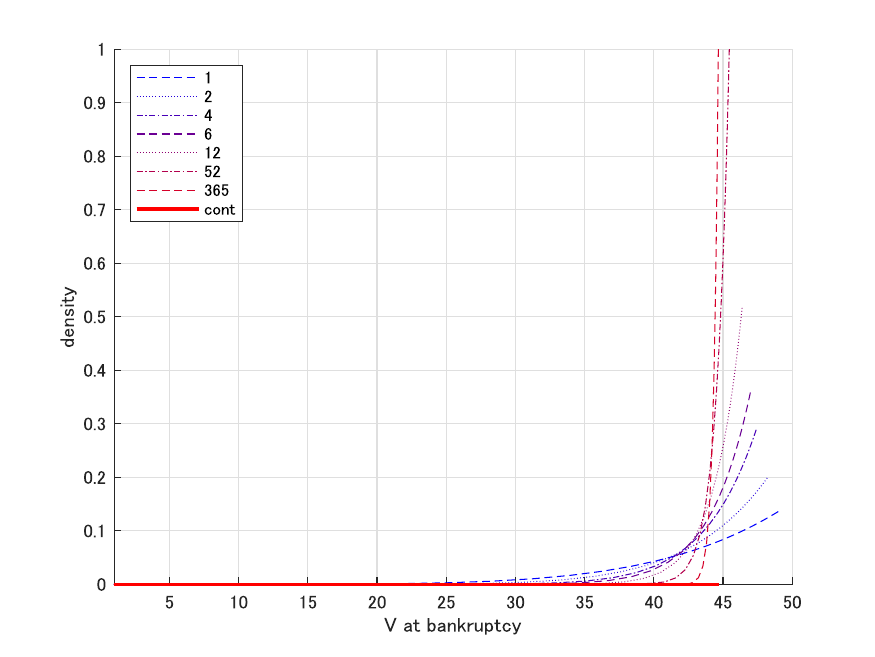} & \includegraphics[scale=0.5]{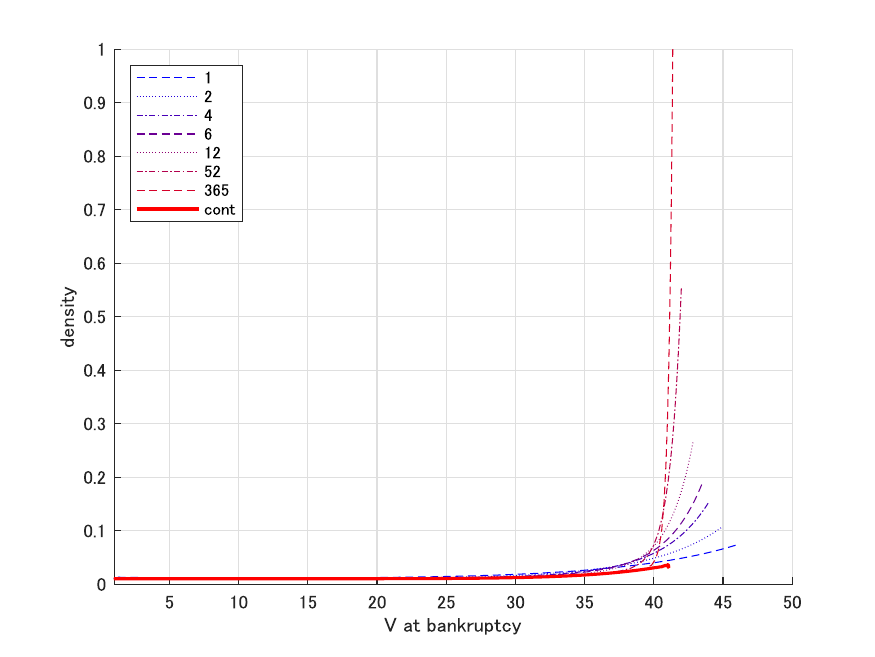} \\
\textbf{Case A:} $\p (V_{T_{V_B^*}^-}  \in \diff v)/ \diff v$ & \textbf{Case B:} $\p (V_{T_{V_B^*}^-}  \in \diff v)/ \diff v$ \\
\includegraphics[scale=0.5]{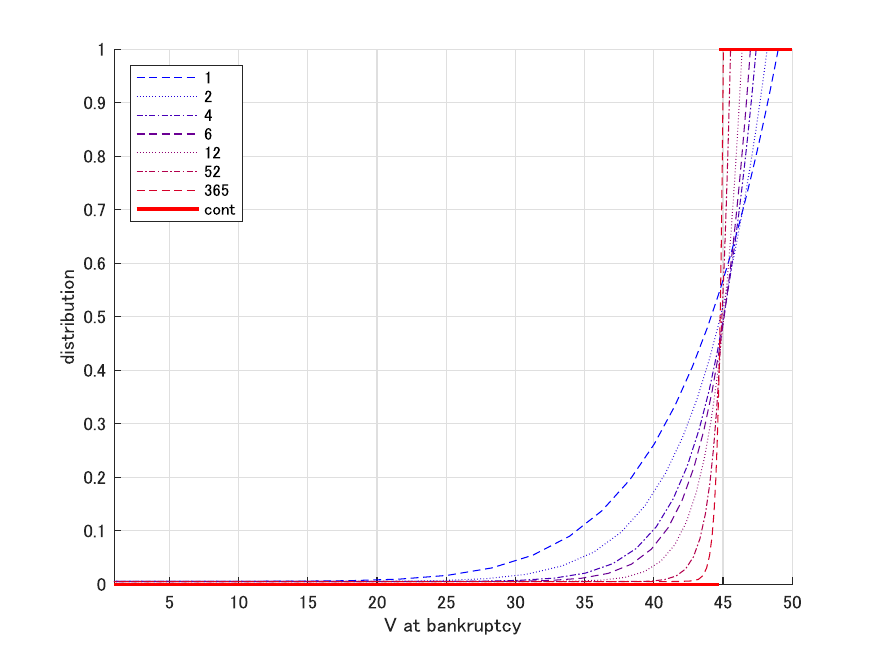}
& \includegraphics[scale=0.5]{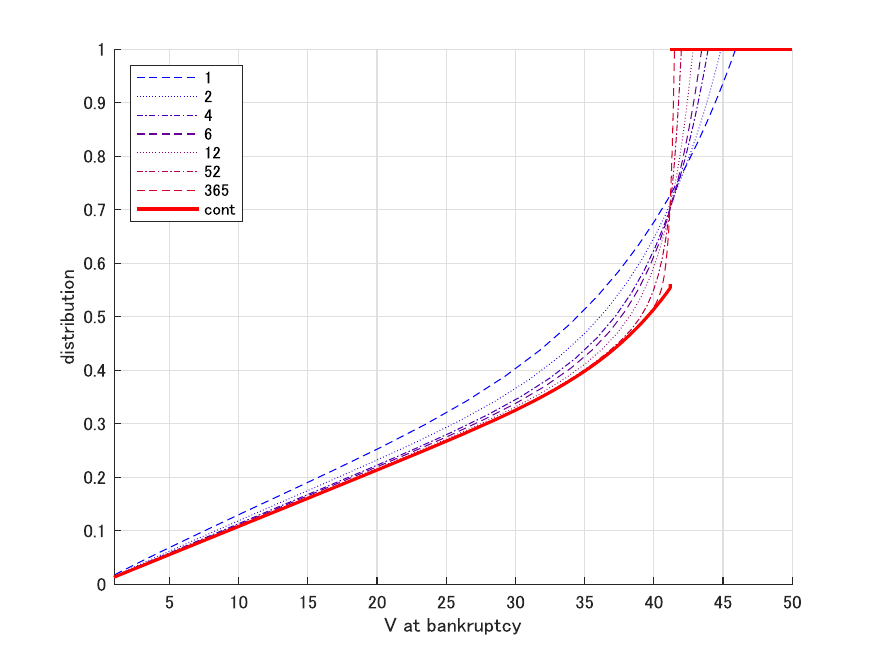} \\
\textbf{Case A:} $\p (V_{T_{V_B^*}^-} \leq v)$ & \textbf{Case B:} $\p (V_{T_{V_B^*}^-} \leq v)$
\end{tabular}
\end{minipage}
\caption{Density $\p (V_{T_{V_B^*}^-}  \in \diff v)/ \diff v$ and distribution $\p (V_{T_{V_B^*}^-} \leq v)$ (indicated by dotted lines) for $\lambda = 1,2,4,6,12,52,365$, the initial value $V = 100$, and  $V_B^*$ determined as in Figure \ref{fig_r}. The classical cases are also shown by solid lines, in which it has a positive mass at the bankruptcy level.} \label{fig_bankruptcy_value_r}
\end{center}
\end{figure}

\begin{figure}[htbp]
\begin{center}
\begin{minipage}{1.0\textwidth}
\centering
\begin{tabular}{cc}
\includegraphics[scale=0.5]{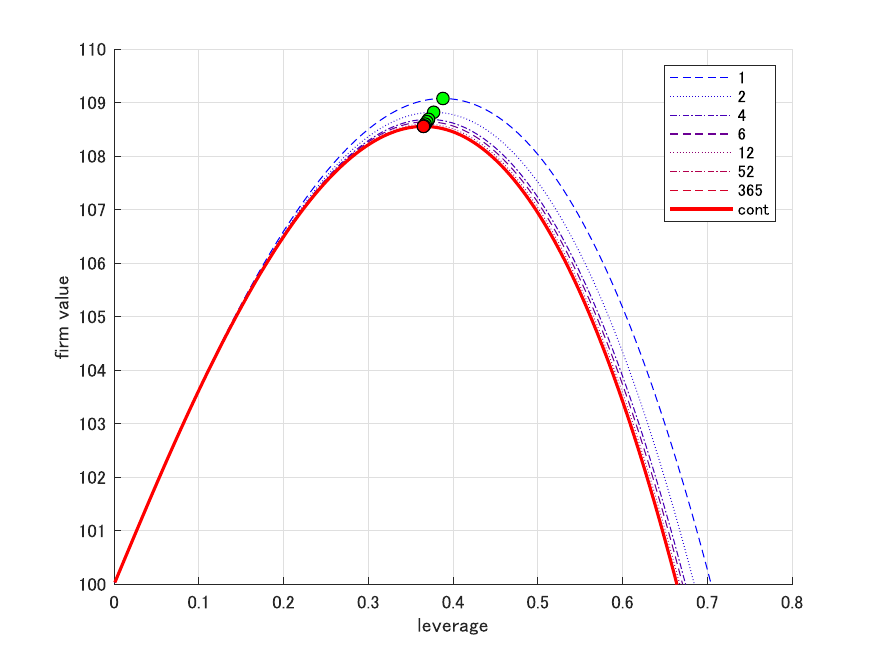}  & \includegraphics[scale=0.5]{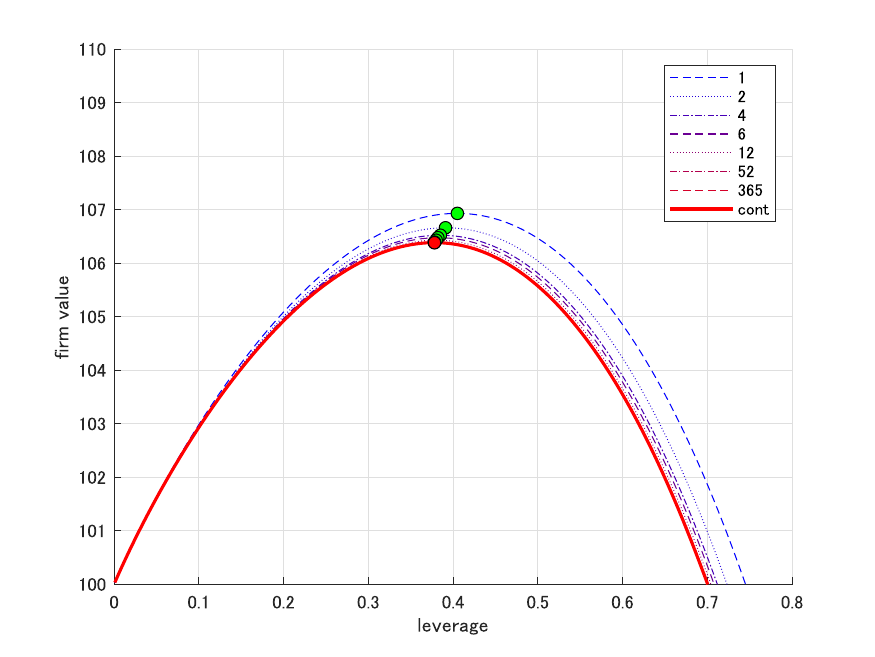}  \\
\textbf{Case A:} firm value  & \textbf{Case B:} firm value  \\
\includegraphics[scale=0.5]{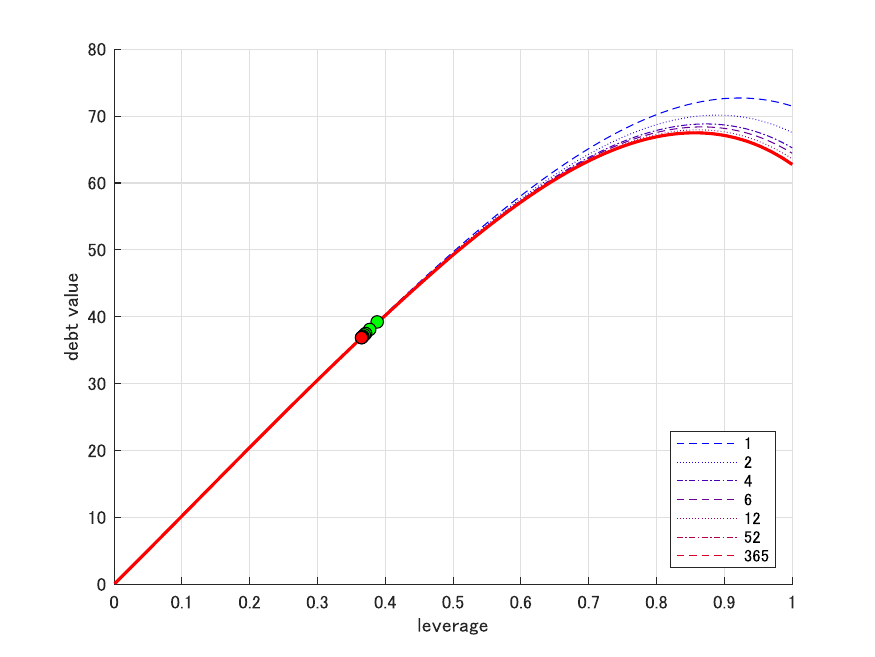}  & \includegraphics[scale=0.5]{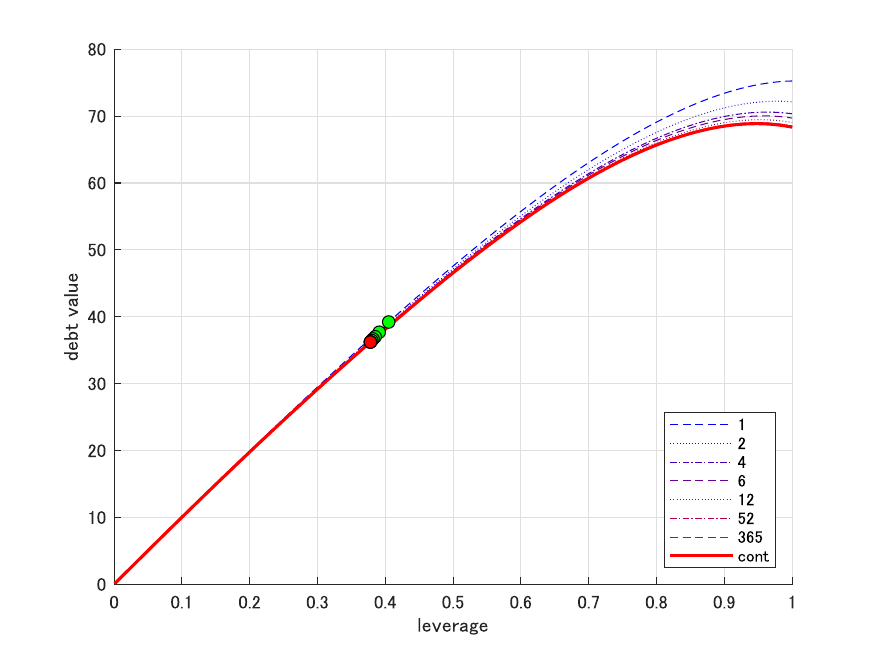}  \\
\textbf{Case A:} debt value  & \textbf{Case B:} debt value
\end{tabular}
\end{minipage}
\caption{\small{The firm values (top) and debt values (bottom) as functions of  the leverage $P/V$ for the two-stage problem for $V = 100$. The periodic cases with $\lambda = 1,2,4,6,12,52,365$ (dotted) are indicated by dotted lines and the classical case corresponds to the solid lines. The points at $P^*/V$ are indicated by the circles.}} \label{figure_two_stage}
\end{center}
\end{figure}


\subsection{The term structure of credit spreads}\label{sec:creditspreads}


We now move onto the analysis of the credit spread.
Let $V_B > 0$ be a fixed bankruptcy level. 
The credit spread is defined as the excess of the amount of coupon over the risk-free interest rate, required to induce the investor to lend one dollar to the firm until maturity time $t$.
To be more precise, by finding the coupon rate $\rho^*$ that makes the value of the debt $d(V;V_B,t)$ defined in (\ref{debt_constant_unit}) of unit face value equal to one, the credit spread $\rho^*-r$ is given after some rearrangement of (\ref{debt_constant_unit}) by
%
\begin{align}\label{eq:spreads}
\mathrm{CS}_{\lambda}(t)=\frac{r}{P}\frac{\mathbb{E}\Big[\big[P-(1-\alpha)V_{T_{V_B^-}} \big]e^{-r T_{V_B}^-}\mathbf{1}_{\{T_{V_B}^- \leq t\}} \Big]}{\mathbb{E} \big[1- e^{-r(t\wedge T_{V_B}^-)}\big]}.
\end{align} 

Before showing numerical results, we prove the following analytical limits.
The proofs are deferred to Appendices \ref{appendix_CS_limit_T_} and \ref{proof_conv_CS_lambda}.
\begin{proposition}\label{CS_limit_T_0}
For $V\not=V_B$, we have
		$\lim_{t \downarrow 0} \mathrm{CS}_{\lambda}(t)=\frac{\lambda}{P}\big[P-(1-\alpha)V\big] \mathbf{1}_{\{V<V_B\}}$. 
\end{proposition}
Let $CS(t)$ denote the credit spread in the classical case
as described in Hilberink and Rogers \cite{Hilberink}.
\begin{proposition}\label{conv_CS_lambda}
For  $V_B > 0$, $V\not=V_B$, and $t > 0$,
we have $\lim_{\lambda\to\infty}CS_{\lambda}(t)=CS(t)$. 
\end{proposition}
\begin{remark}
While theoretically the credit spread vanishes in the limit as in Proposition \ref{CS_limit_T_0}, we will see below that the rate of convergence can be controlled by the selection of $X$ and $\lambda$ and can be made very slow as shown in Figure \ref{plot_credit_spread}.
\end{remark}

 \begin{figure}[htbp]
\begin{center}
\begin{minipage}{1.0\textwidth}
\centering
\begin{tabular}{cc}
  \includegraphics[scale=0.5]{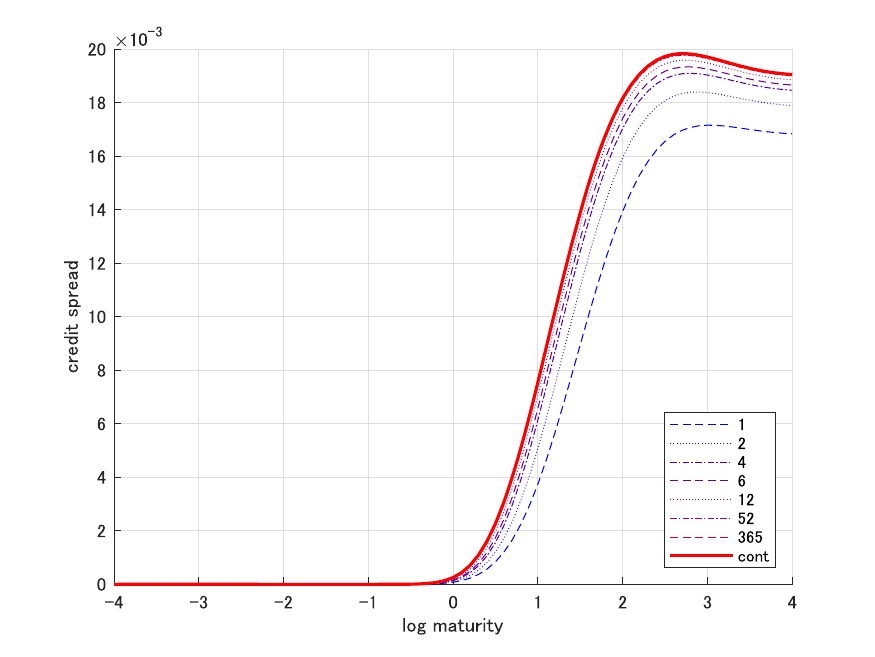} & \includegraphics[scale=0.5]{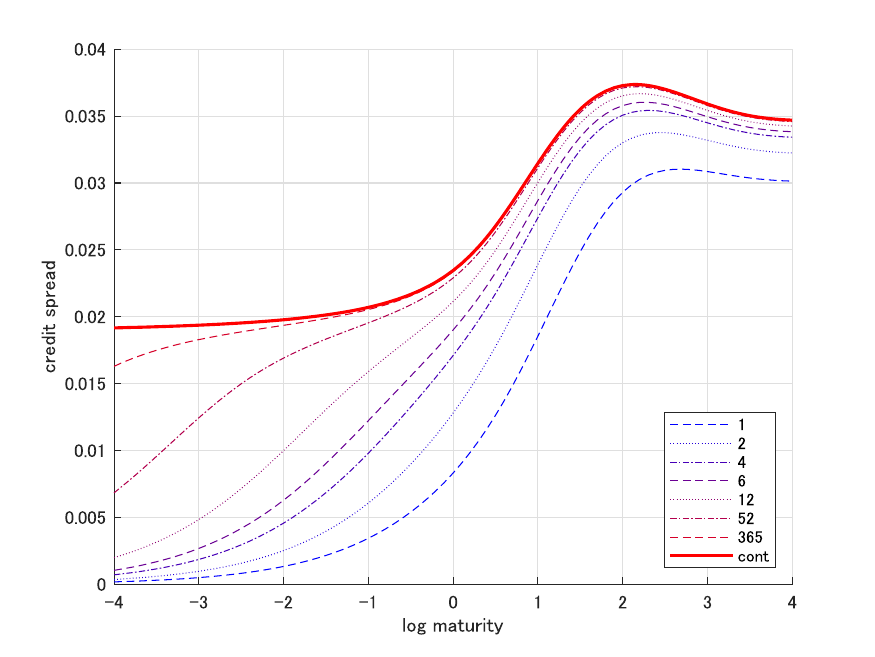} \\
   \textbf{Case A} with $L=50$ & \textbf{Case B} with $L=50$ \\
\includegraphics[scale=0.5]{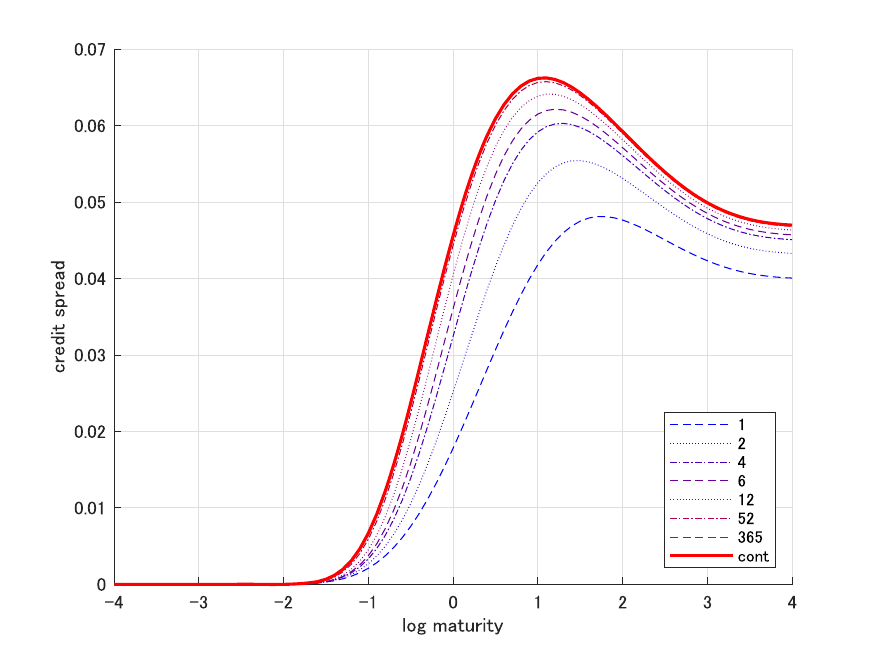} &
 \includegraphics[scale=0.5]{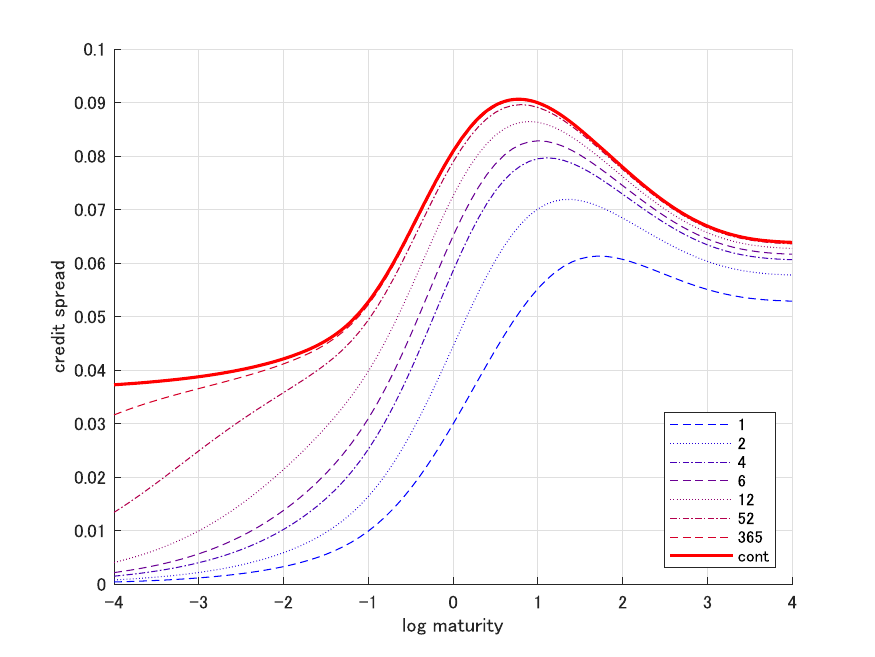}  \\
  \textbf{Case A} with $L=75$ & \textbf{Case B} with $L=75$
\end{tabular}
\caption{Term structure of credit spreads with respect to the logarithm of maturity for $V = 100$. The periodic cases with $\lambda = 1,2,4,6,12,52,365$ (dotted) are indicated by dotted lines and the classical case corresponds to the solid lines.
}  \label{plot_credit_spread}
\end{minipage}
\end{center}
\end{figure}

To compute credit spreads, we follow the procedures for Figure 6 (given in Appendix B) of \cite{Hilberink}.

Fix $V$ and $m$.  The first step is to choose, for a selected leverage $0 \leq L \leq 1$, the face value of debt $\hat{P} \equiv \hat{P}(L)$ and $\hat{\rho} = \hat{\rho}(L)$ satisfying $\mathcal{D}(V; \hat{V}_B^*) \equiv \mathcal{D}(V; \hat{V}_B^*; \hat{P}, \hat{\rho}) = \hat{P}$ and $L = \hat{P}/\mathcal{V}(V; \hat{V}_B) \equiv  \hat{P}/\mathcal{V}(V; \hat{V}_B; \hat{P}, \hat{\rho})$ where  $\hat{V}_B^*$ is the optimal bankruptcy level when $\rho = \hat{\rho}$ and $P = \hat{P}$.  For this computation, at least in our numerical experiments, the mapping $P \mapsto P/\mathcal{V}(V; \hat{V}_B; P, \rho)$, for fixed $\rho$,  is monotonically increasing and hence the root $\hat{P}(\rho)$ solving $L = \hat{P}(\rho)/\mathcal{V}(V; \hat{V}_B; \hat{P}(\rho), \rho)$ was obtained by classical bisection. In addition, $\rho \mapsto \mathcal{D}(V; \hat{V}_B^*; \hat{P}(\rho), \rho) - \hat{P}(\rho)$ was also monotone and hence the desired $\hat{P}$ and $\hat{\rho}$ were obtained by (nested) bisection methods.

For each leverage $L$, after $\hat{P}$ and $\hat{\rho}$ are computed, the second step is to obtain, for each maturity $t > 0$, the root $\rho^* = \rho^*(t)$ such that $1 = d(V; \hat{V}_B^*, t) \equiv d(V; \hat{V}_B^*, t; \rho^*)$ where
\begin{align*}
	d (V; \hat{V}_B^*, t; \rho) := \E \left[ \int_0^{t \wedge T_{\hat{V}_B^*}^-  } e^{-rs} \rho \diff s \right] + \E\left[ e^{-rt} \mathbf{1}_{\{t < T_{\hat{V}_B^*}^-  \}}\right] + \frac 1 {\hat{P}} \E \left[ e^{-r T_{\hat{V}_B^*}^-}V_{T_{\hat{V}_B^*}^-} \left(1-\alpha  \right) \mathbf{1}_{\{T_{\hat{V}_B^*}^-  < t \}} \right].
\end{align*}
The spread is given by $\rho^* - r$ (for each maturity $t$). The expectations on the right hand side can be computed again by the Gaver-Stehfest algorithm, by inverting
$q \mapsto J^{(q, \lambda)}(\cdot; \theta)$ as in \eqref{fun_gamma} for $\theta = 0, 1$. Those for the classical case can be computed by inverting $q \mapsto H^{(q)}(\cdot; \theta)$.

\begin{table}
\begin{center}
\begin{tabular}{l|lllllll|l}
$\lambda$  & 1 &2 & 4 & 6 & 12 & 52 & 365 & $\infty$ \\
 \hline
$\hat{P}$ &53.5721&53.2700 &53.1036&53.0457&52.9877&52.9419&52.9312&52.9297 \\
$\hat{\rho}$ &0.08643&0.08799&0.08892&0.08926&0.08960&0.08987&0.08994&0.08996 \\
$\hat{V}_B$ &53.6339&52.8191&51.9905&51.5509&50.9127&50.0097&49.4447&49.0871
\end{tabular} \\
\textbf{Case A} with  $L = 50$ \\
\begin{tabular}{l|lllllll|l}
$\lambda$ & 1 &2 & 4 & 6 & 12 & 52 & 365 & $\infty$ \\
 \hline
$\hat{P}$ &68.3632&66.8541&66.0011&65.7013&65.3961&65.1581&65.0978&65.0879 \\
$\hat{\rho}$ &0.11814&0.12462&0.1286&0.13006&0.13159&0.13281&0.13312&0.13318 \\
$\hat{V}_B$ &77.6117&76.3951&75.2&74.5702&73.656&72.3608&71.5453&71.0280
\end{tabular} \\
\textbf{Case A} with $L = 75$ \\
\begin{tabular}{l|lllllll|l}
$\lambda$ & 1 &2 & 4 & 6 & 12 & 52 & 365 & $\infty$ \\
 \hline
$\hat{P}$ &53.0411&52.7344&52.5543&52.4887&52.4216&52.3682&52.3529&52.3499 \\
$\hat{\rho}$ &0.10075&0.10459&0.10697&0.10785&0.10878&0.10953&0.10974&0.10977 \\
$\hat{V}_B$ &52.6127&51.8489&51.0405&50.6053&49.9712&49.0748&48.5135&48.1608
\end{tabular} \\
\textbf{Case B} with $L = 50$ \\
\begin{tabular}{l|lllllll|l}
$\lambda$ & 1 &2 & 4 & 6 & 12 & 52 & 365 & $\infty$ \\
 \hline
$\hat{P}$ &69.3832&67.8467&66.9418&66.6138&66.2712&65.9958&65.9225&65.9103 \\
$\hat{\rho}$ &0.1311&0.14061&0.14677&0.14911&0.15163&0.15372&0.15428&0.15438 \\
$\hat{V}_B$ &76.6621&75.5312&74.3924&73.7837&72.8906&71.6139&70.8058&70.2938
\end{tabular} \\
\textbf{Case B} with  $L = 75$
\end{center}
\caption{Values of $\hat{P}$, $\hat{\rho}$ and $\hat{V}_B$ satisfying $\mathcal{D}(V; \hat{V}_B^*) \equiv \mathcal{D}(V; \hat{V}_B^*; \hat{P}, \hat{\rho}) = \hat{P}$ and $L = \hat{P}/\mathcal{V}(V; \hat{V}_B) \equiv  \hat{P}/\mathcal{V}(V; \hat{V}_B; \hat{P}, \hat{\rho})$ for $L = 50,75$ for each $\lambda$ ($\lambda = \infty$ corresponds to the classical case).} \label{table_P_rho_V_B}
\end{table}

Here, we consider leverages $L = 50,75$ again for \textbf{Cases A} and \textbf{B}.  In Table \ref{table_P_rho_V_B}, the computed values of $\hat{P}$, $\hat{\rho}$ and $\hat{V}_B$ are listed for each $\lambda = 1,2,4,6,12,52,365$ along with those for the classical case.  In Figure \ref{plot_credit_spread}, we plot  the credit spread with respect to the log maturity for each $\lambda$. For comparison, we also plot those in the classical case.  The spread appears  to be monotone in $\lambda$ and converges to those in the classical case for each maturity.

Regarding the credit spread limit, while the convergence to zero has been confirmed in Proposition \ref{CS_limit_T_0} for the periodic case, the rate of convergence depends significantly on the selection of $\lambda$ and the underlying asset price process. In \textbf{Case A} (without negative jumps), it is clear that it vanishes quickly as in the classical case.  On the other hand in \textbf{Case B} (where the credit spread limit in the classical case does not vanish), for large values of $\lambda$ the convergence is very slow. In view of these observations, with a selection of asset values with negative jumps and the observation rate $\lambda$, it is capable of achieving realistic short-maturity credit spread behaviors.

\section{Concluding remarks} \label{section_conlusion} 

We studied an extension of the Leland-Toft optimal capital structure model where the information on the asset value is updated only at the jump times of an independent Poisson process. In settings where the asset value follows an exponential \lev process with negative jumps, we obtained explicitly an optimal bankruptcy strategy and the corresponding equity/debt/firm values. These analytical results enabled efficient conduct of  numerical experiments and further analysis of  the impact of the observation rate on the optimal leverages and credit spreads.

There are various venues for future research. First, it is a natural direction of research to consider the case in which the asset value process contains both positive and negative jumps.  Because positive jumps do not have direct influence on the model of the default,  similar results are expected and, for example, the optimal barrier is likely to be given by $V_B$ such that $\mathcal{E}(V_B;V_B) = 0$. 
While the techniques using the scale function employed in this paper cannot be directly applied to the two-sided jump cases, there are several potential alternative approaches.  
One approach would be to add phase-type upward jumps to the spectrally negative \lev process  via fluid embedding and construct a \lev process with two-sided jumps in terms of a Markov additive process. To do this the phase-type jumps of the L\'evy process can be substituted by linear stretches of unit slope. 
This procedure requires though adding a supplementary background  Markov chain; see e.g. 
\cite{Ivphd}  for details.
Another approach would be to focus on the \lev process with two-sided phase-type distributed jumps and use them to approximate for a general case. This may be possible by combining the results of Asmussen et al. \cite{Asmussen} and  Albrecher et al. \cite{Albrecher}.



Second, it is important to consider the constant grace period case  described in (1) of Section \ref{section_model_bankruptcy}.  As discussed, this paper's results, featuring exponential grace periods, may be used to approximate the constant  case when the grace period is short. However, an alternative approach is required when it is long.  One potential approach would be to use Carr's randomization method \cite{Carr} to approximate the constant period in terms of an Erlang random variable, or the sum of i.i.d.\ exponential random variables.  As conducted in \cite{Leung}, a recursive algorithm may be constructed to compute the required fluctuation identities.

\section*{acknowledgements}
The authors thank the anonymous referees and co-editor for careful reading of the paper and constructive comments and suggestions.
They also thank Nan Chen, Sebastian Gryglewicz, and Tak-Yuen Wong for helpful comments and discussions. K. Yamazaki is supported by MEXT KAKENHI grant no. 17K05377.
This paper was supported by the National Science Centre under the grant 2016/23/B/HS4/00566
(2017-2020).
Part of the work was completed while Z. Palmowski was visiting Kansai University and Kyoto University at the invitation of K. Yamazaki. Z. Palmowski is very grateful
for hospitality provided by Kazutoshi Yamazaki, Kouji Yano and Takashi Kumagai.

\appendix

\section{Relation between the bakruptcy model \eqref{our_default} and Parisian ruin.}\label{bankruptcy_parisianruin}
Let $G$ denote the set of the starting
points of the negative excursions of the shifted process $(V_t-V_B)_{t \geq 0}$, and consider a set of mutually independent exponential random variables $\{\mathbf{e}_\lambda^g:g\in G\}$, independent of $(V_t)_{t \geq 0}$ as well, and  $g_t :=\sup\{s\leq t:V_s\geq V_B \}$ be the last time before $t$ the asset value was at or above $V_B$ (i.e., the starting point of the excursion). Then the Parisian ruin with exponential grace periods is defined as
\begin{align}
\inf\{t >0: V_t<V_B\ \text{and} \ t>g_t+\mathbf{e}_\lambda^{g_t}\}. \label{ruin_Parisian_exp}
\end{align}

The equivalence to \eqref{our_default} can be easily verified.  
In each negative excursion with the starting time $g$ for the shifted process $(V_t-V_B)_{t\geq0}$ between two Poissonian observation times, say $T_{i(g)}$ and $T_{i(g)+1}$ for some $i(g) \geq 0$, we consider the waiting time until the next observation $T_{i(g)+1}-g$. Due to the lack of memory property of the exponential distribution and the strong Markov property, 
these waiting times are equal in distribution to a
set of mutually independent exponentially distributed random variables. Consequently,  \eqref{our_default} can be written as \eqref{ruin_Parisian_exp} with $\mathbf{e}_\lambda^{g_t}$ replaced by these independent exponential random variables. 
In fact, it has been shown in Remark 1.1 in \cite{BPPR} that the joint distribution of bankruptcy time \eqref{our_default} and the corresponding position of $X$ is the same as that of \eqref{ruin_Parisian_exp} and the corresponding position of $X$  
(refer to \cite{Pardo_Perez_Rivero, Avram_Perez_Yamazaki}  for related literature).

It is worth investing the impact of the randomness of the grace period. To this end, in Table \ref{table_comparison}, we compare  the expected discounted asset values at bankruptcy  for the cases the grace periods are constant and exponentially distributed (with the common mean $\lambda^{-1}$). When $\lambda$ is low, the random (exponential) case tends to overestimate the asset value, but  as $\lambda$ becomes larger (i.e.\ observation is more frequent), the differences become smaller.  This implies that when the observation is frequent, our model can approximate the constant grace period case reasonably well.
	
\begin{table}[htb]
	\begin{center}
		\begin{tabular}{cc}
			\begin{tabular}{|c|c|c|}
				\hline 
				$\lambda$ & constant & exponential \\
				\hline
				$1$&$4.710(4.674,4.747)$&$6.219(6.176,6.261)$ \\
				$2$&$5.795(5.753,5.838)$&$7.014(6.964,7.064)$ \\
				$4$&$6.717(6.674,6.760)$&$7.639(7.593,7.685)$ \\
				$6$&$7.125(7.070,7.179)$&$7.929(7.872,7.985)$ \\
				$12$&$7.727(7.668,7.785)$&$8.289(8.229,8.349)$ \\
				$52$&$8.543(8.492,8.595)$&$8.819(8.766,8.871)$ \\
				$365$&$8.886(8.824,8.947)$&$9.025(8.964,9.087)$  \\ 
				\hline
			\end{tabular} &
			\begin{tabular}{|c|c|c|}
				\hline 
				$\lambda$ & constant & exponential \\
				\hline
				$1$&$6.238(6.192,6.283)$&$7.749(7.692,7.807)$ \\
				$2$&$7.434(7.388,7.480)$&$8.589(8.537,8.642)$ \\
				$4$&$8.472(8.417,8.528)$&$9.395(9.338,9.451)$ \\
				$6$&$8.825(8.770,8.880)$&$9.584(9.530,9.638)$ \\
				$12$&$9.436(9.376,9.496)$&$9.976(9.914,10.037)$ \\
				$52$&$10.184(10.125,10.242)$&$10.444(10.385,10.503)$ \\
				$365$&$10.726(10.673,10.778)$&$10.820(10.766,10.873)$ \\
				\hline
			\end{tabular} \\
			\textbf{Case A} & \textbf{Case B}
		\end{tabular}
	\end{center}
	\caption{The discounted asset values at bankruptcy $\E[e^{-r \tau_{V_B}^-}V_{\tau_{V_B}^-} \mathbf{1}_{\{\tau_{V_B}^-  < \infty \}} ]$ when $\tau_{V_B}^-$ is the bankruptcy time with  constant and exponential grace periods with mean $\lambda^{-1}$.  The approximated values via Monte Carlo simulation are displayed together with their 95\% confidence intervals. We set  $r = 7.5 \%$ 
		and use  the \lev processes given in Cases A (without jumps) and B (with negative jumps) specified in Section \ref{section_numerics} so that  $(e^{-(r-\delta) t} V_t)_{t \geq 0}$ is a martingale for $\delta = 7\%$. The initial value of the process is $100$ and the bankruptcy level $V_B$ is $40$.} 
	\label{table_comparison}
\end{table}
\section{Proof of Proposition \ref{lambdaidentified}}
\label{proof_prop_lambdaidentified}

	For brevity, throughout the Appendix, we will use the notation
	\begin{align}
	z_T := z - \log V_T, \quad z \in \R. \label{def_z_T}
	\end{align}


We first obtain the $q$-resolvent measure of the spectrally negative \lev process $(X_t)_{t \geq 0}$ killed at the stopping time \eqref{default_time}
 in terms of the function $H^{(q+\lambda)}(\cdot; \theta)$ as in \eqref{der_gamma_0_0}, and
	\begin{multline}\label{Udef}
		I^{(q, \lambda)}(x,y) 
		:= W^{(q+\lambda)}(x+y)-\lambda\int_0^{x}W^{(q)}(x-z)W^{(q+\lambda)}(z+y) \diff z \\ -Z^{(q)}(x;\Phi(q+\lambda))W^{(q+\lambda)}(y), \quad q > 0, \; x, y \in \mathbb{R}.
\end{multline}
The proof of the following is given in Appendix \ref{proof_prop_resolvent}.
\begin{theorem}\label{prop_resolvent} 
	For any bounded measurable function $h:\R \to \R$ with compact support, 
	\begin{align*}
		\E_x\left[ \int_0^{\tilde{T}_z^-  } e^{-qt} h(X_t)  \diff t \right]= \int_{\R}
		h(y+z)R^{(q,\lambda)}(x-z,y) \diff y, \quad x, z \in \R,
	\end{align*}
	where
	\begin{align}\label{resol_dens}
		R^{(q,\lambda)}(x,y):=Z^{(q)}(x;\Phi(q+\lambda))\frac  {\Phi(q+\lambda)- \Phi(q)} {\lambda}
		H^{(q+\lambda)}(-y; \Phi(q))- I^{(q, \lambda)}(x,-y).
	\end{align}
\end{theorem}
Using Theorem \ref{prop_resolvent}, we show Proposition \ref{lambdaidentified}.  The case $V_T = 0$ is trivial and hence we assume $V_T > 0$ for the rest.
By integrating the density in Theorem \ref{prop_resolvent} and using \eqref{def_z_T}, we can write \eqref{Lambda_def} as
\begin{align}\label{def_lambda}
\Lambda^{(r,\lambda)}(x,z)&= Z^{(r)}(x-z;\Phi(r+\lambda))\frac  {\Phi(r+\lambda)- \Phi(r)} {\lambda} \mathcal{H}(z) - \mathcal{I}(x,z),
\end{align}
where we define
\begin{align}
\mathcal{H}(z) :=
\int_{-\infty}^{z_T} H^{(r+\lambda)}(y; \Phi(r))\diff y
 \quad \textrm{and} \quad \mathcal{I}(x,z) := \int_{-\infty}^{z_T}  I^{(r, \lambda)}(x-z,y)\diff y, \label{def_H_I_integral}
\end{align}
which are shown to be finite immediately below. 
The rest of the proof of Proposition \ref{lambdaidentified} is devoted to the simplification of the integrals $\mathcal{H}$ and $\mathcal{I}$.


\begin{lemma} \label{lemma_relation_scale_function_bar} For all $y \in \mathbb{R}$, we have
$\overline{W}^{(r+\lambda)}(y) -\lambda\int_0^{y}W^{(r)}(y-z)  \overline{W}^{(r+\lambda)}(z) \diff z= \overline{W}^{(r)}(y)$. 
\end{lemma}
\begin{proof}
We have
\begin{multline*}
\frac \partial {\partial y} \Big(\overline{W}^{(r+\lambda)}(y) -\lambda\int_0^{y}W^{(r)}(y-z)  \overline{W}^{(r+\lambda)}(z) \diff z \Big)
= \frac \partial {\partial y} \Big(\overline{W}^{(r+\lambda)}(y) -\lambda\int_0^{y}W^{(r)}(z)  \overline{W}^{(r+\lambda)}(y-z) \diff z \Big) \\
=W^{(r+\lambda)}(y) -\lambda\int_0^{y}W^{(r)}(z)  W^{(r+\lambda)}(y-z) \diff z = W^{(r)}(y),
\end{multline*}
where the last equality holds by identity (6) of \cite{LRZ}.
Integrating this and because $\overline{W}^{(r+\lambda)}(0) = \overline{W}^{(r)}(0) = 0$, the proof is complete.
\end{proof}

\begin{lemma}\label{lemma_int_upsilon} 
	We have, for $x,z \in \R$,
\begin{align*}
\mathcal{I}(x,z)&=\overline{W}^{(r+\lambda)}(x- \log V_T) \mathbf{1}_{\{z_T>0\}} +  \overline{W}^{(r)}(x-\log V_T)  \mathbf{1}_{\{z_T\leq 0\}} \\ &-\lambda\int_0^{x-z}W^{(r)}(x-z-u)   \overline{W}^{(r+\lambda)}(u+z_T) \diff u \mathbf{1}_{\{z_T>0\}} -Z^{(r)}(x-z;\Phi(r+\lambda))  \overline{W}^{(r+\lambda)}(z_T).
\end{align*}
\end{lemma}
\begin{proof} 
For $z_T> 0$, we have
\begin{align*} 
\begin{split}
 \int_0^{z_T}  I^{(r, \lambda)}(x-z,y)\diff y &=\int_0^{z_T} W^{(r+\lambda)}(x-z+y) \diff y -\lambda\int_0^{x-z}W^{(r)}(x-z-u)  \int_0^{z_T} W^{(r+\lambda)}(u+y) \diff y \diff u \\
&\qquad -Z^{(r)}(x-z;\Phi(r+\lambda))  \int_0^{z_T} W^{(r+\lambda)}(y) \diff y \\
&=\overline{W}^{(r+\lambda)}(x- \log V_T)-\overline{W}^{(r+\lambda)}(x-z) \\ &\qquad-
\lambda\int_0^{x-z}W^{(r)}(x-z-u)  (\overline{W}^{(r+\lambda)}(u+z_T )-\overline{W}^{(r+\lambda)}(u)  )\diff u \\
	&\qquad
	-Z^{(r)}(x-z;\Phi(r+\lambda))  \overline{W}^{(r+\lambda)}(z_T ) \\
	&=\overline{W}^{(r+\lambda)}(x- \log V_T)-\overline{W}^{(r)}(x-z)   \\ &-\lambda\int_0^{x-z}W^{(r)}(x-z-u)   \overline{W}^{(r+\lambda)}(u+z_T) \diff u -Z^{(r)}(x-z;\Phi(r+\lambda))  \overline{W}^{(r+\lambda)}(z_T),
	\end{split}
\end{align*}
where we used $x-z+z_T = x - \log V_T$ (see \eqref{def_z_T}) in the second equality and Lemma \ref{lemma_relation_scale_function_bar} in the last equality.

On the other hand, because, as in  Remark 4.3(ii) in \cite{BPY}, 
\begin{align}
I^{(r, \lambda)}(x,y)=W^{(r)}(x+y), \quad y < 0, \label{I_below_zero}
\end{align}
we have
\begin{align*}
\int_{-\infty}^{0\wedge z_T}  I^{(r, \lambda)}(x-z,y)\diff y=	\int_{-\infty}^{0\wedge z_T} W^{(r)}(x-z+y)\diff y=\overline{W}^{(r)}(x-z+(0\wedge z_T)). 
\end{align*}
Now the result is immediate by summing up the two integrals and using (again see \eqref{def_z_T})
\begin{align*}
\overline{W}^{(r)}(x-z+(0\wedge z_T)) = \left\{ \begin{array}{ll} \overline{W}^{(r)}(x-z) & \textrm{if} \, z_T > 0, \\ \overline{W}^{(r)}(x-\log V_T) & \textrm{if} \,  z_T \leq 0. \end{array}\right.
\end{align*}
\qed
\end{proof}
We note that \eqref{def_lambda} together with Lemma \ref{lemma_int_upsilon} imply that 
\begin{equation}\label{lambda_b_b}
\Lambda^{(r,\lambda)}(z,z)=\frac  {\Phi(r+\lambda)- \Phi(r)} {\lambda} \int_{-\infty}^{z_T} H^{(r+\lambda)}(y; \Phi(r))\diff y, \quad z \in \mathbb{R}.
\end{equation}
\begin{lemma}\label{nowareferencja}
For $z \in \mathbb{R}$, we have
\begin{align}\label{int_H}
	\mathcal{H}(z)  = \frac 1 {\Phi(r)} \Big(Z^{(r+\lambda)}(z_T;\Phi(r))-\lambda\frac{\Phi(r+\lambda)}{\Phi(r+\lambda)- \Phi(r)}\overline{W}^{(r+\lambda)}(z_T) \Big).
  \end{align}
  \end{lemma}
\begin{proof}
First, by \eqref{der_gamma_0_0}, we have
\begin{align*}
H^{(r+\lambda)}(y;\Phi(r)) = e^{\Phi(r)y}\left(1+\lambda\int_0^{y}e^{-\Phi(r)u}W^{(r+\lambda)}(u)\diff u\right)-\frac{\lambda}{\Phi(r+\lambda)-\Phi(r)}W^{(r+\lambda)}(y), \quad y \in \R,
\end{align*}
where, in particular, $H^{(r+\lambda)}(y;\Phi(r))  = e^{\Phi(r)y}$ for $y < 0$. 
For $z_T>0$, 
\begin{multline*}
\int^{z_T}_0 e^{\Phi(r)y}\int_0^{y}e^{-\Phi(r)u}W^{(r+\lambda)}(u)\diff u\diff y = \int^{z_T}_0 \int_u^{z_T} e^{\Phi(r)y}e^{-\Phi(r)u}W^{(r+\lambda)}(u)\diff y \diff u \\
\quad= \int^{z_T}_0 \frac {e^{\Phi(r) z_T} - e^{\Phi(r) u}} {\Phi(r)} e^{-\Phi(r)u}W^{(r+\lambda)}(u) \diff u = \frac 1 {\Phi(r)} \Big[ \int^{z_T}_0 e^{\Phi(r) (z_T-u)} W^{(r+\lambda)}(u) \diff u -   \overline{W}^{(r+\lambda)}(z_T) \Big],
\end{multline*}
and hence
\begin{align*} 
	\begin{split}
		&
		\int^{z_T}_0H^{(r+\lambda)}(y;\Phi(r))\diff y
		=
		\int^{z_T}_0\left[e^{\Phi(r)y}\left(1+\lambda\int_0^{y}e^{-\Phi(r)u}W^{(r+\lambda)}(u)\diff u\right)-\frac{\lambda}{\Phi(r+\lambda)-\Phi(r)}W^{(r+\lambda)}(y)\right]\diff y \\
		&\quad=\frac 1 {\Phi(r)} \Big[ e^{\Phi(r)z_T}-1+\lambda\int_0^{z_T}e^{\Phi(r)(z_T-u)}W^{(r+\lambda)}(u)\diff u-\lambda\frac{\Phi(r+\lambda)}{\Phi(r+\lambda)- \Phi(r)}\overline{W}^{(r+\lambda)}(z_T) \Big] \\
		&\quad =\frac 1 {\Phi(r)} \Big[ Z^{(r+\lambda)}(z_T;\Phi(r))-1-\lambda\frac{\Phi(r+\lambda)}{\Phi(r+\lambda)- \Phi(r)}\overline{W}^{(r+\lambda)}(z_T) \Big].
	\end{split}
\end{align*}
On the other hand, for $z_T\in\R$,
 $\int^{z_T \wedge 0}_{-\infty} H^{(r+\lambda)}(y; \Phi(r))\diff y
	=
	\int^{z_T \wedge 0}_{-\infty}  e^{\Phi(r)y}\diff y
	=
	e^{\Phi(r)(z_T \wedge 0)
	} / \Phi(r)$.
By summing up the integrals, we obtain \eqref{int_H}.
\qed \end{proof}

Now applying Lemmas \ref{lemma_int_upsilon}  and \ref{nowareferencja} in \eqref{def_lambda}, we get Proposition \ref{lambdaidentified}.


\section{Other proofs}

 \subsection{Proof of Lemma \ref{lemma_Lambda_limit}}\label{proof_lemma_Lambda_limit}

For the case $V_T > 0$,
	\[\Lambda^{(r,\lambda)}(z,z)=\E_{z} \big[ \int_0^{\tilde{T}^-_z} e^{-rt}\mathbf{1}_{\{X_t\geq\log V_T\}}  \diff t \big]=\E_0 \big[ \int_0^{\tilde{T}^-_0} e^{-rt}\mathbf{1}_{\{X_t\geq\log V_T-z\}}  \diff t \big]\]
is clearly non-decreasing in $z$, and, by bounded convergence, $\lim_{z \downarrow -\infty}\Lambda^{(r,\lambda)}(z,z)
=0.$

On the other hand, if $V_T=0$, then, by Proposition \ref{lambdaidentified} and Remark \ref{remark_gamma_0}(1),
	$\Lambda^{(r,\lambda)}(z,z)
	=\frac{1}{r}(1-J^{(r, \lambda)}(0;0))= \frac{1}{\lambda +r}  \frac{\Phi(r+\lambda)}{\Phi(r)}$.
 \subsection{Proof of Proposition \ref{prop_E_derivative_B}} \label{proof_prop_E_derivative_B}
We start from several key introductory identities.
Fix $q > 0$. Because
\begin{align*}
e^{\theta z}Z^{({q})}(x-z;\theta) &=  e^{\theta x} \left( 1 + ({q}- \psi(\theta )) \int_0^{x-z} e^{-\theta  u} W^{({q})}(u) \diff u	\right),
\end{align*}
we have, for $x \neq z$,
\begin{align*}\frac \partial {\partial z} [e^{\theta z}Z^{({q})}(x-z;\theta) ] &=  - e^{\theta z}  ({q}- \psi(\theta ))  W^{({q})}(x-z), \\
\frac \partial {\partial x}Z^{(q)}(x-z;\theta) &= \frac \partial {\partial x} \Big[ e^{\theta (x-z)} \left( 1 + (q- \psi(\theta )) \int_0^{x-z} e^{-\theta  u} W^{(q)}(u) \diff u	\right) \Big] \\ &= \theta Z^{(q)}(x-z;\theta)  +  (q- \psi(\theta ))  W^{(q)}(x-z).
\end{align*}
In particular, 
\begin{align}\label{Z_Phi_q_r_derivative}
	\frac \partial {\partial x}Z^{(q)}(x-z;\Phi(r+\lambda)) 
	&= \Phi(r+\lambda) Z^{(q)}(x-z;\Phi(r+\lambda))  -\lambda W^{(q)}(x-z).
\end{align}
Moreover, we have, for $x \neq z$,
\begin{align}\label{aux_appendix}
&
\frac \partial {\partial x} J^{(q, \lambda)}(x-z; \theta) \notag\\
&=\frac{\lambda}{\lambda +{q}-\psi(\theta)} \frac \partial {\partial x} Z^{({q})}(x-z;\theta)- \frac{\psi(\theta)-{q}}{\lambda +{q}-\psi(\theta)}\frac{\Phi({q}+\lambda)-\Phi({q})}{\theta-\Phi({q})} \frac \partial {\partial x} Z^{({q})}(x-z;\Phi({q}+\lambda)) \notag\\
&= \frac{\lambda}{\lambda +{q}-\psi(\theta)} \theta Z^{({q})}(x-z;\theta)   - \frac{\psi(\theta)-{q}}{\lambda +{q}-\psi(\theta)}\frac{\Phi({q}+\lambda)-\Phi({q})}{\theta-\Phi({q})} \Phi(q+\lambda) Z^{({q})}(x-z;\Phi(q+\lambda))
\notag\\ &\qquad + \frac{\psi(\theta)-{q}}{\lambda +{q}-\psi(\theta)}\frac{\Phi({q}+\lambda)-\theta}{\theta-\Phi({q})} \lambda  W^{({q})}(x-z).
\end{align}

By setting $\theta = 0$, we obtain the following.
\begin{lemma}\label{cztery}
We have, for $x \neq z$ and $q > 0$,
\begin{align}\label{der_gamma_0}
\begin{split}
\frac \partial {\partial z} J^{(q, \lambda)}(x-z; 0)&=-\frac\partial {\partial x} J^{(q, \lambda)}(x-z; 0)
=\frac{{q}}{\lambda +{q}}\frac{\Phi({q}+\lambda)-\Phi({q})}{\Phi({q})} \Phi(q+\lambda)H^{(q)}(x-z; \Phi(q+\lambda)).
\end{split}
\end{align}
\end{lemma}
Noting that $\frac \partial {\partial z} [e^{z} J^{(q, \lambda)}(x-z;1) ] = e^{z} J^{(q, \lambda)}(x-z;1) - e^{z} \frac {\partial} {\partial x} J^{(q, \lambda)}(x-z;1)$, and using \eqref{aux_appendix} with $\theta=1$, we have the following result.
\begin{lemma}\label{pochodne}
We have, for $x \neq z$ and $q > 0$,
\begin{equation}\label{der_gamma_1}
\frac \partial {\partial z} [e^{z} J^{(q, \lambda)}(x-z;1) ] \\
=\frac{\psi(1)-{q}}{\lambda +{q}-\psi(1)}\frac{\Phi({q}+\lambda)-\Phi({q})}{1-\Phi({q})} (\Phi(q+\lambda) -1) e^{z} H^{(q)}(x-z; \Phi(q+\lambda)). \notag
\end{equation}
\end{lemma}
We will also need the following observation.
\begin{lemma}\label{der_Lambda}
	We have,  for $z_T\not=0$ and $x > z$,
\begin{align} \label{der_lambda_3}
&\frac{\partial}{\partial z} \Lambda^{(r,\lambda)}(x,z) =-\frac  {(\Phi(r+\lambda)- \Phi(r))^2} {\lambda}H^{(r)}(x-z;\Phi(r+\lambda))
\mathcal{H}(z).
\end{align}
\end{lemma}
\begin{proof}
By differentiating the identity in  Lemma \ref{lemma_int_upsilon}, for $z_T\not=0$,  by \eqref{Z_Phi_q_r_derivative},
\begin{align}\label{der_upsilon_B}
	\frac \partial {\partial z} \mathcal{I}(x,z)&
	= -\lambda \frac \partial {\partial z}\int_0^{x-z}W^{(r)}(w)   \overline{W}^{(r+\lambda)}(x-w- \log V_T) \diff w \mathbf{1}_{\{ z_T>0\}} \notag\\
	&\qquad+\frac \partial {\partial x}Z^{(r)}(x-z;\Phi(r+\lambda)) \overline{W}^{(r+\lambda)}(z_T ) -Z^{(r)}(x-z;\Phi(r+\lambda))  \frac \partial {\partial z}\overline{W}^{(r+\lambda)}(z_T ) \notag\\
	&= \lambda W^{(r)}(x-z)   \overline{W}^{(r+\lambda)}(z_T) \notag\\
	&\qquad+[\Phi(r+\lambda) Z^{(r)}(x-z;\Phi(r+\lambda)) - \lambda W^{(r)}(x-z)]  \overline{W}^{(r+\lambda)}(z_T ) \notag\\
	&\qquad-Z^{(r)}(x-z;\Phi(r+\lambda))  W^{(r+\lambda)}(z_T ) \notag\\
	&=Z^{(r)}(x-z;\Phi(r+\lambda)) [\Phi(r+\lambda) \overline{W}^{(r+\lambda)}(z_T ) - W^{(r+\lambda)}(z_T ) ].
\end{align}
By \eqref{der_gamma_0_0} and \eqref{Z_Phi_q_r_derivative},  we can write
\begin{equation} \label{Z_prime_H_relation}
\frac \partial {\partial x}Z^{(r)}(x-z;\Phi(r+\lambda)) 
= (\Phi({r}+\lambda)-\Phi({r})) H^{(r)}(x-z; \Phi(r+\lambda)) +\Phi({r}) Z^{({r})}(x-z;\Phi(r+\lambda)).
\end{equation}
Using \eqref{def_H_I_integral}, we have, for  $x > z$ and $z_T\not=0$, 
\begin{multline} \label{Z_H_derivative}
\frac{\partial}{\partial z}\left(Z^{(r)}(x-z;\Phi(r+\lambda))
\mathcal{H}(z)
\right) \\
=- \frac \partial {\partial x}Z^{(r)}(x-z;\Phi(r+\lambda))
\mathcal{H}(z)
+Z^{(r)}(x-z;\Phi(r+\lambda))
H^{(r+\lambda)}(z_T;\Phi(r)).
\end{multline}
By \eqref{int_H} and \eqref{Z_prime_H_relation}, this equals
\begin{multline*}
- (\Phi(r+\lambda)- \Phi(r)) H^{(r)}(x-z;\Phi(r+\lambda))
\mathcal{H}(z)
 +Z^{(r)}(x-z;\Phi(r+\lambda))
\Big( H^{(r+\lambda)}(z_T;\Phi(r)) - \Phi(r)
\mathcal{H}(z)
\Big).
\end{multline*}
Furthermore, by \eqref{int_H},
\begin{align*}
H^{(r+\lambda)}(z_T;\Phi(r)) - \Phi(r) \mathcal{H}(z)
 &= H^{(r+\lambda)}(z_T;\Phi(r)) -Z^{(r+\lambda)}(z_T;\Phi(r))+\lambda\frac{\Phi(r+\lambda)}{\Phi(r+\lambda)- \Phi(r)}\overline{W}^{(r+\lambda)}(z_T) \notag\\
 &=\frac {\lambda} {\Phi(r+\lambda)-\Phi(r)} (\Phi(r+\lambda) \overline{W}^{(r+\lambda)}(z_T) - W^{(r+\lambda)}(z_T) ).
 \end{align*}
In sum, we have
\begin{align}\label{derivative_Z_H}
\begin{split}
&\frac{\partial}{\partial z}\left(Z^{(r)}(x-z;\Phi(r+\lambda))
\mathcal{H}(z)
\right)=- {(\Phi(r+\lambda)- \Phi(r))} H^{(r)}(x-z;\Phi(r+\lambda))
\mathcal{H}(z)
 \\
 &\qquad+ \frac {\lambda} {\Phi(r+\lambda)-\Phi(r)} Z^{(r)}(x-z;\Phi(r+\lambda)) (\Phi(r+\lambda) \overline{W}^{(r+\lambda)}(z_T) - W^{(r+\lambda)}(z_T) ).
 \end{split}
 \end{align}
By applying \eqref{der_upsilon_B} and \eqref{derivative_Z_H} in \eqref{def_lambda}, the proof is complete.
\qed\end{proof}

We now prove Proposition \ref{prop_E_derivative_B}.
Differentiating \eqref{equity_value} and using Lemmas \ref{cztery}, \ref{pochodne} and \ref{der_Lambda}
give
\begin{align*}
	&\frac{\partial}{\partial V_B} \mathcal{E}(V;V_B) =- V_B^{-1}\Bigg[\alpha\frac{\psi(1)-r}{\lambda +{r}-\psi(1)}\frac{\Phi(r+\lambda)-\Phi(r)}{1-\Phi(r)} (\Phi(r+\lambda) -1) V_B\notag\\&
	+P \kappa \rho\frac  {(\Phi(r+\lambda)- \Phi(r))^2} {\lambda}
	\int_{-\infty}^{\log (V_B/V_T)} H^{(r+\lambda)}(y; \Phi(r))\diff y
	\Bigg]H^{(r)} \Big(\log \frac V {V_B}; \Phi(r+\lambda) \Big)\notag\\
	&- V_B^{-1}\Bigg[(1-\alpha)\frac{\psi(1)-{r -m}}{\lambda +r+m-\psi(1)}\frac{\Phi(r+m+\lambda)-\Phi(r+m)}{1-\Phi(r+m)} (\Phi(r+m+\lambda) -1) V_B \\&-
	\frac {P \rho+ p} {r+m}\frac{r+m}{\lambda +{r+m}}\frac{\Phi(r+m+\lambda)-\Phi(r+m)}{\Phi(r+m)} \Phi(r+m+\lambda)\Bigg]H^{(r+m)}\Big(\log  \frac V {V_B}; \Phi(r+m+\lambda) \Big), \notag
\end{align*}
which reduces to \eqref{eqn_E_derivative_B} after simplification using Remark \ref{remark_gamma_0}(1).

\subsection{Proof of Proposition \ref{prop_E_der_B_negative}} \label{proof_prop_E_der_B_negative}
In view of the probabilistic expression \eqref{der_gamma_0_0}, $q \mapsto H^{(q)}(x-z; \Phi(q+\lambda))$
is non-increasing for $x,z\in\R$, and hence
\begin{align*}
	\frac{H^{(r)}(x-z; \Phi(r+\lambda))}{H^{(r+m)}(x-z; \Phi(r+m+\lambda))}\geq 1,\qquad\text{for $x,z\in\R$}.
\end{align*}
On the other hand,
because $\psi$ is strictly convex and strictly increasing on $[\Phi(0), \infty)$, its right-inverse $\Phi$ is strictly concave, 
that is $\Phi'(r+\lambda+x)-\Phi'(r+x)<0$ for $x, \lambda > 0$. Therefore,
\begin{align*}
	\frac{\Phi(r+\lambda)-\Phi(r)}{\Phi(r+m+\lambda)-\Phi(r+m)} > 1\qquad\text{for $\lambda>0$}.
\end{align*}
Combining these,
\begin{align}\label{ineq_H}
	\frac{\Phi(r+\lambda)-\Phi(r)}{\Phi(r+m+\lambda)-\Phi(r+m)}\frac{H^{(r)}(x-z; \Phi(r+\lambda))}{H^{(r+m)}(x-z; \Phi(r+m+\lambda))}> 1.
\end{align}

By Remark  \ref{remark_gamma_0}(2)
and
because \eqref{der_gamma_0_0} implies that  $H^{(r+\lambda)}$ is uniformly nonnegative, we have
\begin{align*}
\alpha(1-J^{(r, \lambda)}(0;1))e^z+P \kappa \rho\frac  {\Phi(r+\lambda)- \Phi(r)} {\lambda}
\mathcal{H}(z)
\geq0.
\end{align*}
Hence, by the previous inequality together with \eqref{lambda_b_b}, \eqref{ineq_H}, and \eqref{E_B_B},
\begin{align*}
&L(\log V, \log V_B) = \frac{H^{(r)}(\log V-\log V_B; \Phi(r+\lambda))}{H^{(r+m)}(\log V-\log V_B; \Phi(r+m+\lambda))} \frac{\Phi(r+\lambda)-\Phi(r)}{\Phi(r+m+\lambda)-\Phi(r+m)}\\ &\qquad\times \Bigg(\alpha(1-J^{(r, \lambda)}(0;1)) V_B+P \kappa \rho
\frac  {\Phi(r+\lambda)- \Phi(r)} {\lambda} \int_{-\infty}^{\log V_B-\log V_T} H^{(r+\lambda)}(y; \Phi(r))\diff y
\Bigg)\\&\qquad+(1-\alpha)(1-J^{(r+m, \lambda)}(0;1)) V_B-
	\frac {P \rho+ p} {r+m}(1-J^{(r+m, \lambda)}(0;0)) \\
	&\quad>  \alpha(1-J^{(r, \lambda)}(0;1)) V_B+P \kappa \rho
	\Lambda^{(r,\lambda)} (\log V_B, \log V_B) \\
	 &\qquad+(1-\alpha)(1-J^{(r+m, \lambda)}(0;1)) V_B-
	\frac {P \rho+ p} {r+m}(1-J^{(r+m, \lambda)}(0;0)) \\
	&\quad= \mathcal{E}(V_B; V_B).
\end{align*}
In addition,
because $V_B \geq V_B^*$, by the monotonicity as in Proposition \ref{E_monotonicity}, we have $\mathcal{E}(V_B;V_B) \geq 0$. Note when $V_B^* = 0$ that $\mathcal{E}(V_B;V_B) \geq 0$ for all $V_B > 0$ by Proposition \ref{E_monotonicity}. 

Now, by Proposition \ref{prop_E_derivative_B} and recalling that $H^{(r+m)}$ is positive,
\begin{align*}
	\frac \partial {\partial V_B} \mathcal{E}(V;V_B)
	&<  -  (\Phi(r+m+\lambda)-\Phi(r+m))H^{(r+m)}\Big(\log  \frac V {V_B}; \Phi(r+m+\lambda) \Big)  \frac {\mathcal{E}(V_B;V_B)} {V_B} \leq 0.
\end{align*}
\qed

\subsection{Proof of Proposition \ref{prop_E_derivative_x}} \label{proof_prop_E_derivative_x}
Using Lemma \ref{lemma_int_upsilon} together with \eqref{Z_Phi_q_r_derivative} and \eqref{der_upsilon_B}, for $x\not=\log V_T$ and $z \in \R$ such that $z_T \neq 0$, 
\begin{align}\label{der_upsilon_x}
	&\frac{\partial}{\partial x} \mathcal{I}(x,z)
	=W^{(r+\lambda)}(x- \log V_T) \mathbf{1}_{\{ z_T>0\}} +  W^{(r)}(x-\log V_T) 
\mathbf{1}_{\{ z_T < 0\}} \notag\\ &\qquad- \frac \partial {\partial x}\lambda\int_0^{x-z}W^{(r)}(w)
	\overline{W}^{(r+\lambda)}(x-z-w+z_T) \diff w \mathbf{1}_{\{ z_T>0\}} \notag\\
	&\qquad-\frac \partial {\partial x}Z^{(r)}(x-z;\Phi(r+\lambda)) \overline{W}^{(r+\lambda)}(z_T) \notag\\
	&\quad=W^{(r+\lambda)}(x- \log V_T) \mathbf{1}_{\{ z_T>0\}} +  W^{(r)}(x-\log V_T)  
\mathbf{1}_{\{ z_T < 0\}} \notag\\ &\qquad- \lambda\int_0^{x-z}W^{(r)}(w)
	 W^{(r+\lambda)}(x-z-w+z_T) \diff w \mathbf{1}_{\{ z_T>0\}} -\lambda W^{(r)}(x-z)   \overline{W}^{(r+\lambda)}(z_T) \notag\\
	&\qquad-(\Phi(r+\lambda) Z^{(r)}(x-z;\Phi(r+\lambda)) - \lambda W^{(r)}(x-z))  \overline{W}^{(r+\lambda)}(z_T)\notag\\
	&\quad=I^{(r, \lambda)}(x-z,z_T)-\frac \partial {\partial z} \mathcal{I}(x,z),
\end{align}	
where we used \eqref{I_below_zero} for the case $z_T <0$.
Hence using \eqref{Z_H_derivative} and \eqref{der_upsilon_x} in \eqref{def_lambda}, and by \eqref{resol_dens},
\begin{align}\label{der_lambda_x_3}
\frac{\partial}{\partial x}\Lambda^{(r,\lambda)}(x,z)&=-\frac{\partial}{\partial z}\Lambda^{(r,\lambda)}(x,z)+R^{(r,\lambda)}(x-z,-z_T).
\end{align}

Now we write \eqref{equity_value} as
\begin{align}\label{new_def_E}
	\mathcal{E}(V;V_B)
	&=\mathcal{A}(\log V, \log V_B)+P \kappa \rho \Lambda^{(r,\lambda)}(\log V, \log V_B),
\end{align}
where
\begin{align*}
\mathcal{A}(x,z) &:= e^x -\alpha   e^{z} J^{(r, \lambda)}(x-z; 1) - \frac {P \rho+ p} {r+m}  (1- J^{(r+m, \lambda)}(x-z;0) ) - \left(1-\alpha  \right) e^{z} J^{(r+m, \lambda)}(x-z; 1).
\end{align*}
Differentiating this with respect to $x$ and $z$, we get
\begin{align*}
	\begin{split}
		\frac\partial {\partial x}\mathcal{A}(x,z)
		&= e^x  - \alpha  e^{z} \frac\partial {\partial x} J^{(r, \lambda)}(x-z; 1) + \frac {P \rho+ p} {r+m}   \frac\partial {\partial x} J^{(r+m, \lambda)}(x-z;0)  - \left(1-\alpha  \right) e^{z} \frac\partial {\partial x} J^{(r+m, \lambda)}(x-z; 1), \\
		\frac \partial {\partial z}\mathcal{A}(x,z)
		&= -  \alpha  e^{z} J^{(r, \lambda)}(x-z; 1) +  \alpha  e^{z} \frac\partial {\partial x} J^{(r, \lambda)}(x-z; 1) \\
		&- \frac {P \rho+ p} {r+m}  \frac\partial {\partial x} J^{(r+m, \lambda)}(x-z;0)  - \left(1-\alpha  \right) e^{z} J^{(r+m, \lambda)}(x-z; 1) +  \left(1-\alpha  \right) e^{z} \frac\partial {\partial x} J^{(r+m, \lambda)}(x-z; 1),
	\end{split}
\end{align*}
and hence
\begin{align}\label{der_A}
	\frac \partial {\partial x} \mathcal{A}(x,z)&=e^x-\frac \partial {\partial z} \mathcal{A}(x,z)-\alpha e^{z} J^{(r, \lambda)}(x-z;1)-(1-\alpha)e^{z} J^{(r+m, \lambda)}(x-z;1).
\end{align}

Finally, using \eqref{der_lambda_x_3} and \eqref{der_A} in \eqref{new_def_E}, we obtain that
\begin{align*}
&\frac \partial {\partial V} \mathcal{E}(V;V_B)
= \frac 1 V \frac \partial {\partial x}\Big[\mathcal{A}(x, \log V_B)+P \kappa \rho \Lambda^{(r,\lambda)}(x, \log V_B) \Big] \Big|_{x=\log V}\\
&= \frac 1 V \Big[ V -\frac \partial {\partial z} \mathcal{A}(\log V,z) |_{z = \log V_B}-\alpha V_B J^{(r, \lambda)}(\log \frac V {V_B};1)\notag\\&-(1-\alpha)V_B J^{(r+m, \lambda)}(\log \frac V {V_B};1)
-P \kappa \rho \frac \partial {\partial z} \Lambda^{(r,\lambda)}(\log V,z) |_{z = \log V_B}+ P \kappa \rho  R^{(r,\lambda)} \Big(\log \frac V {V_B},\log \frac {V_T} {V_B} \Big)\Big]  \notag
\end{align*}
which reduces to the desired expression by noting that \eqref{new_def_E} gives
\begin{align*}
\frac \partial {\partial V_B}  \mathcal{E}(V;V_B)
= \frac 1 {V_B} \Big[ \frac \partial {\partial z} \mathcal{A}(\log V,z) |_{z = \log V_B}+
P \kappa \rho \frac \partial {\partial z} \Lambda^{(r,\lambda)}(\log V,z) |_{z = \log V_B} \Big].
\end{align*}

\subsection{Proof of Lemma \ref{limit_lambda_ot}}\label{appendix_limit_lambda_ot}
First we note by Theorem VII.4 in \cite{Bertoin96a}, that for $q\geq0$
\begin{align}\label{lim_frac_phi}
	\lim_{\lambda\to \infty}\frac{\Phi(\lambda+r+m)}{\Phi(\lambda+q)}=1.
\end{align}
On the other hand, identity \eqref{lambda_b_b} implies, for $V_B> 0$, 
that
\begin{align*}
	\frac{\lambda}{\Phi(r+\lambda)-\Phi(r)}\Lambda^{(r,\lambda)}(\log V_B,\log V_B)
	= \begin{cases} \frac{1}{\Phi(r)}\left(\frac{V_B}{V_T}\right)^{\Phi(r)} & \textrm{if } V_B/V_T < 1, \\ \frac{1}{\Phi(r)}  +\int_0^{\log(V_B/V_T)} H^{(r+\lambda)}(y;\Phi(r))\diff y &\textrm{if }  V_B/V_T \geq 1, \end{cases}
\end{align*}	
where we used $H^{(r+\lambda)}(y;\Phi(r)) = \exp(\Phi(r)y)$ for $y \leq 0$.
In addition, by the probabilistic expression of the probabilistic expression of $H^{(r+\lambda)}$ given in \eqref{der_gamma_0_0} and using dominated convergence, we have
\[ \lim_{\lambda\to\infty}
\int_0^{\log(V_B/V_T)} H^{(r+\lambda)}(y;\Phi(r))\diff y=0. \]
This together with \eqref{lim_frac_phi} gives
\begin{align*}
	\lim_{\lambda\to\infty}\frac{\lambda+r+m }{\Phi(\lambda+r+m)}
	\Lambda^{(r,\lambda)}(\log V_B,\log V_B)
	=
	\frac{1}{\Phi(r)}\left[\left( \frac {V_B} {V_T} \right)^{\Phi(r)}\wedge 1\right]. 
\end{align*}
From Remark \ref{remark_gamma_0}(1) and \eqref{lim_frac_phi}, we can conclude that, for $q\geq0$,
\begin{align*}
\frac{\lambda+r+m }{\Phi(\lambda+r+m)}(1-J^{(q, \lambda)}(0;1)) = \frac{\lambda+r+m }{\Phi(\lambda+r+m)} \frac{\psi(1)-{q}}{\lambda +{q}-\psi(1)}  \frac{\Phi({q}+\lambda)-1}{1-\Phi({q})} \xrightarrow{\lambda \uparrow \infty}\frac{\psi(1)-q}{1-\Phi(q)}.
\end{align*}
Combining these and \eqref{E_B_B}, we obtain \eqref{eq_limit_lambda_ot}.
\subsection{Proof of Proposition \ref{CS_limit_T_0}}\label{appendix_CS_limit_T_}

Fix $t > 0$. Let us define the event
\begin{align*}
E := \{ N^{\lambda}_{t} = 1 \} = \{ T^{\lambda}_1 \leq t, T^{\lambda}_2 > t \} = \{ T^{\lambda}_1 \leq t, S > t - T^{\lambda}_1 \}
\end{align*}
where $S := T^{\lambda}_2 - T^{\lambda}_1$ 
has the exponential distribution with the parameter $\lambda$.
Note that
\begin{align}
E \cap \{ T_{V_B}^- < t \} = \{ T^{\lambda}_1\leq t, \, V_{T^{\lambda}_1}<V_B, \; S> t-T^{\lambda}_1 \}. \label{E_intersect}
\end{align}

We start from analyzing the numerator of \eqref{eq:spreads}.
	We decompose it as follows:
	\begin{align}
		f(t):= \mathbb{E}\Big[\big[P-(1-\alpha)V_{T_{V_B}^-} \big]e^{-r T_{V_B}^-}\mathbf{1}_{\{T_{V_B}^- \leq t\}} \Big]&= f_1(t) + f_2(t),
		\end{align}
		where
		\begin{align*}
		f_1(t) &:= \mathbb{E}\Big[\big[P-(1-\alpha)V_{T_{V_B}^-} \big]e^{-r T_{V_B}^-}\mathbf{1}_{\{T_{V_B}^- \leq t\}} \,
		\mathbf{1}_E
		\Big], \\
		f_2(t) &:=\mathbb{E}\Big[\big[P-(1-\alpha)V_{T_{V_B}^-} \big]e^{-r T_{V_B}^-}\mathbf{1}_{\{T_{V_B}^- \leq t\}} \,
		\mathbf{1}_{E^c}
		\Big].
	\end{align*}
	Here, by \eqref{E_intersect} and because $S$ is an independent exponential random variable with parameter $\lambda$,
	\begin{align*}
		f_1(t)&=\mathbb{E}\Big[\big[P-(1-\alpha)V_{T^{\lambda}_1} \big]e^{-r T^{\lambda}_1}\mathbf{1}_{\{T^{\lambda}_1\leq t, \, V_{T^{\lambda}_1}<V_B, \; S>t-T^{\lambda}_1\}}\Big]\\
		&=\mathbb{E}\Big[\big[P-(1-\alpha)V_{T^{\lambda}_1} \big]e^{-r T^{\lambda}_1}e^{-\lambda(t-T^{\lambda}_1) }\mathbf{1}_{\{T^{\lambda}_1\leq t, \, V_{T^{\lambda}_1}<V_B\}}\Big]\\
		&=\mathbb{E}\Big[\int_0^{t}\lambda e^{-\lambda s}\big[P-(1-\alpha)V_{s} \big]e^{-r s}e^{-\lambda(t-s) } \mathbf{1}_{\{V_{s}<V_B\}}\diff s \Big]\\
		&=\lambda e^{-\lambda t} \mathbb{E} \Big[\int_0^{t}\big[P-(1-\alpha)V_{s} \big]e^{-r s} \mathbf{1}_{\{V_{s}<V_B\}}\diff s \Big],
	\end{align*}
	and $|f_2(t)| \leq (|P|+(1-\alpha)V_B) \mathbb{P} \{ N_t^\lambda \geq 2\}
		=o(t)$ as $t \downarrow 0$. 

	Summing these,
	\begin{align}
		f(t)
		=\lambda e^{-\lambda t}\mathbb{E}\Big[\int_0^{t}\big[P-(1-\alpha)V_{s} \big]e^{-r s} \mathbf{1}_{\{V_{s}<V_B\}}\diff s \Big]+o(t).\label{fzt}
	\end{align}
	
	On the other hand, we transform the denominator of \eqref{eq:spreads} as follows:
	\begin{align*}
		g(t):= \mathbb{E} \big[1- e^{-r(t\wedge T_{V_B}^-)}\big]=1-e^{-rt}+
		g_1(t) + g_2(t)
	\end{align*}
where
	\begin{align*}
		g_1(t) := \E\big[(e^{-rt}-e^{-rT_{V_B}^-})\mathbf{1}_{\{T_{V_B}^-\leq t\}}
	\mathbf{1}_E
		\big]
		\quad \textrm{and} \quad
		g_2(t) := \E\big[(e^{-rt}-e^{-rT_{V_B}^-})\mathbf{1}_{\{T_{V_B}^-\leq t\}}
		\mathbf{1}_{E^c}
		\big].
	\end{align*}
	Similar to the computation for $f_1(t)$ and $f_2(t)$, by \eqref{E_intersect}, 
	\begin{multline*}
	g_1(t)=
		\E\left[(e^{-rt}-e^{-rT^{\lambda}_1})\mathbf{1}_{\{T^{\lambda}_1\leq t, \, V_{T^{\lambda}_1}<V_B, \,
		S
		> t-T^{\lambda}_1\}}\right]\\
		=\E\left[e^{-\lambda (t-T^{\lambda}_1)}(e^{-rt}-e^{-rT^{\lambda}_1})\mathbf{1}_{\{T^{\lambda}_1\leq t, \, V_{T^{\lambda}_1}<V_B\}}\right]= \lambda e^{-\lambda t}\E\left[\int_0^{t}(e^{-rt}-e^{-rs})\mathbf{1}_{\{V_{s}<V_B\}}\diff s\right],
	\end{multline*}
	and $g_2(t) \leq \mathbb{P} \{ N_t^\lambda \geq 2\}
		= o(t).$ 
	Hence putting all the pieces together we get that
	\begin{align}\label{gzt}
		g(t) 
		=1-e^{-rt}+ \lambda e^{-\lambda t}\E\left[\int_0^{t}(e^{-rt}-e^{-rs})\mathbf{1}_{\{V_{s}<V_B\}}\diff s\right]+o(t).
	\end{align}
	
	Now, from \eqref{fzt}, \eqref{gzt}, and the mean value theorem,	\begin{align*}
		\lim_{t\to 0}\frac{f(t)}{t} &=\lim_{t\to 0}\frac{1}{t}\left(\lambda e^{-\lambda t}\mathbb{E}\left[\int_0^{t}\big[P-(1-\alpha)V_{s} \big]e^{-r s} \mathbf{1}_{\{V_{s}<V_B\}}\diff s \right]+o(t)\right) =\lambda\big[P-(1-\alpha)V\big] \mathbf{1}_{\{V<V_B\}}, \\
		\lim_{t\to 0}\frac{g(t)}{t} &=\lim_{t\to 0}\frac{1}{t}\left(1-e^{-rt}+\lambda e^{-\lambda t}\E\left[\int_0^{t}(e^{-rt}-e^{-rs})\mathbf{1}_{\{V_{s}<V_B\}}\diff s\right]+o(t)\right)
		=r.
	\end{align*}
	By dividing the former by the latter, we have the claim.
	\subsection{Proof of Proposition \ref{conv_CS_lambda}}\label{proof_conv_CS_lambda}
		By \eqref{der_gamma_0_0}  and \eqref{fun_gamma}, we can write, for any  $\theta\geq0$, and $q \geq 0$,
		\begin{align*}
			J^{(q, \lambda)}(y;\theta)=\frac{\lambda}{\lambda+q-\psi(\theta)}\left(H^{(q)}(y;\theta)-\frac{\psi(\theta)-q}{\lambda}\frac{\Phi(q+\lambda)-\Phi(q)}{\theta-\Phi(q)}H^{(q)}(y;\Phi(q+\lambda))\right).
		\end{align*}
For $y\not=0$, we  now note the following:
\begin{enumerate}
\item For the cases (i) $y < 0$ or (ii) $y > 0$ and $X$ does not have a diffusion component,  we have that $H^{(q)}(y;\Phi(q+\lambda)) \xrightarrow{\lambda \uparrow \infty} 0 $ (because the process does not creep downward as in Exercise 7.6 of \cite{K}) and 
$\frac{\Phi(q+\lambda)-\Phi(q)}{\lambda}$ is bounded for $\lambda>0$ cut-off from zero (which can be verified by the convexity of $\psi$).
\item For the case $y > 0$ and $X$ has a diffusion component, $H^{(q)}(y;\Phi(q+\lambda))  \xrightarrow{\lambda \uparrow \infty}  \E_y [e^{-q \tau_0^-} \mathbf{1}_{\{X_{\tau_0^-} = 0\}}]$ and $\frac{\Phi(q+\lambda)-\Phi(q)}{\lambda} \xrightarrow{\lambda \uparrow \infty} 0$ (because $\psi(\theta) \sim \frac 1 2 \sigma^2 \theta^2$ as $\theta \rightarrow \infty$ where $\sigma$ is the diffusion coefficient of $X$).
\end{enumerate}
Hence, the previous arguments imply that, for $y \neq 0$,
		\begin{align*}
			\lim_{\lambda\to\infty} J^{(q, \lambda)}(y;\theta)=H^{(q)}(y;\theta),\qquad\text{$\theta\geq0$.}
		\end{align*}
		Given that $J^{(q, \lambda)}(y;\theta)$ is the Laplace transform of the random vector $(\tilde{T}_0^- (\lambda)
		,X_{\tilde{T}_0^-(\lambda)})$ (where we put $(\lambda)$ to spell out the dependency on $\lambda$),  by L\'evy's Continuity Theorem we have that
		$(\tilde{T}_0^- (\lambda)
		,X_{\tilde{T}_0^-}(\lambda))$ converges in distribution to $(\tau_0^-,X_{\tau_0^-})$.
		Hence, using Skorohod's Representation Theorem (see Theorem 6.7 in \cite{Bill}) as well as dominated convergence, we obtain, for $V \neq V_B$,
		\begin{align*}
			\lim_{\lambda\to\infty}\mathrm{CS}_{\lambda}(t)&=\frac{r}{P}\lim_{\lambda\to\infty}\frac{\mathbb{E}_{\log (V/V_B)}\Big[\big[P-(1-\alpha)V_B e^{X_{\tilde{T}_0^-(\lambda)}}\big]e^{-r \tilde{T}_0^-(\lambda)}\mathbf{1}_{\{\tilde{T}_0^-(\lambda)\leq t\}} \Big]}{\mathbb{E}_{\log (V/V_B)}\big[1- e^{-r(t\wedge \tilde{T}_0^-(\lambda))}\big]}\\
			&=\frac{r}{P}\frac{\mathbb{E}_{\log (V/V_B)}\Big[\big[P-(1-\alpha)V_B e^{X_{\tau_0^-}}\big]e^{-r \tau_0^-}\mathbf{1}_{\{\tau_0^-\leq t\}} \Big]}{\mathbb{E}_{\log (V/V_B)}\big[1- e^{-r(t\wedge \tau_0^-)}\big]}=CS(t).
		\end{align*}

\section{Proof of Theorem \ref{prop_resolvent}} \label{proof_prop_resolvent}

From 
Theorem 2.7 of \cite{KKR}
 for any Borel set $A$ on $[0, \infty)$, on $\R$, and $(-\infty,0]$ respectively,
\begin{align} \label{resolvent_density}
	\E_x \Big[ \int_0^{\tau_{0}^- } e^{-qt} \mathbf{1}_{\left\{ X_t \in A \right\}} \diff t\Big] &= \int_A \Big[ e^{-\Phi(q) y}W^{(q)}(x) -{W^{(q)}} (x-y) \Big] \diff y, \quad x \geq 0, \\
	\label{resolvent_density_2}
	\E_x\left[ \int_0^{\infty}e^{-(q+\lambda)t } \mathbf{1}_{\{X_t \in A\}}\diff t\right] &=\int_A \left[\frac{e^{\Phi(q+\lambda)(x-y)}}{\psi'(\Phi(q+\lambda))
	}-W^{(q+\lambda)}(x-y)\right]\diff y, \quad x \in \R,\\
	\E_x \Big[ \int_0^{\tau_{0}^+ } e^{-(q+\lambda)t} \mathbf{1}_{\left\{ X_t \in A \right\}} \diff t\Big] &= \int_A \left(e^{\Phi(q+\lambda)x}W^{(q+\lambda)}(-y)-W^{(q+\lambda)}(x-y)\right) \diff y, \quad x \leq 0\label{killed_resolvent}
\end{align}
where $\tau_0^-$ is defined in \eqref{tauzero} and $\tau_0^+:=\inf\{t\geq 0: X_t>0\}$.

We will prove the result for $z=0$ and compute
\begin{align*}
g(x) :=\E_x\Big[ \int_0^{\tilde{T}_0^- } e^{-qt} h(X_t) \diff t \Big], \quad x \in \mathbb{R}.
\end{align*}
The general case follows because the spatial homogeneity of the L\'evy process implies that
$
\E_x\big[ \int_0^{\tilde{T}_z^- } e^{-qt} h(X_t) \diff t \big]
=\E_{x-z}\big[ \int_0^{\tilde{T}_0^-} e^{-qt} h (X_t+z) \diff t \big]$ for $x,z\in\R$.

For $x \in \R$,  by the strong Markov property, 
\begin{align}\label{MPa}
	g(x)&=\E_x\Big[ \int_0^{\tau_0^-  } e^{-qt} h(X_t) \diff t \Big] + \E_x\left[ e^{-q\tau_0^-}g(X_{\tau^-_0})  \mathbf{1}_{\{   \tau^-_0<\infty\}} \right].
\end{align}
In particular, for $x < 0$, again by the strong Markov property,
\begin{align*}
	g(x) = A(x) g(0)+B(x),
\end{align*}
where, for $x \leq 0$,
\begin{align*}
 A(x):=&\E_x \Big[ e^{-q \tau_0^+} \mathbf{1}_{\{   \tau^+_0< T_1^\lambda\}} \Big] =\E_x \Big[ e^{-(q+\lambda) \tau_0^+} \Big] = e^{\Phi(q+\lambda)x},  \nonumber \\
	\begin{split}
B(x):=&\E_x \Big[ \int_0^{\tau_0^+} e^{-qt} \mathbf{1}_{\{ t < T_1^\lambda \}} h(X_t)     \diff t\Big]
		=\int_{-\infty}^0 h(y)
		\left(e^{\Phi(q+\lambda)x}W^{(q+\lambda)}(-y)-W^{(q+\lambda)}(x-y)\right)
		\diff y.
	\end{split}
\end{align*}
Here, the first equality of the former holds by the fact that $T_1^\lambda$ is an independent exponential random variable with parameter $\lambda$ and Theorem 3.12 of \cite{K}. The second equality of the latter is a consequence of \eqref{killed_resolvent}.

Now, by \eqref{der_gamma_0_0},
\begin{align*}
	\E_{x}&\left[e^{-q\tau_0^-}A(X_{\tau_0^-})\mathbf{1}_{\{\tau_0^-<\infty\}}\right]=H^{(q)}(x;\Phi(q+\lambda)), \quad x \in \R,
\end{align*}
for function $H^{(q)}$ defined in \eqref{der_gamma_0_0}.
In addition, by the proof of Theorem 4.1 in \cite{BPY}, we have that
\begin{align*}
	\E_{x}\left[e^{-q\tau_0^-}B(X_{\tau_0^-})\mathbf{1}_{\{\tau_0^-<\infty\}}\right]
	=W^{(q)}(x)\int_{-\infty}^0 h(y)  H^{(q+\lambda)}(-y; \Phi(q))
	\diff y -\int_{-\infty}^0 h(y) I^{(q, \lambda)}(x,-y) \diff y.
\end{align*}
Substituting these in \eqref{MPa} and, then applying  \eqref{resolvent_density} and Remark 4.3 in \cite{BPY}, we obtain, for all $x \in \mathbb{R}$,
\begin{equation}\label{MPa2}
	g(x)=
	g(0)H^{(q)}(x;\Phi(q+\lambda))
	+W^{(q)}(x)\int_{-\infty}^\infty h(y)  H^{(q+\lambda)}(-y; \Phi(q)) \diff y -\int_{-\infty}^\infty h(y) I^{(q, \lambda)}(x,-y) \diff y.
\end{equation}

On the other hand, by the strong Markov property, we can also write 
\begin{align} \label{g_zero_recursion}
\begin{split}
	g(0) = \E_0 \Bigg[ \E_0\Bigg[ \int_0^{T_1^\lambda } e^{-qt} h(X_t) \diff t + \mathbf{1}_{\{X_{T_1^\lambda} > 0 \}}\int_{T_1^\lambda}^{\tilde{T}_0^- } e^{-qt} h(X_t) \diff t  \Big| T_1^\lambda, (X_u)_{0 \leq u \leq T_1^\lambda }\Bigg] \Bigg]
	= \gamma_1 +  \gamma_2, \\\quad \textrm{where} \;
	\gamma_1 := \E_0 \Bigg[\int_0^{T_1^\lambda} e^{-qt} h(X_t) \diff t \Bigg] \quad\text{and}\quad \gamma_2 :=  \E_0 \Bigg[e^{-q T_1^\lambda} g(X_{T_1^\lambda}) \mathbf{1}_{\{ X_{T_1^\lambda} > 0\}}   \Bigg].
	\end{split}
\end{align}
We will compute $\gamma_1$ and $\gamma_2$ below.
First,  observe that
\begin{equation*}
\gamma_1 = \E_0 \left[\int_0^\infty \mathbf{1}_{\{t <T_1^\lambda \}}e^{-qt} h(X_t) \diff t \right] = \E_0 \left[\int_0^{\infty} e^{-(q+\lambda) t} h(X_t) \diff t \right].
\end{equation*}
For $\gamma_2$, by \eqref{resolvent_density_2}, we can write
\begin{align}
	\gamma_2 =\lambda \E_0 \Big[\int_0^{\infty}e^{-(q+\lambda)s}g(X_s) \mathbf{1}_{\{ X_s > 0\}}\diff s\Big]=\frac{\lambda}{\psi'(\Phi(q+\lambda))}\int_0^{\infty}
	e^{-\Phi(q+\lambda)y}g(y)\diff y, \label{gamma_3_rewrite}
\end{align}
which we shall compute using the expression of $g$ as in \eqref{MPa2}. First, by identity (A.8) in \cite{BPY} we have
\begin{align} \label{Z_Phi_integral}
	\begin{split}
		\int_0^{\infty}e^{-\Phi(q+\lambda)y}Z^{(q)}(y;\Phi(q+\lambda))\diff y=\frac{\psi'(\Phi(q+\lambda))}{\lambda},
	\end{split}
\end{align}
while \eqref{scale_function_laplace} gives $\int_0^\infty e^{-\Phi(q+\lambda) y} W^{(q)}(y) \diff y = \lambda^{-1}$, and hence
\begin{align*}
\int_0^{\infty}e^{-\Phi(q+\lambda)y}H^{(q)}(y;\Phi(q+\lambda))\diff y=\frac{\psi'(\Phi(q+\lambda))}{\lambda}-\frac{1}{\Phi(q+\lambda)-\Phi(q)}.
\end{align*}
Again by the proof of Theorem 4.1 in \cite{BPY}, we have
\begin{align*}
	\int_0^{\infty}e^{-\Phi(q+\lambda)y}\int_{-\infty}^{\infty}h(z) I^{(q, \lambda)}(y,-z) \diff z \diff y
	=\frac{\psi'(\Phi(q+\lambda))}{\lambda} \E_0 \left[\int_0^{\infty}e^{-(q+\lambda)t}h(X_t) \diff t \right].
\end{align*}
Substituting these  in \eqref{gamma_3_rewrite} and with the help of \eqref{MPa2}, 
\begin{align*}
	\gamma_2
	&=g(0)\frac{\lambda}{\psi'(\Phi(q+\lambda))}\Big[\frac{\psi'(\Phi(q+\lambda))}{\lambda}-\frac{1}{\Phi(q+\lambda)-\Phi(q)}\Big]\\
	&\qquad- \E_0 \left[\int_0^{\infty}e^{-(q+\lambda)t}h(X_t)\diff t \right]+\frac{1}{\psi'(\Phi(q+\lambda))} \int_{-\infty}^\infty h(y)  H^{(q+\lambda)}(-y; \Phi(q)) \diff y.
\end{align*}
Now substituting the computed values of $\gamma_1$, and $\gamma_2$  in \eqref{g_zero_recursion}
we obtain
\begin{align*}
	g(0)
	&=g(0)- \frac {\lambda} {\Phi(q+\lambda)- \Phi(q)}\frac{g(0)}{\psi'(\Phi(q+\lambda))}+\frac{1}{\psi'(\Phi(q+\lambda))}\int_{-\infty}^\infty h(y)  H^{(q+\lambda)}(-y; \Phi(q)) \diff y,
\end{align*}
and hence, solving for $g(0)$ we obtain
\begin{align*}
	g(0)=
	 \frac  {\Phi(q+\lambda)- \Phi(q)} {\lambda}
	\int_{-\infty}^\infty h(y) H^{(q+\lambda)}(-y; \Phi(q))
	\diff y.
\end{align*}
Substituting this back in \eqref{MPa2}, we have
\begin{align*}
g(x)=Z^{(q)}(x;\Phi(q+\lambda))\frac  {\Phi(q+\lambda)- \Phi(q)} {\lambda}
\int_{-\infty}^\infty h(y) H^{(q+\lambda)}(-y; \Phi(q))\diff y-\int_{-\infty}^\infty h(y)  I^{(q, \lambda)}(x,-y) \diff y.
\end{align*}
Hence the resolvent density is given by \eqref{resol_dens}, as desired. 
 \qed
		

%
%



\begin{thebibliography}{21}


\bibitem{Albrecher} Albrecher, H., Ivanovs, J., Zhou, X.: Exit identities for L\'evy processes observed at Poisson arrival times.
Bernoulli \textbf{22},  1364--1382 (2016)

\bibitem{Albrecher_Ivanovs} Albrecher, H., Ivanovs, J.: Strikingly simple identities relating exit problems for L\'evy processes under continuous and Poisson observations.
Stoch. Process. Appl. \textbf{127(2)}, 643--656 (2017)



\bibitem{Alili_Kyprianou} Alili, L., Kyprianou, A.E.: Some remarks on first passage of L\'evy processes, the American put and pasting principles.
Ann. Appl. Probab. \textbf{15(3)}, 2062--2080 (2005) 

\bibitem{Asmussen} Asmussen, S., Avram, F., Pistorius, M.R.: Russian and American put options under exponential phase-type L\'evy models.
Stoch. Process. Appl. \textbf{109(1)},  79--111 (2004)


\bibitem{Avanzi_Tu_Wong} Avanzi, B., Tu, V., Wong, B.: On optimal periodic dividend strategies in the dual model with diffusion.
Insurance Math. Econom. \textbf{55}, 210--224 (2014) 

\bibitem{Avanzi_Cheung_Wong_Woo} Avanzi, B., Cheung, E. C., Wong, B., Woo, J.K.: On a periodic dividend barrier strategy in the dual model with continuous monitoring of solvency.
Insurance Math. Econom. \textbf{52}, 98--113 (2013) 


\bibitem{Avram_Kyprianou_Pistorius} Avram, F., Kyprianou, A.E., Pistorius, M.R.: Exit problems for spectrally negative L\'evy processes and applications to (Canadized) Russian options.
Ann. Appl. Probab. \textbf{14(1)}, 215--238 (2004) 

\bibitem{APP2007}
Avram, F., Palmowski, Z., Pistorius, M.: On the optimal dividend problem for a spectrally negative {L}\'evy process.
Ann. Appl. Probab. \textbf{17(1)}, 156--180 (2007)




\bibitem{Avram_Perez_Yamazaki} Avram, F., P\'erez, J.L., Yamazaki, K.: Spectrally negative L\'evy processes with Parisian reflection below and classical reflection above.
Stoch. Process. Appl. \textbf{128(1)}, 255--290 (2018) 


\bibitem{BPPR} Baurdoux, E.J., Pardo, J.C., P\'erez, J.L., Renaud J.F.: Gerber--Shiu distribution at Parisian ruin for L\'evy insurance risk processes. J. Appl. Probab.  \textbf{53}, 572--584, (2016)

\bibitem{Bertoin96a} Bertoin, J.: L\'evy Processes. Cambridge University Press, (1996)


\bibitem{Bill} Billingsley, P.: Convergence of Probability Measures. John Wiley \& Sons, Inc., (1999)

\bibitem{Bielecki} Bielecki, T.R., Rutkowski, M: Credit risk: modeling, valuation and hedging. Springer Science \& Business Media, (2013)



\bibitem{Black} Black, F., Cox, J.: Valuing corporate securities: Some effects of bond indenture  provisions.
J. Finance \textbf{31}, 351--367 (1976)

\bibitem{Brealey_Myers_2001}
Brealey, R.A., Myers, S.C.: Principles of Corporate Finance. McGraw-Hill, New York, (2001)


\bibitem{Brennan} Brennan, M., Schwartz, E.: Corporate income taxes, valuation, and the problem of optimal capital structure.
J. Business \textbf{51}, 103--114 (1978)

\bibitem{Broadie} Broadie, M., Chernov, M., Sundaresan, S.: Optimal debt and equity values in the presence of Chapter 7 and Chapter 11.
J. Finance, \textbf{62(3)}, 1341--1377 (2007)

\bibitem{Carr} Carr, P.: Randomization and the American put. Rev. Financ. Stud., \textbf{11(3)}, 597--626, (1998)

\bibitem{CGMY} Carr, P., Geman, H., Madan, D. B.,  Yor, M.: The fine structure of asset returns: An empirical investigation. J. Business, \textbf{75(2)}, 305-332, (2002)

\bibitem{Chen_Kou_2009}
Chen, N., Kou, S.G.: Credit spreads, optimal capital structure, and implied
              volatility with endogenous default and jump risk.
Math. Finance,
\textbf{19(3)}, 343--378 (2009)



\bibitem{Chesney} Chesney, M., Jeanblanc-Picqu\'e, M., Yor, M.: Brownian excursions and Parisian barrier options.
Adv. Appl. Probab. \textbf{29(1)}, 165--184 (1997)

\bibitem{ContTankov}
Cont, R., Tankov, P.: Financial Modelling with Jump Processes. Chapman \& Hall, (2003)

\bibitem{Duffie} Duffie, D., Lando, D.: Term structure of credit spreads with incomplete accounting information.
Econometrica \textbf{69}, 633--664 (2001)

\bibitem{Dupuis_Wang} Dupuis, P. Wang, H.: Optimal stopping with random intervention times.
Adv. Appl. Probab. \textbf{34(1)}, 141--157 (2002)

\bibitem{Egami_Yamazaki_2010_2}
Egami, M., Yamazaki, K.:
Phase-type fitting of scale functions for spectrally negative L\'evy processes.
J. Comput. Appl. Math. \textbf{264}, 1--22 (2014)

\bibitem{Emery} Emery, D.J.:  Exit problem for a spectrally positive L\'evy process.
Adv. Appl. Probab. \textbf{5}, 498--520 (1973)


\bibitem{Francois_Morellec} Francois, P., Morellec, E.: Capital structure and asset prices: Some effects of bankruptcy procedures.
J. Business \textbf{77(2)}, 387--411 (2004)

\bibitem{Frank} Frank, M.Z., Goyal, V.K.: Trade-off and pecking order theories of debt. 135--202, 
Handbook of empirical corporate finance. Elsevier (2008)

\bibitem{Hilberink} Hilberink, B., Rogers, L.C.G.: Optimal capital structure and endogenous default.
Finance Stoch. \textbf{6}, 237--263 (2002)

\bibitem{Ivphd}
Ivanovs, J.: One-sided Markov additive processes and related exit problems,
PhD dissertation, University of Amsterdam, (2011)

\bibitem{Ju} Ju, N., Parrino, R., Poteshman, A.M., Weisbach, M.S.: Horses and rabbits? Trade-off theory and optimal capital structure.
J. Financ. Quant. Anal.  \textbf{40(2)}, 259--281 (2005)


\bibitem{Kou_Wang}
Kou, S.G., Wang, H.: First passage times of a jump diffusion process.
Adv. Appl. Probab.
\textbf{35(2)}, 504--531 (2003)



\bibitem{Kraus}
Kraus, A., Litzenberger R.H. : A state‐preference model of optimal financial leverage.
J. Finance
\textbf{28(4)}, 911-922 (1973)

\bibitem{Kuznetsov2} Kuznetsov, A., Kyprianou, A. E., Pardo, J. C.: Meromorphic L\'evy processes and their fluctuation identities. Ann. Appl. Probab.  \textbf{22(3)}, 1101-1135,  (2012)


\bibitem{KKR}  Kuznetsov, A., Kyprianou, A.E., Rivero, V.: \rm The
theory of scale functions for spectrally negative L\'evy processes. {\it L\'evy Matters II, Springer Lecture Notes in Mathematics}, (2013)

\bibitem{Kuznetsov} Kuznetsov, A.:
On the Convergence of the Gaver--Stehfest Algorithm.
SIAM J. Numer. Anal. \textbf{51(6)}, 2984--2998 (2013)




\bibitem{K}
Kyprianou, A.E.: Fluctuations of {L}\'evy Processes with Applications.
Springer, Heidelberg (2014)

\bibitem{Kyprianou} Kyprianou, A.E., Surya, B.A.:
Principles of smooth and continuous fit in the determination of endogenous bankruptcy levels.
Finance Stoch. \textbf{11}, 131--152 (2007)

\bibitem{Leland94} Leland, H.E.: Corporate debt value, bond covenants, and optimal capital structure.
J. Finance \textbf{49}, 1213--1252 (1994)

\bibitem{Leland96} Leland, H.E., Toft, K.B.: Optimal capital structure, endogenous
bankruptcy, and the term structure of credit spreads. J. Finance \textbf{51}, 987--1019 (1996)

\bibitem{Leung}Leung, T., Yamazaki, K.,  Zhang, H.: An analytic recursive method for optimal multiple stopping: Canadization and phase-type fitting. Int. J. Theor. Appl. Finance, \textbf{18(05)}, 1550032, (2015)

\bibitem{Loeffen} Loeffen, R. L.:  An optimal dividends problem with a terminal value for spectrally negative L\'evy processes with a completely monotone jump density. J. Appl. Probab.  \textbf{46(1)}, 85-98, (2009)

\bibitem{LRZ}
Loeffen, R., Renaud, J.F., Zhou, X.:
Occupation times of intervals until first passage times for spectrally negative L\'evy processes with applications.
Stoch. Process. Appl.  \textbf{124(3)}, 1408--1435 (2014)

\bibitem{Long}
Long, M.,  Zhang, H.:
On the optimality of threshold type strategies in single and recursive optimal stopping under L\'evy models.
Stoch. Process. Appl. Available online (2018)


\bibitem{Madan} Madan, D. B.,  Seneta, E.:  The variance gamma (VG) model for share market returns. J. Business \textbf{63(4)}, 511-524, (1990)

\bibitem{Merton} Merton, R.C.: On the pricing of corporate debt: the risk structure of interest rate.
J. Finan. Econom. \textbf{29}, 449--470 (1974)

\bibitem{Modigliani} Modigliani, F., Miller, M.: The cost of capital, corporation finance and the theory of investment.
American Economic Review \textbf{48}, 267--297 (1958)

\bibitem{Moraux} Moraux, F.: Valuing corporate liabilities when the default threshold is not an absorbing barrier.
EFMA 2002 London Meetings. Available at SSRN: https://ssrn.com/abstract=314404 or http://dx.doi.org/10.2139/ssrn.314404 (2002)


\bibitem{Noba_Perez_Yamazaki_Yano} Noba, K., P\'erez, J.L., Yamazaki, K., Yano, K.: On optimal periodic dividend strategies for L\'evy risk processes.
Insurance Math. Economics. \textbf{80}, 29--44 (2018) 


\bibitem{Pardo_Perez_Rivero} Pardo, J.C., P\'erez, J.L., Rivero, V.M.: The excursion measure away from zero for spectrally negative L\'evy processes.
Ann. Inst. H. Poincar\'e Probab. Statist. \textbf{54(1)}, 75--99 (2018) 

\bibitem{Perez_Yamazaki_options} P\'erez, J.L., Yamazaki, K.: American options under periodic exercise opportunities.
Stat. Probab. Lett. \textbf{135}, 92--101 (2018)

\bibitem{BPY}
P\'{e}rez, J.L., Yamazaki, K., Bensoussan, A.:
Optimal periodic replenishment policies for spectrally positive L\'{e}vy processes. See arXiv:1806.09216 (2018)


\bibitem{Rodosthenous} Rodosthenous, N., Zhang, H.: Beating the Omega clock: an optimal stopping problem with random time-horizon under spectrally negative L\'evy models.
Ann. Appl. Probab. \textbf{28(4)}, 2105-2140 (2018)



\bibitem{Sims} Sims, C.A.: Implications of rational inattention.
Journal of Monetary Economics \textbf{50(3)}, 665--690 (2003)


\bibitem{Surya} Surya, B.A.: Optimal stopping of L\'evy processes and pasting principles. PhD dissertation, University of Utrecht (2007)


\bibitem{Surya_Yamazaki} Surya, B.A., Yamazaki, K.: Optimal capital structure with scale effects under spectrally negative L\'evy models.
Inter. J. Theor. Appl. Finance \textbf{17(2)}, 1450013 (2014)

\bibitem{Ta} Tak\'acs, L.: Combinatorial Methods in the Theory of Stochastic Processes. Wiley, New York, (1966)


\bibitem{Z} Zolotarev, V.M.: The first passage time of a level and the behavior at infinity for a class of processes with independent increments. Theory Probab. Appl. \textbf{9}, 653--664 (1964)


\end{thebibliography}
\end{document}